\documentclass[a4paper,UKenglish,cleveref]{lipics-v2021}

\usepackage{listings}
\usepackage{wrapfig}
\usepackage{stmaryrd}
\usepackage{tikz}
\title{Multi-graded Featherweight Java} 

\author{Riccardo Bianchini}{DIBRIS, Universit\`a di Genova, Italy}{riccardo.bianchini@edu.unige.it}{https://orcid.org/0000-0003-0491-7652}{}

\author{Francesco Dagnino}{DIBRIS, Universit\`a di Genova, Italy}{francesco.dagnino@dibris.unige.it}{https://orcid.org/0000-0003-3599-3535}{}

\author{Paola Giannini}{DiSSTE, Universit\`a del Piemonte Orientale, Italy}{paola.giannini@uniupo.it}{https://orcid.org/0000-0003-2239-9529}{}

\author{Elena Zucca}{DIBRIS, Universit\`a di Genova, Italy}{elena.zucca@unige.it}{https://orcid.org/0000-0002-6833-6470}{}

\authorrunning{R. Bianchini, F. Dagnino, P. Giannini, and E. Zucca} 

\Copyright{Riccardo Bianchini, Francesco Dagnino, Paola Giannini, and Eelena Zucca} 


\acknowledgements{This work was partially funded by the MUR project ``T-LADIES'' (PRIN 2020TL3X8X) and has the financial support of the Universit\`{a}  del Piemonte Orientale. We thank the anonymous referees for their useful suggestions.}

\nolinenumbers 

\definecolor{violet}{rgb}{0.54, 0.17, 0.89}
\newif\ifsubmit
\submittrue
\ifsubmit
\newcommand{\EZComm}[1]{} 
\newcommand{\FDComm}[1]{} 
  \def\bez{} \def\eez{} 
  \def\bfd{} \def\efd{} 
\newcommand{\RBComm}[1]{} 
  \def\brb{} \def\erb{} 
\newcommand{\PGComm}[1]{} 
  \def\bpa{} \def\epa{} 
\else
\newcommand{\EZComm}[1]{{\scriptsize \textcolor{blue}{[Elena{:} #1]}}}
\newcommand{\FDComm}[1]{{\scriptsize \textcolor{magenta}{[Francesco{:} #1]}}}
  \def\bez{\begin{color}{blue}} 
  \def\eez{\end{color}} 
\newcommand{\RBComm}[1]{{\scriptsize \textcolor{red}{[Riccardo{:} #1]}}}
  \def\bfd{\begin{color}{magenta}} 
  \def\efd{\end{color}} 
  \def\brb{\begin{color}{red}} 
  \def\erb{\end{color}} 
\newcommand{\PGComm}[1]{{\scriptsize \textcolor{violet}{[Paola{:} #1]}}}
  \def\bpa{\begin{color}{violet}} 
  \def\epa{\end{color}} 
\fi

\newcommand{\refToRule}[1]{\textsc{\small (#1)}}
\newcommand{\refItem}[2]{\cref{#1}(\ref{#1:#2})}
\newcommand{\Space}{\hskip 0.7em}
\newcommand{\BigSpace}{\hskip 1.5em}
\newcommand{\HugeSpace}{\hskip 10em}
\newcommand{\meta}[1]{\colorbox{lightgray}{$#1$}}

\newcommand{\Tuple}[1]{\ple{#1}} 
\newcommand{\Pair}[2]{\Tuple{#1,#2}}
\newcommand{\Subst}[3]   {#1[#2/#3]}
\newcommand{\SubstDots}[5]   {#1[#2/#3\ldots#4/#5]}

\newcommand{\AddToMap}[3]   {#1,#2\mapsto#3}
\newcommand{\AddToMapDots}[5]   {#1,#2\mapsto#3,\ldots,#4\mapsto#5}
\newcommand{\dom}[1]{\aux{dom}(#1)}
\newcommand{\NamedRule}[4]{\scriptstyle{\textsc{(#1)}}\Space
\displaystyle                  
\frac{#2}{#3}         
\begin{array}{l}
#4     
\end{array}
}

\newenvironment{grammatica}{$\begin{array}{lcll}}{\end{array}$}
\newcommand{\produzione}[3]{#1&::=&#2&\mbox{#3}}
\newcommand{\seguitoproduzione}[2]{&&#1&\mbox{#2}}
\newcommand{\terminale}[1]{\texttt{#1}}
\newcommand{\aux}[1]{\textsf{#1}}
\newcommand{\metavariable}[1]{\mathit{#1}}
\newcommand{\produzioneinline}[2]{#1::=#2}

\newcommand{\LambdaExpWithType}[3]{\lambda #1{:}#2.#3}
\newcommand{\LambdaExpWithTypeAndCoeffect}[4]{\lambda #1{:}#2[#3].#4}
\newcommand{\INT}{\aux{int}}

\newcommand{\FJ}{\textsc{FJ}}

\newcommand{\C}{\metavariable{C}}
\newcommand{\D}{\metavariable{D}}
\newcommand{\T}{\metavariable{T}}
\newcommand{\m}{\metavariable{m}}
\newcommand{\f}{\metavariable{f}} 
\newcommand{\x}{\metavariable{x}}
\newcommand{\y}{\metavariable{y}} 
\newcommand{\z}{\metavariable{z}}
\newcommand{\e}{\metavariable{e}}
\newcommand{\es}{\metavariable{es}}
\newcommand{\this}{\terminale{this}}

\newcommand{\ConstrCall}[2]{\terminale{new}\; #1\terminale{(}#2\terminale{)}}
\newcommand{\ConstrCallTuple}[3]{{\ConstrCall{#1}{#2_1,  \ldots, #2_{#3}}}}
\newcommand{\ConstrCallATuple}[4]{{\ConstrCall{#1}{\annotated{#2_1}{#3_1},  \ldots, \annotated{#2_{#4}}{#3_{#4}}}}}
\newcommand{\FieldAccess}[2]{#1\terminale{.}#2}
\newcommand{\Block}[4]{\{ #1 \, #2 =#3 \terminale{;}\, #4 \}}
\newcommand{\MethCall}[3]{#1{\terminale{.}}#2\terminale{(}#3\terminale{)}}
\newcommand{\MethCallTuple}[4]{\MethCall{#1}{#2}{#3_1,  \ldots, #3_{#4}}}
\newcommand{\MethCallATuple}[5]{\MethCall{#1}{#2}{\annotated{#3_1}{#4_1},  \ldots, \annotated{#3_{#5}}{#4_{#5}}}}
\newcommand{\Field}[2]{#1\ #2\terminale{;}}

\newcommand{\val}{\metavariable{v}}
\newcommand{\vs}{\metavariable{vs}}

\newcommand{\fields}[1]{\aux{fields}(#1)}
\newcommand{\mbody}[2]{{\aux{mbody}(#1,#2)}}
\newcommand{\ev}{\rightarrow}
\newcommand{\evstar}{\ev^\star}
\newcommand{\evdiv}{\ev^\omega}
\newcommand{\reduce}[4]{\ExpStack{#1}{#2}\ev\ExpStack{#3}{#4}}

\newcommand{\EmptyCtx}{\emptyset}
\newcommand{\ROrdCtx}{\rord} 
\newcommand{\SumCtx}{+}
\newcommand{\MulCtx}{\cdot}

\newcommand{\typeGrad}[2]{#1^{#2}}
\newcommand{\rGrad}{\metavariable{r}}
\newcommand{\sGrad}{\metavariable{s}}
\newcommand{\tGrad}{\metavariable{t}}

\newcommand{\IsWFExp}[4][]{#2\vdash_{#1}#3:#4}
\newcommand{\IsWFExpA}[3]{\IsWFExp[a]{#1}{#2}{#3}} 
\newcommand{\IsWFExpI}[4]{\IsWFExp{#1}{#2}{#3} \rightsquigarrow  #4}
\newcommand{\IsNotWFExp}[4][]{#2\not\vdash_{#1}#3:#4}

\newcommand{\IsWFEnvA}[2]{ #1 \vdash_a #2 }
\newcommand{\IsWFEnvI}[3]{#1 \vdash #2 \rightsquigarrow #3} 
 
\newcommand{\IsWFConfA}[4]{ \IsWFExp[a]{#1}{\ExpStack{#2}{#3}}{#4} } 
\newcommand{\IsWFConfI}[6]{ \IsWFExp{#1}{\ExpStack{#2}{#3}}{#4} \rightsquigarrow \ExpStack{#5}{#6}} 
\newcommand{\mtype}[2]{{\aux{mtype}(#1,#2)}}
\newcommand{\VarTypeCoeffect}[3]{#1:_{#3}#2}
\newcommand{\VarCoeffect}[2]{#1:#2}
\newcommand{\VarType}[2]{#1:#2}

\newcommand{\funType}[2]{#1 \rightarrow #2}
\newcommand{\Deriv}{{\cal D}}

\newcommand{\stack}{\rho}
 
\newcommand{\ExpStack}[2]{#1{\mid}#2}
\newcommand{\reducestack}[5]{\ExpStack{#1}{#2}\ev_{#3}\ExpStack{#4}{#5}}
\newcommand{\reducestackstar}[5]{\ExpStack{#1}{#2}\evstar_{#3}\ExpStack{#4}{#5}}
\newcommand{\reducestackdiv}[3]{\ExpStack{#1}{#2}\evdiv_{#3} }

\newcommand{\emptystack}{\emptyset}
\newcommand{\hpEl}[3]{#1\mapsto\Pair{#2}{#3}}
\newcommand{\annotated}[2]{{[#1]}_{#2}}

\newcommand{\Att}{\texttt{A}}
\newcommand{\Btt}{\Pairtt}

\newcommand{\pxtt}{\texttt{p\_}}
\newcommand{\Pairtt}{\texttt{Pair}}
\newcommand{\ptt}{\texttt{p}}
\newcommand{\att}{\texttt{a}}
\newcommand{\first}{\texttt{first}}
\newcommand{\second}{\texttt{second}}
\newcommand{\pairaa}{\val_\Pairtt}

\newcommand{\fO}{\first}
\newcommand{\fT}{\second}
\newcommand{\fX}{\texttt{f}}

\newcommand{\private}{\aux{private}}
\newcommand{\public}{\aux{public}}
\newcommand{\mutable}{\aux{mutable}}
\newcommand{\readonly}{\aux{readonly}}
\newcommand{\privateS}{\aux{priv}}
\newcommand{\publicS}{\aux{pub}}
\newcommand{\xtt}{\texttt{x}}
\newcommand{\ytt}{\texttt{y}}

\newcommand{\erase}[1]{\lceil#1\rceil} 

\newcommand{\NKind}{\aux{N}}
\newcommand{\TKind}{\aux{T}}
\newcommand{\AKind}{\aux{A}}
\newcommand{\PKind}{\aux{P}}
\newcommand{\PPKind}{\aux{PP}}
\newcommand{\APKind}{\aux{AP}}

\newcommand{\kind}{\metavariable{k}}
\newcommand{\KSet}{\mathcal{K}}
\newcommand{\privA}{\aux{a}}
\newcommand{\privB}{\aux{b}}
\newcommand{\privC}{\aux{c}}
\newcommand{\privD}{\aux{d}}

\newcommand{\refines}{\sqsubset^1}

\newcommand{\annotation}[1]{\mbox{[#1]}}

\newcommand{\PredefMethCall}[3]{\MethCall{#1}{\texttt{#2}}{#3}}

\newcommand{\SumMethod}[2]{\PredefMethCall{#1}{sum}{#2}}

\usepackage{semirings} 

\begin{document}

\maketitle


\begin{abstract}
\bez 
\emph{Resource-aware} type systems statically approximate not only the expected result type of a program, but also the way external resources are used, e.g., how many times the value of a variable is needed.
We extend the type system of Featherweight Java to be resource-aware, parametrically on an arbitrary \emph{grade algebra} modeling a specific usage of resources. 
We prove that \bpa this \epa type system is \emph{sound} with respect to a resource-aware version of reduction,  that is, a well-typed program has a reduction sequence which does not get stuck due to resource consumption. Moreover, we show that the available grades can be \emph{heterogeneous}, that is, obtained by combining
 grades of different kinds, 
 \bpa via \epa 
 a minimal collection of homomorphisms  from one kind to another.  Finally, we show how grade algebras and homomorphisms can be specified as Java classes, so that grade annotations in types can be written in the \mbox{language itself.}
 \eez
\end{abstract}

%
%


\section{Introduction}
\label{sect:intro} 

Recently, a considerable amount of research  \cite{PetricekOM14,BrunelGMZ14,Atkey18,GaboardiKOBU16,GhicaS14,OrchardLE19,ChoudhuryEEW21,DalLagoG22}  has been devoted to type systems allowing  reasoning  about resource usage. 
In \emph{(type-and-)coeffect systems}, the typing judgment takes  the  shape $\IsWFExp{\VarTypeCoeffect{\x_1}{\T_1}{\rGrad_1},\ldots,\VarTypeCoeffect{\x_n}{\T_n}{\rGrad_n}}{\e}{\T}$, where the \emph{coeffect} (\emph{grade})  $\rGrad_i$ models how variable $\x_i$ is used in $\e$. For instance, coeffects of shape $\produzioneinline{\rGrad}{0\mid1\mid\omega}$ trace when a variable is either not used, or used at most once, or used in an unrestricted way, respectively. In this way, functions, e.g., $\LambdaExpWithType{\x}{\INT}{5}$, $\LambdaExpWithType{\x}{\INT}{\x}$, and $\LambdaExpWithType{\x}{\INT}{\x+\x}$, which have the same type in the simply-typed lambda calculus, can be distinguished by adding coeffect annotations: $\LambdaExpWithTypeAndCoeffect{\x}{\INT}{0}{5}$, $\LambdaExpWithTypeAndCoeffect{\x}{\INT}{1}{\x}$, and $\LambdaExpWithTypeAndCoeffect{\x}{\INT}{\omega}{\x+\x}$. Other examples are exact usage (coeffects are natural numbers), and privacy levels.  \emph{Graded modal types} go further, by decorating types themselves with grades,  in order  to specify how \emph{the result of an expression} should be used. 
In the different proposals in literature, grades have a similar algebraic structure, basically a semiring specifying \emph{sum} $\rsum$, \emph{multiplication} $\rmul$, and $0$ and $1$ constants, and some kind of order relation. Here, we will assume a variant of this notion called \mbox{\emph{grade algebra}.}

Resource-aware typing has been exploited  in  a fully-fledged programming language in Granule \cite{OrchardLE19}, a functional language equipped with graded modal types, hence allowing the programmer to write function declarations  similar to  those above. In Granule, different kinds of coeffects can be used at the same time, including naturals for exact usage, \bez privacy \eez levels, intervals, infinity, and products of coeffects; however, available grades are fixed  in the language. 
The initial objective of the work presented here was to study a similar support  for  Java-like languages, by introducing, in a variant of Featherweight Java ($\FJ$) \cite{IgarashiPW99}, types decorated with grades. Moreover, we wanted  these  grades to be taken, parametrically, in an arbitrary grade algebra; even more, we  did not  want  this  grade algebra  to be fixed as in Granule, but  to be extendable  by the programmer with user-defined grades, by relying on the inheritance mechanism of OO languages. In the quest for such goals, we came  up  with several ideas which are novel,  to  our knowledge, with respect to the literature on resource-aware type systems, as detailed in the outline of contributions given below.

\smallskip
\noindent\textbf{Resource-aware parametric $\FJ$ reduction.} Given a resource-aware type system, we would like to prove that typing overapproximates the use of resources. However,  resource usage  is not modeled in standard operational semantics; for this reason,  \cite{ChoudhuryEEW21} proposed an instrumented operational semantics\footnote{Subsequently the model of \cite{ChoudhuryEEW21} was used, in \cite{MarshalV22}, to trace reference counting for uniqueness.}
and proved a soundness theorem  showing  correct accounting of
resource usage. Inspired by this work, we define a \emph{resource-aware semantics} for  $\FJ$,  parametric on an arbitrary grade algebra, which tracks how much each available resource is consumed at each step, and is stuck when the needed amount of a resource is not available. Differently from \cite{ChoudhuryEEW21}, the definition of the semantics is given \emph{independently}  from  the  type system, as is the standard approach in calculi.  That is, the aim is also to provide a simple purely semantic model which takes into account usage of resources. 
The resource-aware reduction is sound with respect to the standard reduction, but clearly not complete, since a reduction step allowed in the standard semantics could be  impossible  due to resource consumption.

\smallskip
\noindent\textbf{Graded $\FJ$.} After defining the resource-aware calculus, we define the resource-aware type system. That is, types are decorated with grades, allowing the programmer to specify how a variable, a field  or  the result of a method should be used, e.g., how many times. Our approach is novel with respect to that generally used in  the  literature on graded modal types. Notably, in such works the production of types is $\produzioneinline{\T}{\ldots \mid \typeGrad{\T}{\rGrad}}{}$, that is, grade decorations can be arbitrarily nested. Correspondingly, the syntax includes an explicit \emph{box}  construct, which transforms a term of type $\T$ into a term of type $\typeGrad{\T}{\rGrad}$, through a \emph{promotion} rule which multiplies the context with $\rGrad$, and a corresponding unboxing mechanism. Here, we prefer a much lighter approach, likely more convenient for Java-like languages, where the syntax of  terms  is not affected. The production for types is  $\produzioneinline{\T}{\typeGrad{\C}{\rGrad}}{}$, that is, all types (here only class names) are (once) graded; in contexts, types are non-graded, and grades are used as coeffects, leading to a judgment of shape $\IsWFExp{\VarTypeCoeffect{\x_1}{\C_1}{\rGrad_1},\ldots,\VarTypeCoeffect{\x_n}{\C_n}{\rGrad_n}}{\e}{\typeGrad{\C}{\rGrad}}$. Finally, since there is no boxing/unboxing, there is no explicit promotion rule, but different grades can be assigned to an expression, assuming different coeffect contexts. 
We prove a soundness theorem, stating  that the graded type system overapproximates resource usage, hence  guaranteeing  soundness, and, as a consequence, completeness with respect to standard reduction for well-typed programs.

\smallskip
\noindent\textbf{Combining grades.} The next matter is how to make the language \emph{multi-graded},  in the sense that the programmer can use grades of different kinds, e.g., both natural numbers and privacy levels. This poses the problem of defining the result when grades of different kinds should be combined by the type system. This issue has been considered in the Granule language \cite{OrchardLE19}, where, however, the available kinds of grades are fixed,  hence can be combined in an ad-hoc way.  We would like to have much more flexibility, that is, to allow the programmer to define grades to be added to those already available, very much in the same way a Java programmer can define her/his own class of exceptions. To this end, we define a construction which, given a family of grade algebras and a family of homomorphisms, leads to a unique grade algebra of \emph{heterogeneous grades}.  This allows a modular approach, in the sense that the developed meta-theory, including the proof of results, applies to this case as well. 

\smallskip
\noindent\textbf{Grades as Java expressions.} Finally, we consider the issue of providing linguistic support to specify the desired grade algebras and homomorphisms. Of course this could be done by using an ad-hoc configuration language. However, we believe an interesting solution is that the grade annotations could be written themselves in Java, again analogously to what happens with exceptions. We describe how Java classes corresponding to grade algebras and homomorphisms could be written, providing some examples.

A preliminary step towards the results described in the current paper is \cite{BianchiniDGZ22}, which proposes a first version of the type system with only coeffects (types are not graded), and  a rudimentary  version of the construction described above where combining coeffects of different kinds leads to the trivial coeffect. 

In \cref{sect:algebraic} we formally define grade algebras and related notions. In \cref{sect:resource-aware} we define the parametric resource-aware reduction for $\FJ$, and in \cref{sect:GrFJ}
the parametric resource-aware type system, proving its soundness. \cref{sect:combining} defines the construction of the grade algebra of heterogeneous grades, and \cref{sect:java} illustrates how to express grade algebras and homomorphisms in Java. Finally, \cref{sect:related} surveys related work and \cref{sect:conclu} summarizes the contributions of the paper, and outlines future work. Proofs of \cref{sect:combining} are provided in the Appendix.


\section{Algebraic preliminaries}
\label{sect:algebraic} 

In this section we introduce the algebraic structures we will use throughout the paper. 
The core of our work \bpa is \epa \emph{grades}, namely, annotations in the code expressing how or how much resources are used by the program. 
As we will see, we need some operations to properly combine grades in  the resource-aware semantics and in the typing rules, hence we will assume grades to form an algebraic structure called \emph{grade algebra} defined below. 

\begin{definition}[Grade algebra] \label{def:gr-alg}
A \emph{grade algebra}  is a tuple $\RR = \ple{\RSet,\rord,\rsum,\rmul,\rzero,\rone}$ such that: 
\begin{itemize}
\item \ple{\RSet,\rord} is a partially ordered set; 
\item \ple{\RSet,\rsum,\rzero} is a commutative monoid; 
\item \ple{\RSet,\rmul,\rone} is a monoid; 
\end{itemize}
and the following axioms are satisfied: 
\begin{itemize}
\item $\rgr\rmul(\sgr\rsum\tgr) = \rgr\rmul \sgr\rsum\rgr\rmul\tgr$ and $(\sgr\rsum\tgr)\rmul\rgr = \sgr\rmul\rgr\rsum\tgr\rmul\rgr$, for all $\rgr,\sgr,\tgr\in\RSet$; 
\item $\rgr\rmul\rzero = \rzero$ and $\rzero\rmul\rgr = \rzero$, for all $\rgr\in\RSet$; 
\item if $\rgr\rord\rgr'$ and $\sgr\rord\sgr'$ then  $\rgr\rsum\sgr\rord\rgr'\rsum\sgr'$ and $\rgr\rmul\sgr\rord\rgr'\rmul\sgr'$, for all $\rgr,\rgr'\sgr,\sgr'\in\RC\RR$; 
\item $\rzero \rord \rgr$, for all $\rgr\in\RSet$. 
\end{itemize}
\end{definition}

Essentially, a grade algebra is an ordered semiring, that is, 
a semiring together with a partial order relation on its carrier which makes addition and multiplication monotonic with respect to it. 
We further require 
the zero of the semiring to be the least element of the partial order. 
 Our definition is a slight variant of others proposed in literature \cite{BrunelGMZ14,GhicaS14,McBride16,Atkey18,GaboardiKOBU16,AbelB20,OrchardLE19,ChoudhuryEEW21,WoodA22}.
In particular, the partial order models overapproximation in the usage of resources, and allows  flexibility, for instance we can have  different usage in the branches of an if-then-else construct. The fact that the zero is the least element means that, in particular, overapproximation can add unused variables, making the calculus \emph{affine}.

\begin{example}\label{ex:gr-alg}
\begin{enumerate}
\item\label{ex:gr-alg:1}
The semiring $\NN = \ple{\N,\le,+,\cdot,0,1}$ of natural numbers with the natural order and usual arithmetic operations is a grade algebra. 

\item\label{ex:gr-alg:0}
The \emph{affinity} grade algebra $\ple{\{0,1,\infty\},\le,+,\cdot,0,1}\}$ is obtained from the previous one by identifying all natural numbers greater than $1$.

\item\label{ex:gr-alg:2}
The trivial semiring \Triv, whose carrier is a singleton set $\RC\Triv = \{\infty\}$, the partial order is the equality, addition and multiplication are defined in the trivial way and $\rzero_\Triv = \rone_\Triv = \infty$, is a grade algebra. 
\item\label{ex:gr-alg:3} 
The semiring $\RRPos = \ple{\RPos,\le,+,\cdot,0,1}$ of extended non-negative real numbers with usual order and operations, extended  to $\infty$ in the expected way, is a grade algebra. 
\item\label{ex:gr-alg:4} 
A distributive lattice $\LL = \ple{\RC\LL, \le, \lor,\land,\bot,\top}$, where $\lor$ and $\land$ denote join and meet operations and $\bot$ and $\top$ the bottom and the top element, respectively, is a grade algebra. 

 \item\label{ex:gr-alg:5} 
The boolean grade algebra is $\BB = \ple{\{0,1\}, \le, \lor,\land,0,1}$ where $0\le 1$. 
It is a grade algebra since it is a distributive lattice. 

\item\label{ex:gr-alg:6} 
Given grade algebras $\RR = \ple{\RSet,\rord_\RR,\rsum_\RR,\rmul_\RR,\rzero_\RR,\rone_\RR}$ and $\RS = \ple{\SSet,\rord_\RS,\rsum_\RS,\rmul_\RS,\rzero_\RS,\rone_\RS}$,
the \emph{product} $\RR\times\RS = \ple{\{\Pair{\rgr}{\sgr}\ |\ \rgr\in\RSet\ \wedge\ \sgr\in\SSet\},\rord,\rsum,\rmul,\Pair{\rzero_\RR}{\rzero_\RS},\Pair{\rone_\RR}{\rone_\RS}}$, where operations are the pairwise application of the operations for $\RR$ and $\RS$, is a grade algebra. 

\item\label{ex:gr-alg:7} 
Given a  grade algebra $\RR = \ple{\RSet,\rord_\RR,\rsum_\RR,\rmul_\RR,\rzero_\RR,\rone_\RR}$, as in \cite{OrchardLE19}
we define
$\Extend{\RR} = \ple{\RSet + \{\infty\},\rord,\rsum,\rmul,\rzero_\RR,\rone_\RR}$ where $\rord$ extends $\rord_\RR$ by adding
$\rgr\rord\infty$ for all $\rgr\in\RC{\Extend{\RR}}$ and $\rsum$ and $\rmul$ extend $\rsum_\RR$ and $\rmul_\RR$ by 
$\rgr\rsum\infty = \infty\rsum\rgr = \infty$, for all $\rgr\in\RC{\Extend\RR}$, and 
$\rgr\rmul\infty = \infty\rmul\rgr = \infty$, for all $\rgr\in\RC{\Extend\RR}$ with $\rgr\ne\rzero_\RR$, and 
$\rzero_\RR\rmul\infty = \infty\rmul\rzero_\RR = \rzero_\RR$. 
Then, $\Extend{\RR}$ is a grade algebra.
\end{enumerate}
\end{example} 

A homomorphism of  grade algebras \fun{f}{\RR}{\RS} is a monotone function \fun{f}{\ple{\RSet,\rord_\RR}}{\ple{\SSet,\rord_{\RS}}} between the underlying partial orders, which preserves the semiring structure, that is, satisfies  the following equations: 
\begin{itemize}
\item $f(\rzero_\RR) = \rzero_\RS$  and $f(\rgr\rsum_\RR\sgr) = f(\rgr)\rsum_\RS f(\sgr)$, for all $\rgr,\sgr\in\RSet$; 
\item $f(\rone_\RR) = \rone_\RS$ and $f(\rgr\rmul_\RR\sgr) = f(\rgr)\rmul_\RS f(\sgr)$, for all $\rgr,\sgr\in\RSet$. 
\end{itemize}
Grade algebras and their homomorphisms form a category denoted by \GrAlg. 

 Consider a grade algebra \RR. 
Then, we can define functions 
\fun{\zeta_\RR}{\RC\RR}{\RC\Triv} and \fun{\iota_\RR}{\RC\NN}{\RC\RR} as follows: \label{nat-hom}
\[
\zeta_\RR(\rgr) = \infty \qquad 
 \iota_\RR(m)= \begin{cases}
\rzero_\RR & \text{if }m=0\\
\iota_\RR(n) \rsum_\RR \rone_\RR &\text{if }m=n+1 
\end{cases}
\]
Roughly, $\zeta_\RR$ maps every element of \RR to $\infty$, while $\iota_\RR$ maps a natural number $n$ to the sum in \RR of $n$ copies of $\rone_\RR$. 
We can easily check that both these functions give rise to homomorphisms \fun{\zeta_\RR}{\RR}{\Triv} and \fun{\iota_\RR}{\NN}{\RR}. 
This fact for $\zeta_\RR$ is straightforward, while for $\iota_\RR$ follows by arithmetic induction. 
Then, relying on these homomorphisms, we can prove the following result.

\begin{proposition}\label{prop:ini-fin-gralg}
The following facts hold: 
\begin{enumerate}
\item\label{prop:ini-fin-gralg:1} \NN is the initial object in \GrAlg; 
\item\label{prop:ini-fin-gralg:2} \Triv is the terminal object in \GrAlg. 
\end{enumerate}
\end{proposition}

Another kind of objects we will work with are  maps assigning grades to variables. 
These inherit a nice algebraic structure from the one of the underlying grade algebra. 

Assume a grade algebra $\RR = \ple{\RC\RR,\rord,\rsum,\rmul,\rzero,\rone}$ and a set $X$. 
The set of functions from $X$ to $\RC\RR$ 
carries a partially ordered commutative monoid structure given by the pointwise extension of the additive structure of \RR. 
That is, given \fun{\grctx,\grctx'}{X}{\RC\RR},  we define 
$\grctx \rord \grctx'$ iff, for all $x \in X$, $\grctx(x)\rord\grctx'(x)$, and 
$(\grctx\rsum\grctx')(x) = \grctx(x) \rsum \grctx'(x)$ and 
$\grctxz(x) = \rzero$, for all $x \in X$. 
Moreover, we can define a \emph{scalar multiplication}, combining elements of $\RC\RR$ and a function \fun{\grctx}{X}{\RC\RR}; 
indeed, we set 
$(\rgr\rmul\grctx)(x) = \rgr \rmul \grctx(x)$, for all $\rgr\in\RC\RR$ and $x \in X$. 
It is easy to see that this operation turns the partially ordered commutative monoid of functions from $X$ to $\RC\RR$ into a partially ordered \RR-module. 
 
The \emph{support} of a function \fun{\grctx}{X}{\RC\RR} is the set 
$\Supp\grctx = \{ x \in X \mid \grctx(x) \ne \rzero \}$. 
Denote by $\RR^X$ the set of functions \fun{\grctx}{X}{\RC\RR} with finite support. 
The partial order and operations defined above  can be safely restricted to $\RR^X$, noting that 
$\Supp{\grctxz} = \emptyset$, $\Supp{\grctx\rsum\grctx'} \subseteq \Supp{\grctx}\cup\Supp{\grctx'}$ and $\Supp{\rgr\rmul\grctx} \subseteq \Supp\grctx$. 
Therefore, $\RR^X$ carries a partially ordered \RR-module structure as well. 

 As we will see in \cref{sect:GrFJ}, \emph{coeffect contexts} are (representations of) functions in $\RR^X$, with $X$ set of variables. The fact that coeffect contexts form a module has been firstly noted in \cite{McBride16,WoodA22}, and fully formalized in \cite{BianchiniDGZS22}, which also shows a \emph{non-structural} example. That is,  a module different from $\RR^X$ described above, used in the present paper and mostly in  the  literature, is needed, where operations on coeffect contexts are not pointwise; that is, coeffects cannot be computed per-variable.


\section{Resource-aware semantics}\label{sect:resource-aware}

 Standard operational models do not say anything about resources used by the computation. 
To address this problem, we follow an approach similar to that in \cite{ChoudhuryEEW21}, that is, we define an \emph{instrumented} semantics which keeps track of resource usage, hence, in particular, it gets stuck if some needed resource is  insufficient.  
However, unlike  \cite{ChoudhuryEEW21}, the definition of our resource-aware semantics, though parameterized on a grade algebra, is given \emph{independently} of the graded type system, as is the standard approach in calculi; in the next section, we will show how the graded type system actually overapproximates resource usage, hence guarantees soundness.   As will  be  detailed in the following, the resource-aware semantics is non-deterministic, in the sense that,  when  a resource is needed, it can be consumed in different ways; hence, soundness is \emph{soundness-may}, meaning that there is a reduction which does not get stuck because of standard typing errors or resource consumption. 

\subsection{Reference calculus}
The calculus is a variant of $\FJ$ \cite{IgarashiPW99}. The syntax is reported in the top section of \cref{fig:calculus}. We write $\es$ as  a  metavariable for $\e_1, \ldots, \e_n$, $n\geq 0$, and analogously for other sequences. 
We assume \emph{variables} $\x,\y,\z,\ldots$, \emph{class names} $\C$, $\D$, \emph{field names} $\f$, and \emph{method names} $\m$.  Types are distinct from class names to mean that they could be extended to include  other types, e.g., primitive types,. 
In addition to the standard $\FJ$ constructs, we have a block expression, consisting of a local  variable declaration, and the body where this variable can be used. 

\begin{figure}
\begin{grammatica}
\produzione{\e}{\x\mid\FieldAccess{\e}{\f}\mid\ConstrCall{\C}{\es}\mid\MethCall{\e}{\m}{\es}\mid\Block{\T}{\x}{\e}{\e'}}{expression}\\
\produzione{\T}{\C}{type (class name)}\\
\produzione{\val}{\ConstrCall{\C}{\vs}}{value}\\
\end{grammatica}

\hrule

\begin{small}
\begin{math}
\begin{array}{l}
\\
\NamedRule{var}
{}
{\reduce{\x}{\stack}{\val}{\stack}}
{ \stack(\x)=\val } 
\\[4ex]

\NamedRule{field-access}
{}
{\reduce{\FieldAccess{\ConstrCall{\C}{\val_1,\dots,\val_n}}{\f_i}}{\stack}{\val_i}{\stack}}
{ \fields{\C}=\Field{\T_1}{\f_1}\dots\Field{\T_n}{\f_n}\\
  i\in 1..n   }
\\[4ex]

\NamedRule{invk}
{}
{\reduce{\MethCall{\val_0}{\m}{\val_1,\ldots,\val_n}}{\stack}{\SubstDots{\Subst{\e}{\y_0}{\this}}{\y_1}{\x_1}{\y_n}{\x_n}}{\stack'}}
{ \val_0 = \ConstrCall{\C}{\_} \\ 
  \mbody{\C}{\m}=\Pair{\x_1\dots\x_n}{\e}\\
  \y_0,\ldots,\y_n\not\in\dom{\stack}\\
  \stack'=\AddToMapDots{\stack}{\y_0}{\val_0}{\y_n}{\val_n}}
\\[8ex]

\NamedRule{block}
{}
{\reduce{\Block{\C}{\x}{\val}{\e}}{\stack}{\Subst{\e}{\y}{\x}}{\AddToMap{\stack}{\y}{\val}}}
{\y\not\in\dom{\stack}}
\\[5ex]

\NamedRule{field-access-ctx}
{\reduce{\e}{\stack}{\e'}{\stack'}}
{\reduce{\FieldAccess{\e}{\f}}{\stack}{\FieldAccess{\e'}{\f}}{\stack'}}
{} 
\\[4ex]

\NamedRule{new-ctx}
{\reduce{\e_i}{\stack}{\e'_i}{\stack'}}
{\reduce{\ConstrCall{\C}{\val_1,\dots,\val_{i-1},\e_i,\ldots,\e_n}}{\stack}{\ConstrCall{\C}{\val_1,\dots,\val_{i-1},\e'_i,\ldots,\e_n}}{\stack'}}
{} 
\\[6ex]

\NamedRule{invk-rcv-ctx}{\reduce{\e_0}{\stack}{\e'_0}{\stack'}}{\reduce{\MethCall{\e_0}{\m}{\e_1,\dots,\e_n}}{\stack}{\MethCall{\e'_0}{\m}{\e_1,\dots,\e_n}}{\stack'}}
{}\\[8ex]

\NamedRule{invk-arg-ctx}{\reduce{\e_i}{\stack}{\e'_i}{\stack'}}{\reduce{\MethCall{\val_0}{\m}{\val_1,\dots,\val_{i-1},\e_i,\ldots,\e_n}}{\stack}{\MethCall{\val_0}{\m}{\val_1,\dots,\val_{i-1},\e'_i,\ldots,\e_n}}{\stack'}}
{}\\[8ex]

\NamedRule{block-ctx}
{\reduce{\e_1}{\stack}{\e_1'}{\stack'}}
{\reduce{\Block{\C}{\x}{\e_1}{\e_2}}{\stack}{\Block{\C}{\x}{\e'_1}{\e_2}}{\stack'}}
{}
\\[5ex]
\end{array}
\end{math}
\end{small}
\caption{Syntax and standard reduction}\label{fig:calculus}
\end{figure}
 
The semantics is defined differently from the original one; that is, reduction  is defined on \emph{configurations} $\ExpStack{\e}{\stack}$, where $\stack$ is an \emph{environment}, a finite map from variables into values. In this way, variable occurrences are replaced one at a time by their value in the environment, rather than once and for all.   
This definition can be easily shown to be equivalent to the original one, and is convenient for our aims since,  in this presentation, free variables in an expression can be naturally seen as \emph{resources}  which are consumed each time a variable occurrence is \emph{used} (replaced by its value) during execution. In other words, this semantics can be naturally \emph{instrumented} by adding grades expressing the ``cost'' of resource consumption, \bez as we will do in \cref{fig:instr-red}. \eez Apart from that, the rules are straightforward;  only note that, in rules \refToRule{invk} and \refToRule{block},  parameters (including $\this$) and local variable are renamed to fresh variables, to avoid clashes. Single contextual rules are given, rather than defining evaluation contexts, to be uniform with the instrumented version, where this presentation is more convenient. 

To be concise, the class table is abstractly modeled as follows, omitting its (standard) syntax:
\begin{itemize}
\item $\fields{\C}$ gives, for each class $\C$, the sequence $\Field{\T_1}{\f_1}\ldots\Field{\T_n}{\f_n}$ of its fields,  assumed to have distinct names,   with their types;
\item $\mbody{\C}{\m}$ gives, for each method $\m$ of class $\C$, its parameters and body.
\end{itemize}

\subsection{Instrumented reduction} 
\bez This \eez reduction uses \emph{grades}, ranged over by $\rGrad, \sGrad, \tGrad$, assumed to form a grade algebra, specifying a \emph{partial order} $\rord$, a \emph{sum} $\rsum$, a \emph{multiplication} $\rmul$, and constants $\rzero$ and $\rone$, satisfying some axioms, 
as detailed in \cref{def:gr-alg} of \cref{sect:algebraic}.  

In order to keep track of usage of resources, parametrically on a given grade algebra, we \emph{instrument} reduction as follows.
\begin{itemize}
\item The environment associates, to each resource (variable), besides its value, a grade modeling its \emph{allowed usage}.
\item  Moroever, the reduction relation is \emph{graded}, that is, indexed by a grade $\rGrad$, meaning that it aims at producing a value to be used (at most) $\rGrad$ times, or, in more general (non-quantitative) terms, to be used (at most) with grade $\rGrad$. 
\item   The grade of a variable in the environment decreases, each time the variable is used, of the amount  specified in the reduction grade\footnote{More precisely, the reduction grade acts as a lower bound for this amount, 
 see comment to rule \refToRule{var}. }. 
 \item Of course, this can only happen if the current grade of the variable \emph{can} be reduced of such an amount; otherwise the reduction is stuck. 
\end{itemize} 
Before giving the formal definition, we show  some simple examples  of reductions,  considering the grade algebra of naturals of \refItem{ex:gr-alg}{1}, tracking how many times a resource is used. 
\begin{example}\label{ex:ex1-sem}
Assume the following classes:
\begin{lstlisting}
class A {}
class Pair {A first; A second}
\end{lstlisting}

We write $\pairaa$ as an abbreviation for $\ConstrCall{\Pairtt}{\ConstrCall{\Att}{},\ConstrCall{\Att}{}}$.
\begin{quote}
${\ExpStack{\Block{\Att}{\att}{\annotated{\ConstrCall{\Att}{}}{4}}{\Block{\Pairtt}{\ptt}{\annotated{\ConstrCall{\Pairtt}{\att,\att}}{2}}{\ConstrCall{\Pairtt}{\FieldAccess{\ptt}{\first},\FieldAccess{\ptt}{\second}}}}}{\emptyset}\ev_{1}}$\\
$\ExpStack{\Block{\Pairtt}{\ptt}{\annotated{\ConstrCall{\Pairtt}{\att,\att}}{2}}{\ConstrCall{\Pairtt}{\FieldAccess{\ptt}{\first},\FieldAccess{\ptt}{\second}}}}{\att\mapsto\Pair{\ConstrCall{\Att}{}}{4}}\ev_{1}$\\
$\ExpStack{\Block{\Pairtt}{\ptt}{\annotated{\ConstrCall{\Pairtt}{\ConstrCall{\Att}{},\att}}{2}}{\ConstrCall{\Pairtt}{\FieldAccess{\ptt}{\first},\FieldAccess{\ptt}{\second}}}}{\att\mapsto\Pair{\ConstrCall{\Att}{}}{2}}\ev_{1}$\\
$\ExpStack{\Block{\Pairtt}{\ptt}{\annotated{\ConstrCall{\Pairtt}{\ConstrCall{\Att}{},\ConstrCall{\Att}{}}}{2}}{\ConstrCall{\Pairtt}{\FieldAccess{\ptt}{\first},\FieldAccess{\ptt}{\second}}}}{\att\mapsto\Pair{\ConstrCall{\Att}{}}{0}}\ev_{1}$\\
$\ExpStack{\ConstrCall{\Pairtt}{\FieldAccess{\ptt}{\first},\FieldAccess{\ptt}{\second}}}{\att\mapsto\Pair{\ConstrCall{\Att}{}}{0},\ptt\mapsto\Pair{\pairaa}{2}}\ev_{1}$\\
$\ExpStack{\ConstrCall{\Pairtt}{\FieldAccess{\pairaa}{\first},\FieldAccess{\ptt}{\second}}}{\att\mapsto\Pair{\ConstrCall{\Att}{}}{0},\ptt\mapsto\Pair{\pairaa}{1}}\ev_{1}$\\
$\ExpStack{\ConstrCall{\Pairtt}{\ConstrCall{\Att}{},\FieldAccess{\ptt}{\second}}}{\att\mapsto\Pair{\ConstrCall{\Att}{}}{0},\ptt\mapsto\Pair{\pairaa}{1}}\ev_{1}$\\
$\ExpStack{\ConstrCall{\Pairtt}{\ConstrCall{\Att}{},\FieldAccess{\pairaa}{\second}}}{\att\mapsto\Pair{\ConstrCall{\Att}{}}{0},\ptt\mapsto\Pair{\pairaa}{0}}\ev_{1}$\\
$\ExpStack{\pairaa}{\att\mapsto\Pair{\ConstrCall{\Att}{}}{0},\ptt\mapsto\Pair{\pairaa}{0}}$
\end{quote}

In the example, the top-level reduction is graded 1, meaning that a single value is produced. Subterms are annotated with the grade of their reduction.  For instance, in the outer block, the initialization expression is annotated 4, meaning that its result can be used (at most) 4 times. To lighten the notation, in this example we omit the index 1.
A local variable introduced in a block is added\footnote{Modulo renaming to avoid clashes, omitted in the example for simplicity.} as another available resource in the environment, with the value and the grade of its initialization expression; for instance, the outer local variable is added with grade 4.
When evaluating the initialization expression of the inner block, which is reduced with grade 2, each time the variable $\att$ is used its grade in the environment is decremented by 2.
\end{example}

 It is important to  notice  that the annotations in subterms are \emph{not} type annotations. 
Except those in arguments of constructor invocation, explained below, 
annotations are only  needed to ensure that reduction of a subterm happens at each step with the same grade, see the formal definition below. 
We plan to investigate in future work a big-step formulation which would not need such an artifice.   
In the example above, we have chosen for the reduction of subterms the minimum grade allowing to perform the top-level reduction. 
We could have chosen any greater grade; instead, with a strictly lower grade, the reduction would be stuck.

As anticipated, in a constructor invocation $\ConstrCall{\C}{\annotated{\e_1}{\rGrad_1},\ldots,\annotated{\e_n}{\rGrad_n}}$, 
 the annotation $\rGrad_i$ plays a special role: 
intuitively, it specifies 
that the object to be constructed should contain $\rGrad_i$ copies of that field. 
Formally, this is reflected by the reduction grade of the subterm $\e_i$, which must be exactly $\rGrad\rmul\rGrad_i$, 
if $\rGrad$ is the reduction grade of the object, specifying how many copies of it the reduction is constructing. 
Correspondingly, an access to the field can be used (at most) $\rGrad\rmul\rGrad_i$ times.  
This is illustrated by the following variant of the previous example.

\begin{example}\label{ex:ex2-sem} \  Consider the term 
\begin{quote}
$\Block{\Att}{\att}{\annotated{\ConstrCall{\Att}{}}{4}}{\Block{\Pairtt}{\ptt}{\annotated{\ConstrCall{\Pairtt}{\att,\att}}{2}}{\ConstrCall{\Pairtt}{\meta{\annotated{\FieldAccess{\ptt}{\first}}{2}},\FieldAccess{\ptt}{\second}}}}$
\end{quote}
As highlighted in grey, the first argument of the constructor invocation which is the body of the inner block is now annotated with 2, meaning that the resulting object should have ``two copies'' of the field. 
As a consequence, the expression \lstinline{p.first} should be reduced with grade $2$, as shown below, where  
$\pairaa=\ConstrCall{\Pairtt}{\ConstrCall{\Att}{},\ConstrCall{\Att}{}}$,
 the first four reduction steps are as in \cref{ex:ex1-sem}  and  we explicitly write some annotations 1 for clarity 
 \begin{small}
\begin{quote}
${\ExpStack{\Block{\Att}{\att}{\annotated{\ConstrCall{\Att}{}}{4}}{\Block{\Pairtt}{\ptt}{\annotated{\ConstrCall{\Pairtt}{\annotated{\att}{1},\att}}{2}}{\ConstrCall{\Pairtt}{\meta{\annotated{\FieldAccess{\annotated{\ptt}{1}}{\first}}{2}},\FieldAccess{\ptt}{\second}}}}}{\emptyset}\ev^\ast_{1}}$\\
$\ExpStack{\ConstrCall{\Pairtt}{\meta{\annotated{\FieldAccess{\annotated{\ptt}{1}}{\first}}{2}},\FieldAccess{\ptt}{\second}}}{\att\mapsto\Pair{\ConstrCall{\Att}{}}{0},\ptt\mapsto\Pair{\pairaa}{2}}\ev_{1}$\\
$\ExpStack{\ConstrCall{\Pairtt}{\meta{\annotated{\FieldAccess{\annotated{\pairaa}{1}}{\first}}{2}},\FieldAccess{\ptt}{\second}}}{\att\mapsto\Pair{\ConstrCall{\Att}{}}{0},\ptt\mapsto\Pair{\pairaa}{1}}$\BigSpace STUCK\\
\end{quote}
\end{small}
Reduction of the subterm in grey,  aiming at constructing a value ($\ConstrCall{\Att}{}$)  which can be used twice,  is stuck, since we cannot obtain  two copies of $\ConstrCall{\Att}{}$   from the  field $\first$ of the object $\pairaa$. If we choose, instead, to reduce the occurrence of $\ptt$ to be used twice, then we get the following reduction, where again we omit steps which are as before:
 \begin{small}
 \begin{quote}
${\ExpStack{\Block{\Att}{\att}{\annotated{\ConstrCall{\Att}{}}{4}}{\Block{\Pairtt}{\ptt}{\annotated{\ConstrCall{\Pairtt}{\annotated{\att}{1},\att}}{2}}{\ConstrCall{\Pairtt}{\meta{\annotated{\FieldAccess{\annotated{\ptt}{2}}{\first}}{2}},\FieldAccess{\ptt}{\second}}}}}{\emptyset}\ev^\star_{1}}$\\
$\ExpStack{\ConstrCall{\Pairtt}{\meta{\annotated{\FieldAccess{\annotated{\ptt}{2}}{\first}}{2}},\FieldAccess{\ptt}{\second}}}{\att\mapsto\Pair{\ConstrCall{\Att}{}}{0},\ptt\mapsto\Pair{\pairaa}{2}}\ev_{1}$\\
$\ExpStack{\ConstrCall{\Pairtt}{\meta{\annotated{\FieldAccess{\annotated{\pairaa}{2}}{\first}}{2}},\FieldAccess{\ptt}{\second}}}{\att\mapsto\Pair{\ConstrCall{\Att}{}}{0},\ptt\mapsto\Pair{\pairaa}{0}}\ev_1$\\
$\ExpStack{\ConstrCall{\Pairtt}{\annotated{\ConstrCall{\Att}{}}{2},\FieldAccess{\ptt}{\second}}}{\att\mapsto\Pair{\ConstrCall{\Att}{}}{0},\ptt\mapsto\Pair{\pairaa}{0}}$\BigSpace STUCK
\end{quote}
\end{small}
\end{example}
In this case, the reduction is stuck since we consumed all the available copies of $\ptt$  to produce two copies of the field $\first$,
so now we cannot reduce $\FieldAccess{\ptt}{\second}$.  To obtain a non-stuck reduction, we should choose to reduce the initialization expression of $\ptt$ with index 3, hence that of $\att$ with index 6.   To complete the construction of the $\Pairtt$, that is, to get a non-stuck reduction, we should have 3 copies of $\ptt$ and therefore 6 copies of $\att$. 

The formal definition of the instrumented semantics is given in \cref{fig:instr-red}.  To make the notation lighter, we use the same metavariables of the standard semantics in \cref{fig:calculus}. As explained above, reduction is defined on annotated terms. Notably, in each construct, the subterms which are reduced in contextual rules  are annotated, so that their reduction always happens with a fixed grade. 

\begin{figure}
\begin{grammatica}
\produzione{\e}{\x\mid\FieldAccess{\annotated{\e}{\rGrad}}{\f}\mid\ConstrCall{\C}{\annotated{\e_1}{\rGrad_1},\ldots,\annotated{\e_n}{\rGrad_n}}\mid}{(annotated) expression}\\
\seguitoproduzione{\MethCall{\annotated{\e_0}{\rGrad_0}}{\m}{\annotated{\e_1}{\rGrad_1},\ldots,\annotated{\e_n}{\rGrad_n}}{\es}\mid\Block{\T}{\x}{\annotated{\e}{\rGrad}}{\e'}}{}
\\[2ex]
\produzione{\val}{\ConstrCall{\C}{\annotated{\val_1}{\rGrad_1},\ldots,\annotated{\val_n}{\rGrad_n}}}{(annotated) value}\\
\\
\end{grammatica}

\hrule

\begin{small}
\begin{math}
\begin{array}{l}
\\
\NamedRule{var}
{}
{\reducestack{\x}{\stack,\x\mapsto\Pair\val\sGrad}{\rGrad}{\val}{\stack,\x\mapsto\Pair{\val}{\sGrad'}}}
{ \rGrad\rord\rGrad'\ne\rzero \\
  \sGrad' \rsum \rGrad' \rord \sGrad } 
\\[4ex]

\NamedRule{field-access}
{}
{\reducestack{\FieldAccess{\annotated{\ConstrCall{\C}{\annotated{\val_1}{\rGrad_1},\dots,\annotated{\val_n}{\rGrad_n}}}{\rGrad}}{\f_i}}{\stack}{\sGrad}{\val_i}{\stack}}
{ \fields{\C}=\Field{\T_1}{\f_1}\dots\Field{\T_n}{\f_n}\\
  i\in 1..n \\
  \sGrad\rord\rGrad\rmul\rGrad_i  }
\\[4ex]

\NamedRule{invk}
{}
{\reducestack{\MethCall{\annotated{\val_0}{\rGrad_0}}{\m}{\annotated{\val_1}{\rGrad_1},\dots,\annotated{\val_n}{\rGrad_n}}}{\stack}{\rGrad}{\SubstDots{\Subst{\e}{\y_0}{\this}}{\y_1}{\x_1}{\y_n}{\x_n}}{\stack'}}
{ \val_0 = \ConstrCall{\C}{\_} \\ 
  \mbody{\C}{\m}=\Pair{\x_1\dots\x_n}{\e}\\
  \y_0,\ldots,\y_n\not\in\dom{\stack}\\
  \stack'=\AddToMapDots{\stack}{\y_0}{\Pair{\val_0}{\rGrad_0}}{\y_n}{\Pair{\val_n}{\rGrad_n}}}
\\[8ex]

\NamedRule{block}
{}
{\reducestack{\Block{\C}{\x}{\annotated{\val}{\rGrad}}{\e}}{\stack}{\sGrad}{\Subst{\e}{\y}{\x}}{\AddToMap{\stack}{\y}{\Pair{\val}{\rGrad}}}}
{\y\not\in\dom{\stack}}
\\[5ex]

\NamedRule{field-access-ctx}
{\reducestack{\e}{\stack}{\rGrad}{\e'}{\stack'}}
{\reducestack{\FieldAccess{\annotated{\e}{\rGrad}}{\f}}{\stack}{\sGrad}{\FieldAccess{\annotated{\e'}{\rGrad}}{\f}}{\stack'}}
{} 
\\[4ex]

\NamedRule{new-ctx}
{\reducestack{\e_i}{\stack}{\rGrad\rmul\rGrad_i}{\e'_i}{\stack'}}
{\begin{array}{l}
\ExpStack{\ConstrCall{\C}{\annotated{\val_1}{\rGrad_1},\dots,\annotated{\val_{i-1}}{\rGrad_{i-1}},\annotated{\e_i}{\rGrad_i},\ldots,\annotated{\e_n}{\rGrad_n}}}{\stack}\ev_{\rGrad}\\
\HugeSpace\ExpStack{\ConstrCall{\C}{\annotated{\val_1}{\rGrad_1},\dots,\annotated{\val_{i-1}}{\rGrad_{i-1}},\annotated{\e'_i}{\rGrad_i},\ldots,\annotated{\e_n}{\rGrad_n}}}{\stack'}
\end{array}}
{} 
\\[8ex]

\NamedRule{invk-rcv-ctx}{\reducestack{\e_0}{\stack}{\rGrad_0}{\e'_0}{\stack'}}{\reducestack{\MethCall{\annotated{\e_0}{\rGrad_0}}{\m}{\annotated{\e_1}{\rGrad_1},\dots,\annotated{\e_n}{\rGrad_n}}}{\stack}{\rGrad}{\MethCall{\annotated{\e'_0}{\rGrad_0}}{\m}{\annotated{\e_1}{\rGrad_1},\dots,\annotated{\e_n}{\rGrad_n}}}{\stack'}}
{}\\[6ex]

\NamedRule{invk-arg-ctx}{\reducestack{\e_i}{\stack}{\rGrad_i}{\e'_i}{\stack'}}
{\begin{array}{ll}
\ExpStack{\MethCall{\annotated{\e_0}{\rGrad_0}}{\m}{\annotated{\val_1}{\rGrad_1},\dots,\annotated{\val_{i-1}}{\rGrad_{i-1}},\annotated{\e_i}{\rGrad_i},\ldots,\annotated{\e_n}{\rGrad_n}}}{\stack}\ev_{\rGrad}{}{}\\\HugeSpace\ExpStack{\MethCall{\annotated{\e_0}{\rGrad_0}}{\m}{\annotated{\val_1}{\rGrad_1},\dots,\annotated{\val_{i-1}}{\rGrad_{i-1}},\annotated{\e'_i}{\rGrad_i},\ldots,\annotated{\e_n}{\rGrad_n}}}{\stack'}
\end{array}}
{}\\[8ex]

\NamedRule{block-ctx}
{\reducestack{\e_1}{\stack}{\sGrad}{\e_1'}{\stack'}}
{\reducestack{\Block{\C}{\x}{\annotated{\e_1}{\sGrad}}{\e_2}}{\stack}{\rGrad}{\Block{\C}{\x}{\annotated{\e'_1}{\bez\sGrad\eez}}{\e_2}}{\stack'}}
{}
\\[5ex]
\end{array}
\end{math}
\end{small}
\caption{Instrumented reduction}\label{fig:instr-red}
\end{figure}

 In rule \refToRule{var}, which is the key rule where resources are consumed, 
a variable occurrence is replaced by the associated value in the environment, 
and its grade $\sGrad$ decreases to $\sGrad'$, burning a non-zero amount $\rGrad'$ of resources which has to be at least the reduction grade. 
The side condition $\rGrad'\rsum\sGrad'\rord\sGrad$ ensures that the initial grade of the variable suffices to cover 
both  the consumed grade and the residual grade.  
 To show why  the amount of resource consumption should be non-zero,  consider, e.g., the following variant of \cref{ex:ex1-sem}:
\begin{quote}
$\ExpStack{\Block{\Att}{\att}{\annotated{\ConstrCall{\Att}{}}{4}}{\Block{\Pairtt}{\ptt}{\annotated{\ConstrCall{\Pairtt}{\att,\att}}{0}}{\ConstrCall{\Pairtt}{\att,\att}}}}{\emptyset}$
\end{quote}
The local variable $\ptt$ is never used in the  body of the block, so it makes sense for its initialization expression to be reduced with grade 0, since execution needs no copies of the result. Yet, the expression \emph{needs to be reduced}, and to produce its useless result two copies of $\att$ are consumed; 
in a sense, they are wasted. However,  such resource usage is tracked, whereas it would be lost if decrementing by 0.  Removing the non-zero requirement would lead to a variant of resource-aware reduction where usage of resource which are useless to construct the final result is not tracked.

 In rule \refToRule{field-access}, the reduction grade should be (overapproximated by)  the multiplication of the grade of the receiver with that of the field (constructor argument). 
Indeed, the former specifies how many copies of the object we have and the latter how many copies of the field each of such objects has; 
thus, their product provides an upper bound to the grade of the resulting value.  
Note that, in this way, some reductions could be \emph{forbidden}.
For instance, taking the grade algebra of naturals,  an access to a field whose value can be used 3 times, of an object reduced with grade 2, can be reduced with grade (at most) $6$.    Another more significant example is given in the following, taking the grade algebra of \emph{privacy levels}. 

Rule \refToRule{invk} adds each method parameter, including $\this$, as available resource in the environment, modulo renaming with a fresh variable to avoid clashes. The associated value and grade are that of the corresponding argument. Rule \refToRule{block} is exactly analogous, apart that only one variable is added. 

Coming to contextual rules, the reduction grade of the subterm is that of the corresponding annotation, so that all steps happen with a fixed grade. 
The only exception is rule \refToRule{new-ctx}, where, symmetrically to rule \refToRule{field-access}, the reduction grade for subterms should be the multiplication  of the reduction grade of the object with  the annotation of the field (constructor argument), capturing the intuition that the latter specifies the grade of the field for a single copy of the object.  
For instance, taking the grade algebra of naturals,  to obtain an object which can be used twice, with a field which can be used 3 times, the value of such field should be an object which can be used 6 times. 

Note that, besides the standard typing errors such as looking for a missing method or field, reduction graded $\rGrad$ can get stuck since either rule \refToRule{var} cannot be applied since the side conditions do not hold, or rule \refToRule{field-access} cannot be applied since the side condition $\sGrad\rord\rGrad\rmul\rGrad_i$ does not hold. 
Informally, either some resource (variable) is exhausted, that is, can no longer be replaced by its value, or some field of some  object  cannot be extracted.
 It is also important to note that the instrumented reduction is non-deterministic, due to rule \refToRule{var}.

\smallskip

In the grade algebra used in the previous example, grades model \emph{how many times} resources are used.
However, grades can also model a non-quantitative\footnote{Suck kind of applications are called \emph{informational} in \cite{AbelB20}.} knowledge, that is, track possible \emph{modes} in which a resource can be used, or, in other words, possible \emph{constraints} on how it could be used. A typical example of this situation are \emph{privacy levels}, which can be formalized 
similarly to what is done in \cite{AbelB20}, as described below. 
\begin{example}\label{ex:privacy}
Starting from any distributive semilattice lattice $\LL$, like in \refItem{ex:gr-alg}{4}, 
define $\LL_0 = \ple{\RC{\LL_0}, \le_0, \lor_0, \land_0, 0, \top}$, where
 $\RC{\LL_0} = \RC\LL + \{0\}$ with  $0\le_0 x$, $x \lor_0 0 = 0 \lor_0 x = x$ and $x \land_0 0 = 0 \land_0 x = 0$,  for all $x \in \RC\LL$; 
on elements of $\RC\LL$ the order and the operations are those of $\LL$. 
That is, we assume that the privacy levels form a distributive semilattice   with order representing ``decreasing privacy'', 
and we add a grade $0$  modeling ``non-used''.   The simplest instance consists of just two privacy levels, that is, $0\rord\private\rord\public$. Sum is the join, meaning that we obtain a privacy level which is less restrictive than both: for instance, a variable which is used as $\public$ in a subterm, and as $\private$ in another, is overall used as $\public$. Multiplication is the meet, meaning that we obtain a privacy level which is more restrictive than both: for instance, an access to a field whose value has been obtained in \lstinline{public} mode, of an object reduced in \lstinline{private} mode, is reduced in \lstinline{private} mode\footnote{As in \emph{viewpoint adaptation} \cite{DietlEtAl07}, where permission to a field access can be restricted based on the
 permission to the base object.}. Note that exactly the same structure could be used to model, e.g., rather than privacy levels, modifiers $\readonly$ and $\mutable$ in an imperative setting, corresponding to forbid field assignment and no restrictions, respectively. The following examples illustrates the use of such grade algebra.  We write $\privateS$ and $\publicS$ for short, and classes $\Att$ and $\Pairtt$ are as in the previous examples. 
\begin{enumerate}
\item\label{ex:red1:1}
Let $\e_1=\Block{\Att}{\ytt}{\annotated{\ConstrCall{\Att}{}}{\publicS}}{\Block{\Att}{\xtt}{\annotated{\ytt}{\privateS}}{\xtt}}$
and $\pxtt$ be either $\publicS$ or $\privateS$, $\e_1$ starting with the empty environment reduces with grade $\private$ as follows:
\[
\begin{array}{lcl}
\ExpStack{\e_1}{\emptystack}&\ev_{\privateS} &   \ExpStack{\Block{\Att}{\xtt}{\annotated{\ytt}{\privateS}}{\xtt}}{\hpEl{\ytt}{\ConstrCall{\Att}{}}{\publicS}} \text{ with }\refToRule{block} \\
&\ev_{\privateS} &   \ExpStack{\Block{\Att}{\xtt}{\annotated{\ConstrCall{\Att}{}}{\privateS}}{\xtt}}{\hpEl{\ytt}{\ConstrCall{\Att}{}}{\pxtt}} \text{ with }\refToRule{block-ctx}\text{ and}\\
  & &\reducestack{\ytt}{\hpEl{\ytt}{\ConstrCall{\Att}{}}{\publicS}}{\bez\privateS\eez}{\ConstrCall{\Att}{}}{\hpEl{\ytt}{\ConstrCall{\Att}{}}{\pxtt}}\\
  &\ev_{\privateS} &   \ExpStack{\xtt}{\hpEl{\ytt}{\ConstrCall{\Att}{}}{\pxtt},\hpEl{\xtt}{\ConstrCall{\Att}{}}{\privateS}} \text{ with }\refToRule{block} \\
  &\ev_{\privateS} &   \ExpStack{\ConstrCall{\Att}{}}{\hpEl{\ytt}{\ConstrCall{\Att}{}}{\pxtt},\hpEl{\xtt}{\ConstrCall{\Att}{}}{\privateS}} \text{ with }\refToRule{var} 
\end{array}
\]
Instead reduction with grade $\public$ would be stuck since $\publicS \not\rord\privateS$ and so 
\[
\ExpStack{\xtt}{\hpEl{\ytt}{\ConstrCall{\Att}{}}{\pxtt},\hpEl{\xtt}{\ConstrCall{\Att}{}}{\privateS}}\not\ev_{\publicS}
\] 
Also the reduction of $\e_2=\Block{\Att}{\ytt}{\annotated{\ConstrCall{\Att}{}}{\privateS}}{\Block{\Att}{\xtt}{\annotated\ytt{\publicS}}{\xtt}}$ with grade $\private$
\[
\begin{array}{lcl}
\ExpStack{\e_2}{\emptystack}&\ev_{\privateS} &   \ExpStack{\Block{\Att}{\xtt}{\annotated\ytt{\publicS}}{\xtt}}{\hpEl{\ytt}{\ConstrCall{\Att}{}}{\privateS}} \text{ with }\refToRule{Block} \\
&\not\ev_{\privateS} &  
\end{array}
\]
would be stuck since $\ExpStack{\ytt}{\hpEl{\ytt}{\ConstrCall{\Att}{}}{\privateS}}\not\ev_{\publicS}$. Note that both $\e_1$ and $\e_2$ reduce to $\ConstrCall{\Att}{}$ with
the semantics of \cref{fig:calculus}.
\item\label{ex:red1:2} Let
$e_3=\Block{\Att}{\xtt}{\annotated{\ConstrCall{\Att}{}}\publicS}{\ConstrCall{\Btt}{\annotated\xtt{\publicS},\annotated\xtt{\privateS}}}$, $\e_3$ starting with the empty environment reduces with grade $\public$ as follows:
\[
\begin{array}{lcl}
\ExpStack{\e_3}{\emptystack}&\ev_{\publicS} &   \ExpStack{\ConstrCall{\Btt}{\annotated\xtt{\publicS},\annotated\xtt{\privateS}}}{\hpEl{\xtt}{\ConstrCall{\Att}{}}{\publicS}} \text{ with }\refToRule{Block} \\
&\ev_{\publicS} &   \ExpStack{\ConstrCall{\Btt}{\annotated{\ConstrCall{\Att}{}}\publicS,\annotated\xtt\privateS}}{\hpEl{\xtt}{\ConstrCall{\Att}{}}{\pxtt}} \text{ with }\refToRule{New-Ctx}\text{ and}\\
  & &\reducestack{\xtt}{\hpEl{\xtt}{\ConstrCall{\Att}{}}{\publicS}}{\publicS}{\ConstrCall{\Att}{}}{\hpEl{\xtt}{\ConstrCall{\Att}{}}{\pxtt}}\\
&\ev_{\publicS} &   \ExpStack{\ConstrCall{\Btt}{\annotated{\ConstrCall{\Att}{}}\publicS,\annotated{\ConstrCall{\Att}{}}\privateS}}{\hpEl{\xtt}{\ConstrCall{\Att}{}}{\pxtt}} \text{ with }\refToRule{New-Ctx}\text{ and}\\
  & &\reducestack{\xtt}{\hpEl{\xtt}{\ConstrCall{\Att}{}}{\pxtt}}{\privateS}{\ConstrCall{\Att}{}}{\hpEl{\xtt}{\ConstrCall{\Att}{}}{\pxtt}}\\
\end{array}
\]
It is easy to see that also $\ExpStack{\e_3}{\emptystack}\ev^\ast_{\privateS} \ExpStack{\ConstrCall{\Btt}{\annotated{\ConstrCall{\Att}{}}\publicS,\annotated{\ConstrCall{\Att}{}}\privateS}}{\hpEl{\xtt}{\ConstrCall{\Att}{}}{\pxtt}} $. So we have 
\[
\begin{array}{lcl}
\ExpStack{\FieldAccess{\annotated{\e_3}{\rGrad}}{\fX}}{\emptystack}&\ev^\ast_{\sGrad} &   
\ExpStack{ \FieldAccess{ \annotated{\ConstrCall{\Btt}{\annotated{\ConstrCall{\Att}{}}\publicS,\annotated{\ConstrCall{\Att}{}}\privateS}}\rGrad }{\fX}}{\hpEl{\xtt}{\ConstrCall{\Att}{}}{\pxtt}}
\end{array}
\]
where $\fX$ can be either $\fO$ or $\fT$ and $\rGrad$ and $\sGrad$ can be either $\publicS$ or $\privateS$. Now, the
reductions of grade $\privateS$ accessing either  $\fO$ or $\fT$ produce the value of the fields
\[
\begin{array}{lcl}  
\ExpStack{ \FieldAccess{ \annotated{\ConstrCall{\Btt}{\annotated{\ConstrCall{\Att}{}}\publicS,\annotated{\ConstrCall{\Att}{}}\privateS}}\rGrad }{\fX}}{\_}\ev_{\privateS} \ExpStack{\ConstrCall{\Att}{}}{\_}
\end{array}
\]
However, looking at the reductions of grade $\publicS$, only 
\[
\begin{array}{lcl}  
\ExpStack{ \FieldAccess{ \annotated{\ConstrCall{\Btt}{\annotated{\ConstrCall{\Att}{}}\publicS,\annotated{\ConstrCall{\Att}{}}\privateS}}\publicS }{\fO}}{\_}\ev_{\publicS} \ExpStack{\ConstrCall{\Att}{}}{\_}
\end{array}
\]
is not stuck. That is, we  produce a value that can be used as $\public$ only if we get a $\public$ field of a $\public$ object, whereas any value can be used as $\private$.
\end{enumerate}
\end{example}
 
We now  state  some simple properties of the semantics 
we will use to prove type soundness.
The former establishes that  reduction does not remove  variables from the environment, the latter states that we can always decrease the grade of a reduction step. 
\begin{proposition}\label{prop:step-env}
If $\reducestack{\e}{\stack}{\rGrad}{\e'}{\stack'}$ then $\dom\stack\subseteq\dom{\stack'}$ and for all $\x\in\dom\stack$, $\stack(\x) = \Pair\val\rGrad$ implies $\stack'(\x) = \Pair\val\sGrad$ with $\sGrad\rord\rGrad$. 
\end{proposition}
\begin{proposition}\label{prop:red-gr}
If  $\reducestack{\e}{\stack}{\rGrad}{\e'}{\stack'}$  and $\sGrad \rord \rGrad$ then $\reducestack{\e}{\stack}{\sGrad}{\e'}{\stack'}$.
\end{proposition}

We expect the instrumented reduction to be \emph{sound} with respect to the standard reduction, in the sense that by erasing annotations from an instrumented reduction sequence we get a standard reduction sequence.
This is formally stated below.

For any $\e$ expression, let us denote by $\erase{\e}$ the expression obtained by erasing annotations, defined in the obvious way, and analogously for environments\bez, where grades associated to variables are removed as well\eez.
\begin{proposition}[Soundness of instrumented semantics] \label{prop:gr-sem-sound}
If $\reducestack{\e}{\stack}{\rGrad}{\e'}{\stack'}$, then $\reduce{\erase{\e}}{\erase{\stack}}{\erase{\e'}}{\erase{\stack'}}$.
\end{proposition}
\bez The converse does not hold, since a configuration could be annotated in a way that makes it stuck; notably, some resource (variable) could be exhausted or some field of an  object  could not be extracted. \eez
The graded type system in the next section will generate annotations which ensure soundness, hence also completeness with respect to the standard reduction.


\section{Graded Featherweight Java}\label{sect:GrFJ}

We define the parametric resource-aware type system, show some examples, and prove its soundness. 

\subsection{Graded type system}

Types (class names) are annotated with \emph{grades}, as shown in \cref{fig:graded-calculus}. 
 
 As anticipated at the end of \cref{sect:algebraic},  a \emph{coeffect context}, of shape $\grctx=\VarCoeffect{\x_1}{\rgr_1}, \ldots, \VarCoeffect{\x_n}{\rgr_n}$, where order is immaterial and $\x_i\neq\x_j$ for $i\neq j$, represents a map from variables  to  grades (called \emph{coeffects} when used in this position) where only a finite number of variables have non-zero coeffect. 
A \emph{(type-and-coeffect) context}, of shape $\Gamma=\VarTypeCoeffect{\x_1}{\C_1}{\rgr_1},\ldots,\VarTypeCoeffect{\x_n}{\C_n}{\rgr_n}$, with analogous conventions, represents the pair of the standard type context $\VarType{\x_1}{\C_1}\ldots,\VarType{\x_n}{\C_n}$, and the coeffect context $\VarCoeffect{\x_1}{\rgr_1}, \ldots, \VarCoeffect{\x_n}{\rgr_n}$. 
We write $\dom{\Gamma}$ for $\{\x_1,\ldots,\x_n\}$.

As customary in type-and-coeffect systems, in typing rules contexts are combined by means of some operations, which are, in turn, defined in terms of  the  corresponding operations on coeffects (grades). \\
 More precisely, we define: 
\begin{itemize}
\item a partial order $\ROrdCtx$ 
\begin{align*} 
\EmptyCtx &\ROrdCtx \EmptyCtx  & \\ 
\VarTypeCoeffect{\x}{\C}{\sGrad}, \Gamma &\ROrdCtx \VarTypeCoeffect{\x}{\C}{\rGrad}, \Delta &&\text{if $\sGrad \rord \rGrad$  and $\Gamma \ROrdCtx \Delta$} \\ 
\Gamma &\ROrdCtx \VarTypeCoeffect{\x}{\C}{\rGrad}, \Delta  &&\text{if $\x\not\in\dom{\Gamma}$  and $\Gamma \ROrdCtx \Delta$} 
\end{align*} 
\item a sum $\SumCtx$ 
\begin{align*} 
\EmptyCtx \SumCtx \Gamma &= \Gamma  \\ 
(\VarTypeCoeffect{\x}{\C}{\sGrad}, \Gamma) \SumCtx ( \VarTypeCoeffect{\x}{\C}{\rGrad}, \Delta) &= \VarTypeCoeffect{\x}{\C}{\sGrad \rsum \rGrad}, (\Gamma \SumCtx\Delta) \\ 
(\VarTypeCoeffect{\x}{\C}{\sGrad}, \Gamma) \SumCtx \Delta &= \VarTypeCoeffect{\x}{\C}{\sGrad}, (\Gamma \SumCtx \Delta)  
  &\text{if $\x \notin \dom{\Delta}$} 
\end{align*} 
\item a scalar multiplication $\MulCtx$ 
\begin{align*}  
\sGrad \MulCtx \EmptyCtx = \EmptyCtx 
  && 
\sGrad \MulCtx (\VarTypeCoeffect{\x}{\C}{\rGrad},\Gamma) =  \VarTypeCoeffect{\x}{\C}{\sGrad \rmul \rGrad}, (\sGrad {\MulCtx} \Gamma) 
\end{align*} 
\end{itemize}  
As the reader may notice, 
these operations on type-and-coeffect contexts can be equivalently defined by lifting the corresponding operations 
on coeffect contexts, which are the pointwise extension of those on coeffects, 
to handle types as well. 
In this step, the addition becomes partial since a variable in the domain of both contexts is required to have the same type.  

\begin{figure}[t]
\begin{grammatica}
\produzione{\e}{\x\mid\FieldAccess{\e}{\f}\mid\ConstrCall{\C}{\es}\mid\MethCall{\e}{\m}{\es}\mid\Block{\T}{\x}{\e}{\e'}}{expression}\\
\produzione{\T}{\typeGrad{\C}{\rGrad}\BigSpace}{(graded) type}\\
\produzione{\val}{\ConstrCall{\C}{\vs}}{value}\\
\end{grammatica}
\caption{Syntax with grades}\label{fig:graded-calculus}
\end{figure}

The type system relies on the type information extracted from the class table, which, again to be concise, is abstractly modeled as follows:
\begin{itemize}
\item the subtyping relation $\leq$ on class names is the reflexive and transitive closure of the \texttt{extends} relation
\item $\mtype{\C}{\m}$ gives, for each method $\m$ of class $\C$, its enriched method type,  where  the types of the parameters and of  $\this$ have coeffect annotations.
\end{itemize}

 Moreover, $\fields{\C}$ gives now a sequence $\Field{\typeGrad{\C_1}{\rGrad_1}}{\f_1}\ldots\Field{\typeGrad{\C_n}{\rGrad_n}}{\f_n}$, meaning that, to construct an object of type $\C$, we need to provide, for each $i\in 1..n$, a value with a grade at least $\rGrad_i$. 

The subtyping relation on graded types is defined as follows:
\begin{quote}
$\typeGrad{\C}{\rGrad}\leq\typeGrad{\D}{\sGrad}$ iff $\C\leq\D$ and $\sGrad\rord\rGrad$
\end{quote}
That is, a graded type is a subtype of another if the class is a heir class and the grade is more constraining. For instance,  taking the affinity grade algebra of \refItem{ex:gr-alg}{0},  an invocation of a method with return type $\typeGrad{\C}{\omega}$ can be used in a context where a type $\typeGrad{\C}{\rone}$ is required, e.g., to initialize a $\typeGrad{\C}{\rone}$ variable. 

 The typing judgment has shape $\IsWFExpI{\Gamma}{\e}{\T}{\e'}$, where $\Gamma$ is a type-and-coeffect context, and $\e'$ is an annotated expression, as defined in \cref{fig:instr-red}. That is, typechecking generates annotations in code such that evaluation cannot  get  stuck, as will be formally expressed and proved  in the following.

In a well-typed class table, method bodies are expected to conform to method types. That is,
$\mtype{\C}{\m}$ and $\mbody{\C}{\m}$ should be either both undefined or both defined with the same number of parameters. In the latter case, the method body should be well-typed with respect to the method type, notably by typechecking the method body we should get coeffects which are (overapproximated by) those specified in the annotations.  Formally, if $\mbody{\C}{\m}=\Pair{\x_1\dots\x_n}{\e}$, and $\mtype{\C}{\m}=\funType{\rGrad_0,\typeGrad{\C_1}{\rGrad_1} \ldots \typeGrad{\C_n}{\rGrad_n}}{\T}$, then the following condition must hold:

\begin{quote}\label{t-meth}
\refToRule{t-meth}\Space${\IsWFExpI{\VarTypeCoeffect{\this}{\C}{\rgr_0},\VarTypeCoeffect{\x_1}{\C_1}{\rgr_1},\ldots,\VarTypeCoeffect{\x_n}{\C_n}{\rgr_n}}{\e}{\T}{\e'}}$
\end{quote}
 Moreover, we assume the standard coherence conditions on the class table with respect to inheritance. 
That is, if $\C\leq\D$, then 
$\fields\D$ is a prefix of $\fields\C$ and, 
if $\mtype\C\m = \funType{\rGrad_0,\typeGrad{\C_1}{\rGrad_1} \ldots \typeGrad{\C_n}{\rGrad_n}}{\T}$, then 
$\mtype\D\m = \funType{\rGrad_0,\typeGrad{\C_1}{\rGrad_1}\ldots \typeGrad{\C_n}{\rGrad_n}}{\T'}$ with $\T' \leq \T$.  

In \cref{fig:typing}, we describe the typing rules,  which are \emph{parameterized} on the underlying grade algebra. 

\begin{figure}
\begin{math}
\begin{array}{l}
\NamedRule{t-sub}
{\IsWFExpI{\Gamma}{\e}{\T}{\e'}}
{\IsWFExpI{\Gamma'}{\e}{\T'}{\e'}}
{ \Gamma\ROrdCtx\Gamma' \\ 
  \T\leq\T' } 
\BigSpace

\NamedRule{t-var}
{}
{ \IsWFExpI{\VarTypeCoeffect{\x}{\C}{\rGrad}}{\x}{\typeGrad{\C}{\rGrad}}{\x} } 
{ \rGrad\neq\rzero} 

\\[4ex]

\NamedRule{t-field-access}
{ \IsWFExpI{\Gamma}{\e}{\typeGrad{\C}{\rGrad}}{\e'} }
{ \IsWFExpI{\Gamma}{\FieldAccess{\e}{\f_i}}{\typeGrad{\C_i}{\rGrad\rmul\rGrad_i}}{\FieldAccess{\annotated{\e'}{\rGrad}}{\f_i}} }
{ \fields{\C}=\Field{\typeGrad{\C_1}{\rGrad_1}}{\f_1} \ldots \Field{\typeGrad{\C_n}{\rGrad_n}}{\f_n} } 
\\[4ex]

\NamedRule{t-new}
{\IsWFExpI{\Gamma_i}{\e_i}{\typeGrad{\C_i}{\rGrad\rmul\rGrad_i}}{\e'_i}\Space \forall i\in 1..n}
{\begin{array}{l}
\IsWFExpI{\Gamma_1\SumCtx\ldots\SumCtx\Gamma_n}{\ConstrCallTuple{\C}{\e}{n}}{\typeGrad{\C}{\rGrad}}{}\\
\HugeSpace\ConstrCallATuple{\C}{\e'}{\rGrad}{n}
\end{array}}
{\fields{\C}=\Field{\typeGrad{\C_1}{\rGrad_1}}{\f_1} \ldots \Field{\typeGrad{\C_n}{\rGrad_n}}{\f_n}} 
\\[7ex]

\NamedRule{t-invk}
{ \IsWFExpI{\Gamma_0}{\e_0}{\typeGrad{\C}{\rGrad_0}}{\e'_0} 
  \Space
  \IsWFExpI{\Gamma_i}{\e_i}{\typeGrad{\C_i}{\rGrad_i}}{\e'_i} \Space \forall 1\in 1..n }
{ \begin{array}{l}
\IsWFExpI{\Gamma_0\SumCtx  \ldots \SumCtx \Gamma_n}{\MethCallTuple{\e_0}{\m}{\e}{n}}{\T}{}\\
\HugeSpace{\MethCallATuple{\annotated{\e'_0}{\rGrad_0}}{\m}{\e}{\rGrad}{n}} 
\end{array}}
{ \mtype{\C}{\m}=\funType{\rGrad_0,\typeGrad{\C_1}{\rGrad_1} \ldots \typeGrad{\C_n}{\rGrad_n}}{\T}\\ }
\\[7ex]

\NamedRule{t-block}
{ \IsWFExpI{\Gamma_1}{\e_1}{\typeGrad{\C}{\rGrad}}{\e'_1}
  \Space
  \IsWFExpI{\Gamma_2, \VarTypeCoeffect{\x}{\C}{\rGrad}}{\e_2}{\T}{\e'_2} }
{ \IsWFExpI{\Gamma_1 \SumCtx \Gamma_2}{\Block{\typeGrad{\C}{\rGrad}}{\x}{\e_1}{\e_2}}{\T}{\Block{\C}{\x}{\annotated{\e'_1}{\rGrad}}{\e'_2}}  }
{}
\\[4ex] 
 
\NamedRule{t-env}
{ \IsWFExpI{}{\val_i}{\typeGrad{\C_i}{\rGrad_i}}{\val'_i} \Space \forall i \in 1..n }
{ \IsWFEnvI{\Gamma}{\stack}{\stack'} }
{ \Gamma = \VarTypeCoeffect{\x_1}{\C_1}{\rGrad_1},\ldots,\VarTypeCoeffect{\x_n}{\C_n}{\rGrad_n} \\ 
  \stack = \x_1\mapsto\Pair{\val_1}{\rGrad_1},\ldots,\x_n\mapsto\Pair{\val_n}{\rGrad_n} \\ 
  \stack' = \x_1\mapsto\Pair{\val'_1}{\rGrad_1},\ldots,\x_n\mapsto\Pair{\val'_n}{\rGrad_n}  }

\\[4ex]
\NamedRule{t-conf}
{ \IsWFExpI{\Delta}{\e}{\T}{\e'} \Space \IsWFEnvI{\Gamma}{\stack}{\stack'} }
{ \IsWFConfI{\Gamma}{\e}{\stack}{\T}{\e'}{\stack'}  }
{ \Delta\ROrdCtx\Gamma } 

\end{array}
\end{math}

\caption{Graded type system}\label{fig:typing}
\end{figure}

 In rule \refToRule{t-sub},  both the coeffect context and the (graded) type can be made more general. This means that, on one hand, variables can get less constraining coeffects. For instance, assuming again affinity coeffects, an expression which can be typechecked assuming to use a given variable at most once (coeffect 1) can be typechecked as well with no constraints (coeffect $\omega$). On the other hand, recalling that grades are contravariant in types, an expression can get a more constraining grade. For instance, an expression of grade $\omega$ can be used where a grade 1 is required.

If we take $\rGrad=\rone$, then rule \refToRule{t-var} is analogous to the standard rule for variable in coeffect systems, where the coeffect context is the map where the given variable is used once, and no other is used.  Here, more generally, the variable can get an arbitrary grade $\rGrad$, provided that  it gets the same grade in the context.  However, the use of the variable cannot be just discarded, as expressed by the side condition $\rGrad\neq\rzero$.

In rule \refToRule{t-field-access}, the grade of the field is multiplied by the grade of the receiver. As already mentioned, this is a form of \emph{viewpoint adaptation} \cite{DietlEtAl07}. For instance, using affinity grades,  a field graded $\omega$ of an object graded $1$ can be used at most once.

In rule \refToRule{t-new}, analogously to rule \refToRule{t-var}, the constructor invocation can get an arbitrary grade $\rGrad$, provided that the grades of the fields are multiplied by the same grade. Coeffects of the subterms are summed, as customary in type-and-coeffect systems. 

In rule \refToRule{t-invk}, the coeffects of the arguments are summed as well. The rule uses the function \aux{mtype} on the class table, which,  given a class name and a method name, returns its parameter and return (graded) types. For the implicit parameter $\this$ only the grade 
is specified.  Note that the grades of the parameters are used in two different ways:
\begin{itemize}
\item as (part of) types, when typechecking the arguments
\item as coeffects, when typechecking  the method body.
\end{itemize}
In rule \refToRule{t-block}, the coeffects of the initialization expression are summed with those of the body, excluding the local variable. Analogously to method parameters, the grade of the local variable is both used as (part of) type, when typechecking the initialization expression, and as coeffect, when typechecking the body.

Finally, we have straightforward rules for typing environments and configurations. Values in the environment are assumed to be closed, since we are in a call-by-value calculus. Also note that, in the judgment for environments and configurations, since no subsumption rule is available, variables in the context are exactly those in the domain of the environment, which are a superset of those used in the expression.  

\subsection{Examples}
\begin{example}\label{ex:ex1}

We show a simple example illustrating the use of graded types, assuming affinity grades. We write in square brackets the grade of the implicit $\this$ parameter. The class \lstinline{Pair} declares three versions of the getter for the \lstinline{first} field, which differ for the grade of the result: either $0$, meaning that the result of the method \emph{cannot} be used, or $1$, meaning it can be used at most once, or $\omega$, meaning it can be used with no constraints. Note that the first version, clearly useless in a functional calculus, could make sense adding effects, e.g. in an imperative calculus, playing a role similar to that of \lstinline{void}.

\begin{lstlisting}
class Pair { A$^1$ first; A$^1$ second; 
  A$^0$ getFirstZero() [$1$]{this.first}
  A$^1$ getFirstAffine() [$1$]{this.first}
  A$^\omega$ getFirst() [$1$]{this.first}
}
\end{lstlisting}
The coeffect of $\this$ is $1$ in  all  versions, and it is actually used once in the bodies. The occurrence of $\this$ in the bodies can get any non-zero grade thanks to rule \refToRule{t-var}, and fields are graded $1$, meaning that a field access does not affect the grade of the receiver, hence the three bodies can get any non-zero grade as well, so they are well-typed with respect to the grade in the method return type.

In the client code below, a call of the getter is assigned to a local variable of the same grade, which is then used consistently with such grade. 

\begin{lstlisting}  
Pair$^1$ p = ...

{A$^0$ a = p.getFirstZero(); new Pair(new A(),new A())}
{A$^1$ a = p.getFirstAffine(); new Pair(a,new A())}
{A$^\omega$ a = p.getFirst(); new Pair(a,a)}
\end{lstlisting}

The following blocks are, instead, ill-typed, for two different reasons.

\begin{lstlisting}
{A$^1$ a = p.getFirst(); new Pair(a,a)}
{A$^\omega$ a = p.getFirstAffine(); new Pair(a,a)}
\end{lstlisting}

In the first one, the initialization is correct, by subsumption, since we use an expression of a less constrained grade.  However, the variable is then used in a way which is not compatible with its grade.
In the second one, instead, the variable is used consistently with its grade, but the initialization is ill-typed, since we use an expression of a more constrained grade.

Finally, note that the coeffect of $\this$ could be safely changed to be $\omega$ in the  three  methods, providing an overapproximated information; in this case, however, the  three  invocations in the client code would be wrong, since the receiver $\ptt$ is required to be used at most once. 
\end{example}

\begin{example}\label{ex:ex2}
 Consider the following source (that is, non-annotated) version  of the expression in \cref{ex:ex2-sem}. 
\begin{lstlisting}
{A$^\public$ y = new A(); {A$^\private$ x = y; x}}
\end{lstlisting}
The $\private$ variable $\xtt$ is initialized with the  $\public$ expression/variable $\ytt$. The  block  expression has type $\Att^\private$ as the following type derivation shows.
\[
\NamedRule{t-block}{
\NamedRule{t-new}{}{
\IsWFExp{}{\ConstrCall{\Att}{}}{\Att^\publicS} }{}\Space\Space
\Deriv
{}}
{
\IsWFExp{}{\Block{\Att^\publicS}{\ytt}{\ConstrCall{\Att}{}}{\Block{\Att^\privateS}{\xtt}{\ytt}{\xtt}}}{\Att^\privateS}
}
{}
\]
where $\Deriv$ is the following derivation
\[
\NamedRule{t-block}{
\NamedRule{t-sub}{
\NamedRule{t-var}{}{\IsWFExp{\VarTypeCoeffect{\ytt}{\Att}{\publicS}}{\ytt}{\Att^\publicS}}{}
}{\IsWFExp{\VarTypeCoeffect{\ytt}{\Att}{\publicS}}{\ytt}{\Att^\privateS}}{}
\ \ \NamedRule{t-var}{}{\IsWFExp{\VarTypeCoeffect{\ytt}{\Att}{\publicS},\VarTypeCoeffect{\xtt}{\Att}{\privateS}}{\xtt}{\Att^\privateS}}{}
}
{
\IsWFExp{\VarTypeCoeffect{\ytt}{\Att}{\publicS}}{\Block{\Att^\privateS}{\xtt}{\ytt}{\xtt}}{\Att^\privateS} }{}
\]
 On the other hand, initializing  a $\public$ variable with a $\private$ expression as in
\begin{lstlisting}
{A$^\private$ y = new A(); {A$^\public$ x = y; x}}
\end{lstlisting}
is not possible, as expected, since $\IsNotWFExp{\VarTypeCoeffect{\ytt}{\Att}{\privateS}}{\ytt}{\Att^\publicS}$ .

Consider now the class $\Btt$ with a $\private$ field and a $\public$ one.
\begin{lstlisting}
class B { A$^\public$ f1; A$^\private$ f2; }
\end{lstlisting}
The expression $\e$
\begin{lstlisting}
{A$^\public$ x = new A();  new B(x,x)}
\end{lstlisting}
can be given type $\Btt^\public$ as follows:
\begin{small}
\[
\NamedRule{t-block}{
\NamedRule{t-new}{}{
\IsWFExp{}{\ConstrCall{\Att}{}}{\Att^\publicS} }{}\ \ 
\NamedRule{t-new}{
\NamedRule{t-var}{}
{\IsWFExp{\VarTypeCoeffect{\xtt}{\Att}{\publicS}}{\xtt}{\Att^\publicS}}{}\ 
\NamedRule{t-sub}{
\NamedRule{t-var}{}{\IsWFExp{\VarTypeCoeffect{\xtt}{\Att}{\publicS}}{\xtt}{\Att^\publicS}}{}
}{\IsWFExp{\VarTypeCoeffect{\xtt}{\Att}{\publicS}}{\xtt}{\Att^\privateS}}{}
}
  {\IsWFExp{\VarTypeCoeffect{\xtt}{\Att}{\publicS}}{\ConstrCall{\Btt}{\xtt,\xtt}}{\Btt^\publicS}}{}
}
{\IsWFExp{}{\Block{\Att^\publicS}{\xtt}{\ConstrCall{\Att}{}}{\ConstrCall{\Btt}{\xtt,\xtt}}}{\Btt^\publicS}
}{}
\]
\end{small}
By \refToRule{t-sub} we can also derive $\IsWFExp{}{\e}{\Btt^\privateS}$ and so we get 
\[
\begin{array}{c}
\NamedRule{t-field}{\IsWFExp{}{\e}{\Btt^\privateS}}
{\IsWFExp{}{\FieldAccess{\e}{\fO}}{\Att^\privateS}}{}
\Space\Space\Space\Space
\NamedRule{t-field}{\IsWFExp{}{\e}{\Btt^\publicS}}
{\IsWFExp{}{\FieldAccess{\e}{\fT}}{\Att^\privateS}}{}
\end{array}
\]
that is, accessing a $\public$ field  of  a $\private$ expression we get a $\private$ result as well as accessing a
$\private$ field  of  a $\public$ expression.\\
Also note that the following expression $\e'$ 
\begin{lstlisting}
{A$^\private$ x = new A();  new B(x,x)}
\end{lstlisting}
can be given only type $\Btt^\private$ by 
\[
\NamedRule{t-block}{
\NamedRule{t-new}{}{
\IsWFExp{}{\ConstrCall{\Att}{}}{\Att^\privateS} }{}\ \ 
\NamedRule{t-new}{
\NamedRule{t-var}{}
{\IsWFExp{\VarTypeCoeffect{\xtt}{\Att}{\privateS}}{\xtt}{\Att^\privateS}}{}\ 
\NamedRule{t-var}{}
{\IsWFExp{\VarTypeCoeffect{\xtt}{\Att}{\privateS}}{\xtt}{\Att^\privateS}}{}\ 
}
  {\IsWFExp{\VarTypeCoeffect{\xtt}{\Att}{\privateS}}{\ConstrCall{\Btt}{\xtt,\xtt}}{\Btt^\privateS}}{}
}
{\IsWFExp{}{\Block{\Att^\privateS}{\xtt}{\ConstrCall{\Att}{}}{\ConstrCall{\Btt}{\xtt,\xtt}}}{\Btt^\privateS}
}{}
\]
We cannot derive $\IsWFExp{}{\e'}{\Btt^\publicS}$, since the grade of $\fO$ is $\public$ and \refToRule{t-new} 
 would require $\IsWFExp{\VarTypeCoeffect{\xtt}{\Att}{\privateS}}{\xtt}{\Att^{\publicS\rmul\publicS}}$,
which does not hold.
\end{example}

\subsection{Resource-aware soundness}

 We \bpa state \epa that the graded type system is sound with respect to the resource-aware semantics.
In other words, the graded type system prevents both standard typing errors, such as invoking a missing field or method, and resource-usage errors, such as requiring a resource which is exhausted (cannot be used in the needed way). 

In order to state and prove a soundness theorem, we need to introduce  a (straightforward) typing judgment $\vdash_{a}$ for annotated expressions, environments and configurations. 
The typing rules  are reported in \cref{fig:type-aexp}. 

\begin{figure}
\begin{math}
\begin{array}{l}
\NamedRule{t-sub}
{\IsWFExpA{\Gamma}{\e}{\T}}
{\IsWFExpA{\Gamma'}{\e}{\T'}}
{ \Gamma\ROrdCtx\Gamma' \\ 
  \T\leq\T' } 
\BigSpace
\NamedRule{t-var}
{}
{ \IsWFExpA{\VarTypeCoeffect{\x}{\C}{\rGrad}}{\x}{\typeGrad{\C}{\rGrad}} } 
{ \rGrad\neq\rzero} 
\\[4ex]
\NamedRule{t-field-access}
{ \IsWFExpA{\Gamma}{\e}{\typeGrad{\C}{\rGrad}} }
{ \IsWFExpA{\Gamma}{\FieldAccess{\annotated{\e}{\rGrad}}{\f_i}}{\typeGrad{\C_i}{\rGrad\rmul\rGrad_i}} }
{ \fields{\C}=\Field{\typeGrad{\C_1}{\rGrad_1}}{\f_1} \ldots \Field{\typeGrad{\C_n}{\rGrad_n}}{\f_n} } 
\\[4ex]
\NamedRule{t-new}
{\IsWFExpA{\Gamma_i}{\e_i}{\typeGrad{\C_i}{\rGrad\rmul\rGrad_i}} \Space \forall i\in 1..n}
{\IsWFExpA{\Gamma_1\SumCtx\ldots\SumCtx\Gamma_n}{\ConstrCallATuple{\C}{\e}{\rGrad}{n}}{\typeGrad{\C}{\rGrad}} }
{\fields{\C}=\Field{\typeGrad{\C_1}{\rGrad_1}}{\f_1} \ldots \Field{\typeGrad{\C_n}{\rGrad_n}}{\f_n}} 
\\[4ex]

\NamedRule{t-invk}
{ \IsWFExpA{\Gamma_0}{\e_0}{\typeGrad{\C}{\rGrad_0}}
  \Space
  \IsWFExpA{\Gamma_i}{\e_i}{\typeGrad{\C_i}{\rGrad_i}} \Space \forall 1\in 1..n }
{ \IsWFExpA{\Gamma_0\SumCtx  \ldots \SumCtx \Gamma_n}{\MethCallATuple{\annotated{\e_0}{\rGrad_0}}{\m}{\e}{\rGrad}{n}}{\T} }
{ \mtype{\C}{\m}=\funType{\rGrad_0,\typeGrad{\C_1}{\rGrad_1} \ldots \typeGrad{\C_n}{\rGrad_n}}{\T}\\ }
\\[4ex]

\NamedRule{t-block}
{ \IsWFExpA{\Gamma_1}{\e_1}{\typeGrad{\C}{\rGrad}} 
  \Space
  \IsWFExpA{\Gamma_2, \VarTypeCoeffect{\x}{\C}{\rGrad}}{\e_2}{\T} }
{ \IsWFExpA{\Gamma_1 \SumCtx \Gamma_2}{\Block{\C}{\x}{\annotated{\e_1}{\rGrad}}{\e_2}}{\T}  }
{}
\\[4ex] 
 
\NamedRule{t-env}
{ \IsWFExpA{}{\val_i}{\typeGrad{\C_i}{\rGrad_i}} \Space \forall i \in 1..n }
{ \IsWFEnvA{\Gamma}{\stack}  }
{ \Gamma = \VarTypeCoeffect{\x_1}{\C_1}{\rGrad_1},\ldots,\VarTypeCoeffect{\x_n}{\C_n}{\rGrad_n} \\ 
  \stack = \x_1\mapsto\Pair{\val_1}{\rGrad_1},\ldots,\x_n\mapsto\Pair{\val_n}{\rGrad_n} } 
 
\\[4ex]
\NamedRule{t-conf}
{ \IsWFExpA{\Delta}{\e}{\T} \Space \IsWFEnvA{\Gamma}{\stack} }
{ \IsWFConfA{\Gamma}{\e}{\stack}{\T} }
{ \Delta\ROrdCtx\Gamma } 
\end{array}
\end{math}
\caption{Graded type system for annotated syntax}
\label{fig:type-aexp} 
\end{figure}

 Recall that $\erase{\_}$ denotes erasing annotations. 
 It is easy to see that an annotated expression is well-typed if and only if it is produced by the type system: 
\begin{proposition}\label{prop:wt-ann}
$\IsWFExpI\Gamma\e\T{\e'}$ if and only if 
$\erase{\e'} = \e$ and $\IsWFExpA\Gamma{\e'}\T$. 
\end{proposition}
A similar property  holds  for environments and configurations.

The main result is the following resource-aware progress theorem.

\begin{theorem}[Resource-aware progress] \label{thm:res-progress}
If $\IsWFConfA\Gamma\e\stack{\typeGrad\C\rGrad}$  then 
either $\e$ is a value or 
$\reducestack{\e}{\stack}{\rGrad}{\e'}{\stack'}$ and 
$\IsWFConfA{\Gamma'}{\e'}{\stack'}{\typeGrad\C\rGrad}$ with $\dom\Gamma\subseteq\dom{\Gamma'}$ and $\Gamma'\ROrdCtx\Gamma,\Delta$. 
\end{theorem}

 When reduction is non-deterministic,  we can distinguish two flavours of soundness,   \emph{soundness-must}  meaning that no computation can be stuck, and \emph{soundness-may}, meaning that at least one computation is not stuck.
The terminology of \emph{may} and \emph{must} properties is very general and comes originally from \cite{DeNicolaH84}; the specific names soundness-may and soundness-must  were  introduced in \cite{DagninoBZD20,Dagnino22} in the context of  big-step semantics. 
In our case, graded reduction is non-deterministic  since, as discussed before,    the rule \refToRule{var}  could be instantiated in different ways, possibly consuming the resource more than necessary. 
However, we expect that, for a well-typed configuration, there is at least one computation which is not stuck, hence a soundness-may result. Soundness-may can be proved  by a theorem like  the one above,  which can be seen as a subject-reduction-may result, including standard progress. In our case, if the configuration is well-typed, that is, annotations have been generated by the type system, 
\emph{there is a step} which leads, in turn, to a well-typed configuration. More in detail, the type is preserved, resources initially available  may  have reduced grades, and other available resources may be added.

To prove this result, we need some standard lemmas.

\begin{lemma}[Environment typing]\label{lem:env-inv}
The following facts hold
\begin{enumerate}
\item\label{lem:env-inv:1} If $\IsWFEnvA{\Gamma}{\stack}$ and $\Gamma = \Gamma',\VarTypeCoeffect\x\C\rGrad$, then 
$\stack = \stack',\x\mapsto\Pair\val\rGrad$ and $\IsWFEnvA{\Gamma'}{\stack'}$ and $\IsWFExp{}{\val}{\typeGrad\C\rGrad}$. 
\item\label{lem:env-inv:2} If $\IsWFEnvA{\Gamma}{\stack}$ and $\x\notin\dom\stack$ and $\IsWFExpA{}{\val}{\typeGrad\C\rGrad}$, then 
$\IsWFEnvA{\Gamma,\VarTypeCoeffect\x\C\rGrad}{\stack,\x\mapsto\Pair\val\rGrad}$. 
\end{enumerate}
\end{lemma}

\begin{lemma}[Strengthening for values]\label{lem:val_empty}
If $\IsWFExpA\Gamma\val\T$, then $\IsWFExpA{}{\val}{\T}$. 
\end{lemma}

\begin{lemma}[Canonical Forms]\label{lem:cf}
If $\IsWFExpA{\Gamma}{\val}{\typeGrad\C\rGrad}$ then 
$\val  = \ConstrCallATuple{\D}{\val}{\rGrad}{n}$ and 
$\D\leq\C$, $\fields{\D} = \Field{\typeGrad{\D_1}{\rGrad_1}}{\f_1}\ldots\Field{\typeGrad{\D_n}{\rGrad_n}}{\f_n}$ and 
$\Gamma_1 \SumCtx\ldots\SumCtx\Gamma_m \ROrdCtx \Gamma$ and 
$\IsWFExpA{\Gamma_i}{\val_i}{\typeGrad{\D_i}{\sGrad\rmul\rGrad_i}}$, for all $i \in 1..m$, with $\rGrad\rord\sGrad$. 
\end{lemma} 

\begin{lemma}[Renaming]\label{lem:subst}
If $\IsWFExpA{\Gamma,\VarTypeCoeffect\x\C\rGrad}\e\T$ and $\y\notin\dom\Gamma$, then $\IsWFExpA{\Gamma,\VarTypeCoeffect\y\C\rGrad}{\Subst\e\y\x}{\T}$. 
\end{lemma}

\begin{lemma}\label{lem:sum-comma}
If $\dom\Theta\cap\dom\Delta = \emptyset$ then 
$(\Gamma,\Theta)\SumCtx\Delta = (\Gamma\SumCtx\Delta),\Theta$. 
\end{lemma}

 \cref{thm:res-progress} is proved as a special case of the following more general result,  which makes explicit the invariant needed to carry out the induction. 
Indeed, 
by looking at the reduction rules, we can see that computational ones either add new variables to the environment or reduce the grade of a variable of some amount that depends on the grade of the reduction. 
In the latter case, the amount can be arbitrarily chosen with the only restrictions that it is non zero and at least the grade of the reduction. 
However, to prove progress, we not only have to prove that a reduction can be done, but, if the reduction is done in a context, say evaluating 
the argument of a constructor, then after such reduction we still have enough resources to go on with the reduction, that is, to evaluate the 
rest of the context (the other arguments of the constructor). 
This means that the resulting environment has enough resources to  type the whole context (the constructor call). 
For this reason, in the statement of the theorem that follows, 
we add to the assumption of \cref{thm:res-progress} a typing context $\Theta$ that would contain the information on the amount of resources that we want to preserve during the reduction (see \cref{thm:res-progress-gen:4}  of the theorem). 
This allows us to choose the appropriate grade to be kept  when reducing a variable and to reconstruct a typing derivation when using contextual reduction rules. 
For the expression at the top level, as we see from the proof of \cref{thm:res-progress}, $\Theta$ is simply $\rzero$ for all variables in the typing context in which the expression is typed.
 
\begin{theorem}\label{thm:res-progress-gen} 
If $\IsWFExpA{\Delta}{\e}{\typeGrad\C\rGrad}$ and $\IsWFEnvA{\Gamma}{\stack}$ and $\Delta\SumCtx\Theta\ROrdCtx\Gamma$ and $\dom{\Delta} \subseteq \dom{\Theta}$ and $\e$ is not a value,
then there are $\e'$, $\stack'$, $\Delta'$, $\Gamma'$ and $\Theta'$  such that 
\begin{enumerate}
\item\label{thm:res-progress-gen:1} $\reducestack{\e}{\stack}{\rGrad}{\e'}{\stack'}$ and 
\item\label{thm:res-progress-gen:2} $\IsWFExpA{\Delta'}{\e'}{\typeGrad\C\rGrad}$ with $\Delta'\ROrdCtx\Delta,\Theta'$ and 
\item\label{thm:res-progress-gen:3} $\IsWFEnvA{\Gamma'}{\stack'}$ with $\Gamma'\ROrdCtx\Gamma,\Theta'$ and 
\item\label{thm:res-progress-gen:4} $\Delta'\SumCtx \Theta \ROrdCtx \Gamma'$. 
\end{enumerate}
\end{theorem} 
 
\begin{proof}  
The proof is by induction on typing rules. 
\begin{description}

\item [\refToRule{t-var}] 
By hypothesis, we know that 
$\Delta = \VarTypeCoeffect{\x}{\C}{\rGrad}$ and, with $\rGrad\ne\rzero$. 
Since $\dom\Delta\subseteq\dom\Theta$, we have $\Theta = \Theta_1,\VarTypeCoeffect{\x}{\C}{\sGrad'}$ and, 
since $\Delta\SumCtx\Theta\ROrdCtx\Gamma$, we have $\Gamma = \Gamma_1,\VarTypeCoeffect\x\C\sGrad$ with $\rGrad\rsum\sGrad'\rord\sGrad$.
Moreover, because $\IsWFEnvA{\Gamma}{\stack}$, by \refItem{lem:env-inv}{1}, we have that 
$\stack = \stack',\x\mapsto\Pair\val\sGrad$, $\IsWFEnvA{\Gamma_1}{\stack'}$ and $\IsWFExpA{}{\val}{\typeGrad\C\sGrad}$. 
Then, by rule \refToRule{var}, we get 
$\reducestack{\x}{\stack}{\rGrad}{\val}{\stack',\x\mapsto\Pair{\val}{\sGrad'}}$. 
Since $\sGrad' \rord \rGrad \rsum \sGrad' \rord \sGrad$, by rule \refToRule{t-sub}, we get 
$\IsWFExpA{}{\val}{\typeGrad\D{\sGrad'}}$ and by \refItem{lem:env-inv}{2} we conclude 
$\IsWFEnvA{\Gamma_1,\VarTypeCoeffect\x\C{\sGrad'}}{\stack',\x\mapsto\Pair\val{\sGrad'}}$. 
Since $\rGrad\rord\rGrad\rsum\sGrad'\rord\sGrad$, again by rule \refToRule{sub}, we get 
$\IsWFExpA{}{\val}{\typeGrad\C\rGrad}$.
Now, let us set 
$\Delta' = \Theta' = \emptyset$ and $\Gamma' = \Gamma_1,\VarTypeCoeffect\x\C{\sGrad'}$. 
We immediately get $\Delta'\ROrdCtx\Delta,\Theta'$ and $\Gamma'\ROrdCtx\Gamma,\Theta'$. 
Furthermore, we get 
$\Delta'\SumCtx\Theta\ROrdCtx\Gamma'$ as 
$\Delta' \SumCtx\Theta = \Theta = \Theta_1,\VarTypeCoeffect\x\C{\sGrad'}$  and, 
from $\Theta\ROrdCtx\Delta\SumCtx\Theta \ROrdCtx\Gamma = \Gamma_1,\VarTypeCoeffect\x\C\sGrad$, we get 
$\Theta_1\ROrdCtx\Gamma_1$, therefore 
$\Theta_1,\VarTypeCoeffect\x\C{\sGrad'} \ROrdCtx \Gamma_1,\VarTypeCoeffect\x\C{\sGrad'} = \Gamma'$, as needed. 

\item [\refToRule{t-sub}] 
By hypothesis we know that 
$\IsWFExpA{\Delta_1}{\e}{\typeGrad{\D}{\sGrad}}$ with $\Delta_1\ROrdCtx\Delta$ and $\typeGrad\D\sGrad\leq \typeGrad\C\rGrad$, which implies $\D\leq\C $ and $\rGrad\rord\sGrad$. 
We distinguish two cases. 
\begin{itemize}
\item If $\e$ is a value then we have the thesis
\item Otherwise, notice that $\Delta_1\SumCtx\Theta \ROrdCtx\Delta\Theta\ROrdCtx \Gamma$ holds by monotonicity of $\SumCtx$; 
then, by induction hypothesis, we have
\begin{enumerate} 
\item \label{itm:sub:1} $\reducestack{\e}{\stack}{\sGrad}{\e'}{\stack'}$ and 
\item \label{itm:sub:2} $\IsWFExpA{\Delta'}{\e'}{\typeGrad\D\sGrad}$ with $\Delta'\ROrdCtx\Delta_1,\Theta'$ and 
\item \label{itm:sub:3} $\IsWFEnvA{\Gamma'}{\stack'}$ with $\Gamma'\ROrdCtx\Gamma,\Theta'$ and 
\item \label{itm:sub:4} $\Delta'\SumCtx \Theta\ROrdCtx\Gamma'$.
\end{enumerate}
By \cref{itm:sub:1,prop:red-gr}, since $\rGrad\rord\sGrad$,  we get 
$\reducestack{\e}{\stack}{\rGrad}{\e'}{\stack'}$. 
Since $\typeGrad\D\sGrad\leq\typeGrad\C\rGrad$, from \cref{itm:sub:2}, by rule \refToRule{t-sub}, we get 
$\IsWFExpA{\Delta'}{\e'}{\typeGrad\C\rGrad}$ and, 
Since $\Delta_1\ROrdCtx\Delta$, we get 
$\Delta'\ROrdCtx \Delta_1,\Theta'\ROrdCtx\Delta,\Theta'$, proving the thesis. 
\end{itemize}

\item [\refToRule{t-field-access}]
By hypothesis we know that 
$\IsWFExpA{\Delta}{\FieldAccess{\annotated{\e}{\rGrad}}{\f_i}}{\typeGrad{\C_i}{\rGrad\rmul\rGrad_i}}$, 
$\IsWFExpA{\Delta}{\e_1}{\typeGrad\C\rGrad}$ 
and $\fields{\C}=\Field{\typeGrad{\C_1}{\rGrad_1}}{\f_1} \ldots \Field{\typeGrad{\C_n}{\rGrad_n}}{\f_n}$. 
We distinguish two cases. 
\begin{itemize}
\item If $\e_1$ is a value, then, by \cref{lem:cf}, we have 
$\e_1 = \ConstrCallATuple{\D}{\val}{\sGrad}{m}$ and 
$\D\leq\C$, $\fields{\D} = \Field{\typeGrad{\D_1}{\sGrad_1}}{\f'_1}\ldots\Field{\typeGrad{\D_m}{\sGrad_m}}{\f'_m}$ and 
$\Delta_1 \SumCtx\ldots\SumCtx\Delta_m \ROrdCtx\Delta$ and 
$\IsWFExpA{\Delta_j}{\val_j}{\typeGrad{\D_j}{\sGrad\rmul\sGrad_j}}$, for all $j \in 1..m$, with $\rGrad\rord\sGrad$. 
By coherence conditions on the class table, 
we know that $n \leq m$  and, for all $j \in 1..n$, $\D_j = \C_j$ and $\sGrad_j = \rGrad_j$ and $\f'_j = \f_j$. 
Hence, since $i \in 1..n$, we have $i \in 1..m$ and so, by rule \refToRule{field-access}, 
we get $\reducestack{\FieldAccess{\annotated{\e_1}{\rGrad}}{\f_i}}{\stack}{\rGrad\rmul\rGrad_i}{\val_i}{\stack}$. 
Since $\rGrad\rord\sGrad$, we have $\rGrad\rmul\rGrad_i\rord\sGrad\rmul\rGrad_i = \sGrad\rmul\sGrad_i$ and $\D_i = \C_i$, 
hence, by rule \refToRule{t-sub}, we derive $\IsWFExpA{\Delta_i}{\val_i}{\typeGrad{\C_i}{\rGrad\rmul\rGrad_i}}$. 
Let us set $\Delta' = \Delta_i$, $\Gamma' = \Gamma$ and $\Theta' = \emptyset$. 
Then, the thesis trivially follows as $\Delta' = \Delta_i \ROrdCtx \sum_{j = 1}^m \Delta_j \ROrdCtx \Delta$. 
\item Otherwise, 
by induction hypothesis, 
we get 
\begin{enumerate} 
\item \label{itm:acc:1} $\reducestack{\e_1}{\stack}{\rGrad}{\e_1'}{\stack'}$ and 
\item \label{itm:acc:2} $\IsWFExpA{\Delta'}{\e_1'}{\typeGrad\C\rGrad}$ with $\Delta'\ROrdCtx\Delta,\Theta'$ and 
\item \label{itm:acc:3} $\IsWFEnvA{\Gamma'}{\stack'}$ with $\Gamma'\ROrdCtx\Gamma,\Theta'$ and 
\item \label{itm:acc:4} $\Delta'\SumCtx \Theta\ROrdCtx \Gamma'$. 
\end{enumerate}
Let us set 
$\e' = \FieldAccess{\annotated{\e_1'}{\rGrad}}{\f_i}$. 
By \cref{itm:acc:1} and the hypothesis, using rule \refToRule{field-access-ctx}, we derive $\reducestack{\e}{\stack}{\rGrad\rmul\rGrad_i}{\e'}{\stack'}$. 
By \cref{itm:acc:2} and the hypothesis, using rule \refToRule{t-field-access}, we get $\IsWFExpA{\Delta'}{\e'}{\typeGrad{\C}{\rGrad\rmul\rGrad_i}}$. 
Finally, by \cref{itm:acc:3,itm:acc:4}, we get the thesis. 
\end{itemize}

\item [\refToRule{t-new}] 
By hypothesis we have $\Delta = \Delta_1\SumCtx\ldots\SumCtx\Delta_n$ and 
$\IsWFExpA{\Delta_i}{\e_i}{\typeGrad{\C_i}{\rGrad\rmul\rGrad_i}}$, for all $i \in 1..n$. 
We distinguish two cases. 
\begin{itemize}
\item If $\e_i$ is a value for all $i \in 1..n$, then 
$\e$ is a value as well and this proves the thesis. 
\item Otherwise, there is an $i \in 1..n$ such that $\e_i$ is not a value while $\e_j$ is a value, for all $j \in 1..(i-1)$.
Let us set $\widehat\Theta = \Theta \SumCtx \sum_{j = 1}^{i-1} \Delta_j \SumCtx \sum_{j = i+1}^n \Delta_j$.
Since $\dom{\Delta_i}\subseteq\dom{\Delta}\subseteq\dom\Theta$, we have $\dom{\Delta_i} \subseteq \dom{\widehat\Theta}$ and, by construction, 
we have $\Delta_i\SumCtx\widehat\Theta = \Delta\SumCtx\Theta \ROrdCtx\Gamma$. 
Then, by induction hypothesis, 
we get 
\begin{enumerate} 
\item \label{itm:new:1} $\reducestack{\e_i}{\stack}{\rGrad\rmul\rGrad_i}{\e_i'}{\stack'}$
\item \label{itm:new:2} $\IsWFExpA{\Delta_i'}{\e_i'}{\typeGrad\C\rGrad\rmul\rGrad_i}$ with $\Delta_i'\ROrdCtx\Delta_i,\Theta'$ and 
\item \label{itm:new:3} $\IsWFEnvA{\Gamma'}{\stack'}$ with $\Gamma'\ROrdCtx\Gamma,\Theta'$ and 
\item \label{itm:new:4} $\Delta_i'\SumCtx \widehat\Theta\ROrdCtx \Gamma'$. 
\end{enumerate}
Let us set 
$\e' = \ConstrCall{\C}{\annotated{\val_1}{\rGrad_1},\dots,\annotated{\val_{i-1}}{\rGrad_{i-1}},\annotated{\e'_i}{\rGrad_i},\ldots,\annotated{\e_n}{\rGrad_n}}$ and 
$\Delta' = \sum_{j = 1}^{i-1} \Delta_j \SumCtx\Delta'_i\SumCtx\sum_{j = i+1}^n \Delta_j$. 
By \cref{itm:new:1} and the hypothesis, using rule \refToRule{new-ctx}, we derive $\reducestack{\e}{\stack}{\rGrad}{\e'}{\stack'}$. 
By \cref{itm:new:2} and the hypothesis, using rule \refToRule{t-new}, we get $\IsWFExpA{\Delta'}{\e'}{\typeGrad\C\rGrad}$. 
By \cref{itm:new:3}, we have $\dom{\Theta'}\cap\dom\Gamma = \emptyset$ and, since $\Delta_j \ROrdCtx \Delta\SumCtx\Theta\ROrdCtx\Gamma$ for all $j \in 1..n$, we get $\dom{\Delta_j}\cap\dom{\Theta'} = \emptyset$. 
Hence, by monotonicity of $\SumCtx$ and \cref{lem:sum-comma}, we get 
\begin{align*} 
\Delta'
  &\ROrdCtx  \sum_{j = 1}^{i-1} \Delta_j \SumCtx (\Delta_i,\Theta') \SumCtx \sum_{j = i+1}^n \Delta_j 
   = \left(\sum_{j = 1}^n \Delta_j\right),\Theta' 
   = \Delta,\Theta'
\end{align*} 
Finally, since $\Delta'\SumCtx\Theta = \Delta'_i\SumCtx\widehat\Theta$, by \cref{itm:new:4} we get the thesis. 
\end{itemize}
\item [\refToRule{t-invk}] 
By hypothesis we know that 
$\MethCall{\annotated{\e_0}{\rGrad_0}}{\m}{\annotated{\e_1}{\rGrad_1},\ldots,\annotated{\e_n}{\rGrad_n}}$, 
$\Delta = \Delta_0\SumCtx\ldots\SumCtx\Delta_n$, 
$\mtype{\C_0}{\m} = \funType{\rGrad_0,\typeGrad{\C_1}{\rGrad_1} \ldots \typeGrad{\C_n}{\rGrad_n}}{\typeGrad\C\rGrad}$ and 
$\IsWFExpA{\Delta_i}{\e_i}{\typeGrad {\C_i}{\rGrad_i}}$, for all $i \in 0..n$. 
Then, we distinguish two cases. 
\begin{itemize}
\item If, for all $i \in 0..n$, $\e_i$ is a value, say $\e_i  \val_i$, 
by \cref{lem:cf}, we have $\val_0 = \ConstrCall{\D}{\_}$ with $\D\leq\C_0$ and, 
by coherence conditions on the class table, we have 
$\mbody{\D}{\m}=\Pair{\x_1\dots\x_n}{\e'}$. 
Then, by rule \refToRule{invk},  we get 
$\reducestack{\e}{\stack}{\rGrad}{\SubstDots{\Subst{\e'}{\y_0}{\this}}{\y_1}{\x_1}{\y_n}{\x_n}}{\stack'}$ 
with $\y_0,\ldots,\y_n\notin\dom\stack$ and 
$\stack' = \stack,\y_0\mapsto\Pair{\val_0}{\rGrad_0},\ldots,\y_n\mapsto\Pair{\val_n}{\rGrad_n}$. 
By \cref{lem:val_empty}, we get $\IsWFExpA{}{\val_i}{\typeGrad{\C_i}{\rGrad_i}}$, for all $i \in 0..n$, hence, by \refItem{lem:env-inv}{2}, we get 
$\IsWFEnvA{\Gamma, \VarTypeCoeffect{\y_0}{\C_0}{\rGrad_0}, \ldots, \VarTypeCoeffect{\y_n}{\C_n}{\rGrad_n}}{\stack'}$. 
By condition \refToRule{t-meth} and rule \refToRule{t-sub}, we know that 
$\IsWFExpA{\VarTypeCoeffect{\this}{\C_0}{\rGrad_0}, \VarTypeCoeffect{\x_1}{\C_1}{\rGrad_1},\ldots,\VarTypeCoeffect{\x_n}{\C_n}{\rGrad_n}}{\e'}{\typeGrad\C\rGrad}$. 
Let us set 
$\Delta' = \Theta' = \VarTypeCoeffect{\y_0}{\C_0}{\rGrad_0},\ldots,\VarTypeCoeffect{\y_n}{\C_n}{\rGrad_n}$,  and $\Gamma' = \Gamma,\Theta'$, 
hence, we immediately have 
$\Gamma'\ROrdCtx\Gamma,\Theta'$ and $\Delta' = \Theta'\ROrdCtx\Delta,\Theta'$. 
By \cref{lem:subst}, we have $\IsWFExpA{\Delta'}{\SubstDots{\Subst{\e'}{\y_0}{\this}}{\y_1}{\x_1}{\y_n}{\x_n}}{\typeGrad\C\rGrad}$. 
Finally,
since $\Theta\ROrdCtx\Delta\SumCtx\Theta \ROrdCtx\Gamma$ and $\y0,\ldots,\y_n\notin\dom\Gamma$ imply $\y_0,\ldots,\y_n\notin\dom\Theta$, 
we get 
$\Delta'\SumCtx\Theta = \Theta,\Delta'\ROrdCtx (\Delta\SumCtx\Theta),\Delta' \ROrdCtx \Gamma,\Delta' = \Gamma'$. 
\item Otherwise, there is an $i \in 0..n$ such that $\e_i$ is not a value while $\e_j$ is a value, for all $j \in 0..(i-1)$.
Let us set $\widehat\Delta = \sum_{j=0}^{i-1} \Delta_j \SumCtx \sum_{j = i+1}^n \Delta_j$ and $\widehat\Theta = \Theta\SumCtx\widehat\Delta$. 
so $\Delta_i\SumCtx \widehat\Theta = \Delta\SumCtx\Theta \ROrdCtx \Gamma$.
By induction hypothesis, we have 
\begin{enumerate}
\item \label{itm:invk:1} $\reducestack{\e_i}{\stack}{\rGrad_i}{\e_i'}{\stack'}$
\item \label{itm:invk:2} $\IsWFExpA{\Delta_i'}{\e_i'}{\typeGrad{\C_i}{\rGrad_i}}$ with $\Delta_i'\ROrdCtx\Delta_i,\Theta'$ and 
\item \label{itm:invk:3} $\IsWFEnvA{\Gamma'}{\stack'}$ with $\Gamma'\ROrdCtx\Gamma,\Theta'$ and 
\item \label{itm:invk:4} $\Delta_i'\SumCtx \widehat\Theta \ROrdCtx \Gamma'$. 
\end{enumerate} 
Let us set $\e' = \MethCall{\annotated{\e_0}{\rGrad_0}}{\m}{\ldots,\annotated{\e_{i-1}}{\rGrad_{i-1}},\annotated{\e'_i}{\rGrad_i}, \annotated{\e_{i+1}}{\rGrad_{i+1}},\ldots,\annotated{\e_n}{\rGrad_n}}$ and 
$\Delta' = \Delta'_i \SumCtx \widehat\Delta$. 
By \cref{itm:invk:1} and the hypothesis, using either rule \refToRule{invk-rcv-ctx} or \refToRule{invk-args-ctx}, depending on whether $i = 0$ or not, 
we derive $\reducestack{\e}{\stack}{\rGrad}{\e'}{\stack'}$. 
By \cref{itm:invk:2} and the hypothesis, using rule \refToRule{t-invk}, we have $\IsWFExpA{\Delta'}{\e'}{\typeGrad\C\rGrad}$. 
By \cref{itm:invk:3}, we have $\dom{\Theta'}\cap\dom\Gamma = \emptyset$ and, 
since $\widehat\Delta \ROrdCtx \Delta\SumCtx\Theta\ROrdCtx\Gamma$, we get $\dom{\widehat\Delta}\cap\dom{\Theta'} = \emptyset$. 
Hence, by \cref{itm:invk:2}, monotonicity of $\SumCtx$ and \cref{lem:sum-comma}, we get 
$\Delta' = \Delta'_i\SumCtx\widehat\Delta  \ROrdCtx (\Delta_i,\Theta')\SumCtx\widehat\Delta  = (\Delta_i\SumCtx\widehat\Delta),\Theta' = \Delta,\Theta'$. 
Finally, since 
$\Delta' \SumCtx\Theta = \Delta'_i \SumCtx \widehat\Delta \SumCtx\Theta = \Delta'_i\SumCtx\widehat\Theta$, by \cref{itm:invk:4} we get the thesis. 
\end{itemize}
\item [\refToRule{t-block}] 
By hypothesis we have $\Delta=\Delta_1\SumCtx\Delta_2$ and 
$\IsWFExpA{\Delta_1}{\e_1}{\typeGrad\D\sGrad}$ and 
$\IsWFExpA{\Delta_2,\VarTypeCoeffect\x\D\sGrad}{\e_2}{\typeGrad\C\rGrad}$. 
Then, we distinguish two cases. 
\begin{itemize}
\item If $\e_1 = \val$ is a value, then 
by rule \refToRule{block}, 
we have $\reducestack{\Block{\D}{\x}{\annotated{\val}{\sGrad}}{\e_2}}{\stack}{\rGrad}{\Subst{\e_2}{\y}{\x}}{\AddToMap{\stack}{\y}{\Pair{\val}{\sGrad}}}$ with $\y \notin \dom{\stack}$. 
By \cref{lem:val_empty}, we get $\IsWFExpA{}{\val}{\typeGrad\D\sGrad}$, hence, by \refItem{lem:env-inv}{2}, we get 
$\IsWFEnvA{\Gamma,\VarTypeCoeffect\y\D\sGrad}{\stack,\y\mapsto\Pair\val\sGrad}$. 
Notice that this implies $\y\notin\dom\Gamma$, thus $\y\notin\dom{\Delta_1}$ and $\y\notin\dom{\Delta_2}$. 
Then, let us set 
$\Theta' = \VarTypeCoeffect\y\D\sGrad$, $\Delta' = \Delta_2,\Theta'$ and $\Gamma' = \Gamma,\Theta'$, 
hence, we immediately have 
$\Gamma'\ROrdCtx\Gamma,\Theta'$ and $\Delta'=\Delta_2,\Theta'\ROrdCtx\Delta,\Theta'$, as $\Delta_2\ROrdCtx\Delta_1\SumCtx\Delta_2 = \Delta$. 
By \cref{lem:subst}, we have $\IsWFExpA{\Delta_2,\VarTypeCoeffect\y\D\sGrad}{\Subst{\e_2}\y\x}{\typeGrad\C\rGrad}$. 
Finally,
since $\Theta\ROrdCtx\Gamma$ and $\y\notin\dom\Gamma$ imply $\y\notin\dom\Theta$, 
by \cref{lem:sum-comma}, we get 
$\Delta'\SumCtx\Theta = (\Delta_2\SumCtx\Theta),\Theta' \ROrdCtx (\Delta\SumCtx\Theta),\Theta' \ROrdCtx \Gamma,\Theta' = \Gamma'$. 
\item Otherwise, 
let us set $\widehat\Theta = \Theta \SumCtx \Delta_2$.
Since $\dom{\Delta_1}\subseteq\dom{\Delta}\subseteq\dom\Theta$, we have $\dom{\Delta_1} \subseteq \dom{\widehat\Theta}$ and, by construction, 
we also have $\Delta_1\SumCtx\widehat\Theta = \Delta\SumCtx\Theta \ROrdCtx\Gamma$. 
By induction hypothesis, we have 
\begin{enumerate}
\item \label{itm:block:1} $\reducestack{\e_1}{\stack}{\sGrad}{\e_1'}{\stack'}$
\item \label{itm:block:2} $\IsWFExpA{\Delta_1'}{\e_1'}{\typeGrad\D\sGrad}$ with $\Delta_1'\ROrdCtx\Delta_1,\Theta'$ and 
\item \label{itm:block:3} $\IsWFEnvA{\Gamma'}{\stack'}$ with $\Gamma'\ROrdCtx\Gamma,\Theta'$ and 
\item \label{itm:block:4} $\Delta_1'\SumCtx \widehat\Theta \ROrdCtx \Gamma'$. 
\end{enumerate} 
Let us set $\e' = \Block{\D}{\x}{\annotated{\e_1'}{\sGrad}}{\e_2}$ and $\Delta' = \Delta'_1\SumCtx\Delta_2$. 
By \cref{itm:block:1} and the hypothesis, using rule \refToRule{block-ctx}, we derive $\reducestack{\e}{\stack}{\rGrad}{\e'}{\stack'}$. 
By \cref{itm:block:2} and the hypothesis, using rule \refToRule{t-block}, we have $\IsWFExpA{\Delta'}{\e'}{\typeGrad\C\rGrad}$. 
By \cref{itm:block:3}, we have $\dom{\Theta'}\cap\dom\Gamma = \emptyset$ and, since $\Delta_2 \ROrdCtx \Delta\SumCtx\Theta\ROrdCtx\Gamma$, we get $\dom{\Delta_2}\cap\dom{\Theta'} = \emptyset$. 
Hence, by monotonicity of $\SumCtx$ and \cref{lem:sum-comma}, we get 
$\Delta' \ROrdCtx (\Delta_1,\Theta')\SumCtx\Delta_2 = (\Delta_1\SumCtx\Delta_2),\Theta' = \Delta,\Theta'$. 
Finally, since 
$\Delta' \SumCtx\Theta = \Delta'_1\SumCtx\widehat\Theta$, by \cref{itm:block:4} we get the thesis. 
\end{itemize}

\qedhere 
\end{description}
\end{proof}

 We are now ready to prove \cref{thm:res-progress}. 
\begin{proof}[Proof of \cref{thm:res-progress}]
Inverting rule \refToRule{t-conf}, we get $\IsWFEnvA{\Gamma}{\stack}$ and $\IsWFExpA{\Delta}{\e}{\typeGrad\C\rGrad}$ with $\Delta\ROrdCtx\Gamma$. 
Applying \cref{thm:res-progress-gen} with $\Theta = \rzero\MulCtx\Delta$ we get 
\begin{enumerate}
\item\label{itm:res:1} $\reducestack{\e}{\stack}{\rGrad}{\e'}{\stack'}$ and 
\item\label{itm:res:2} $\IsWFExpA{\Delta'}{\e'}{\typeGrad\C\rGrad}$ with $\delta'\ROrdCtx\Delta,\Theta'$ and 
\item\label{itm:res:3} $\IsWFEnvA{\Gamma'}{\stack}$ with $\Gamma'\ROrdCtx \Gamma,\Theta'$ and 
\item\label{itm:res:4} $\Delta'\SumCtx\Theta\ROrdCtx\Gamma'$. 
\end{enumerate}
By \cref{itm:res:4}, we get $\Delta'\ROrdCtx\Gamma'$, hence by \refToRule{t-conf} and \cref{itm:res:2,itm:res:3}, we conclude  $\IsWFConfA{\Gamma'}{\e'}{\stack'}{\typeGrad\C\rGrad}$. 
Finally, by \cref{prop:step-env}, we have $\dom{\stack}\subseteq \dom{\stack'}$ and, 
since by rule \refToRule{t-env} and \cref{itm:res:2,itm:res:3} we know that $\dom\stack = \dom\Gamma$ and $\dom{\stack'} = \dom{\Gamma'}$, we get the thesis. 
\end{proof}

Using \cref{thm:res-progress} we can prove a resource-aware soundness  theorem. 
As already noticed, it is a form of soundness-may, that is, it states that a well-typed configuration either converges to a well-typed value or diverges. 
We write $\reducestackdiv\e\stack\rGrad$ when there exists an infinite sequence of steps in $\ev_\rGrad$ starting with $\ExpStack\e\stack$. 
Note that this judgement can be equivalently defined coinductively by the following rule: 
\[ \frac{ \reducestackdiv{\e'}{\stack'}{\rGrad} }{ \reducestackdiv{\e}{\stack}{\rGrad} }\ \reducestack\e\stack\rGrad{\e'}{\stack'} \] 
\begin{corollary}[Resource-aware soundness]\label{cor:res-sound}
If $\IsWFConfA{\Gamma}{\e}{\stack}{\typeGrad\C\rGrad}$ then 
either $\reducestackstar\e\stack\rGrad\val{\stack'}$ with $\IsWFConfA{\Gamma'}{\val}{\stack'}{\typeGrad\C\rGrad}$, 
or $\reducestackdiv\e\stack\rGrad$. 
\end{corollary}
\begin{proof}
We say that a well-typed configuration $\IsWFConfA{\Gamma}{\e}{\stack}{\typeGrad\C\rGrad}$ is well-converging if 
there are $\val$, $\stack'$, $\Gamma'$ and $\Delta$  such that 
$\reducestackstar\e\stack\rGrad\val{\stack'}$ and 
$\IsWFConfA{\Gamma'}{\val}{\stack'}{\typeGrad\C\rGrad}$ with $\Gamma'\ROrdCtx\Gamma,\Delta$ and $\dom\Gamma\subseteq\dom{\Gamma'}$. 
The statement is equivalent to the following: 
if $\IsWFConfA{\Gamma}{\e}{\stack}{\typeGrad\C\rGrad}$ and it is not well-converging, then  $\reducestackdiv\e\stack\rGrad$. 
We prove this by coinduction. 
Let us consider a well-typed configuration $\IsWFConfA{\Gamma}{\e}{\stack}{\typeGrad\C\rGrad}$ which is not well-converging. 
Then, by \cref{thm:res-progress}, we get $\reducestackstar\e\stack\rGrad{\e'}{\stack'}$  where 
$\IsWFConfA{\Gamma'}{\e'}{\stack'}{\typeGrad\C\rGrad}$ with 
$\dom\Gamma\subseteq{\Gamma'}$ and $\Gamma'\ROrdCtx\Gamma,\Delta$. 
To conclude the proof by coinduction, we just have to check that $\IsWFConfA{\Gamma'}{\e'}{\stack'}{\typeGrad\C\rGrad}$ is not well-converging. 
Suppose it is well-converging, then $\reducestackstar{\e'}{\stack'}{\rGrad}{\val}{\stack''}$  where 
$\IsWFConfA{\Gamma''}{\val}{\stack''}{\typeGrad\C\rGrad}$ with $\dom{\Gamma'}\subseteq \dom{\Gamma''}$ and $\Gamma'' \ROrdCtx\Gamma',\Delta'$. 
Therefore, we have 
$\ExpStack\e\stack \ev_\rGrad \reducestackstar{\e'}{\stack'}{\rGrad}\val{\stack''}$ and 
$\dom{\Gamma}\subseteq\dom{\Gamma'}\subseteq\dom{\Gamma''}$ and 
$\Gamma'' \ROrdCtx\Gamma',\Delta' \ROrdCtx\Gamma,\Delta,\Delta'$, 
proving that $\IsWFConfA{\Gamma}{\e}{\stack}{\typeGrad\C\rGrad}$ is well-converging, which is a contradiction, as needed. 
\end{proof}
 
Finally, the following corollary states both subject-reduction for the standard semantics, that is, type and coeffects are preserved, and completeness of the instrumented semantics, that is, for well-typed configurations, every reduction step in the usual semantics can be simulated by an appropriate step in the instrumented semantics.
 
\begin{corollary}[Subject reduction]\label{cor:sr}
If $\IsWFConfI{\Gamma_1}{\e_1}{\stack_1}{\typeGrad\C\rGrad}{\e'_1}{\stack'_1}$ and $\ExpStack{\e_1}{\stack_1} \ev \ExpStack{\e_2}{\stack_2}$, then 
$\IsWFConfI{\Gamma_2}{\e_2}{\stack_2}{\typeGrad\C\rGrad}{\e'_2}{\stack'_2}$ with $\dom{\Gamma_1}\subseteq\dom{\Gamma_2}$ and $\Gamma_2\ROrdCtx\Gamma_1,\Delta$, and $\reducestack{\e'_1}{\stack'_1}{\rGrad}{\e'_2}{\stack'_2}$. 
\end{corollary}
\begin{proof}
By \cref{prop:wt-ann} we get $\IsWFConfA{\Gamma_1}{\e'_1}{\stack'_1}{\typeGrad\C\rGrad}$ and, by \cref{thm:res-progress}, 
$\reducestack{\e'_1}{\stack'_1}{\rGrad}{\e'_2}{\stack'_2}$ and $\IsWFConfA{\Gamma_2}{\e'_2}{\stack'_2}{\typeGrad\C\rGrad}$ with $\dom{\Gamma_1}\subseteq \dom{\Gamma_2}$ and $\Gamma_2\ROrdCtx\Gamma_1,\Delta$. 
By \cref{prop:wt-ann}, we get $\IsWFConfI{\Gamma_2}{\erase{\e'_2}}{\erase{\stack'_2}}{\typeGrad\C\rGrad}{\e'_2}{\stack'_2}$ and 
by \cref{prop:gr-sem-sound}, we get $\reduce{\e_1}{\stack_1}{\erase{\e'_2}}{\erase{\stack'_2}}$. 
By the determinism of the standard semantics we have $\erase{\e'_2}= \e_2$ and $\erase{\stack'_2}= \stack_2$, hence the thesis. 
\end{proof}


\section{Combining grades}\label{sect:combining}

As we have seen, each grade algebra encodes a specific notion of resource usage. 
However, in a program one may need different notions of usage for different kinds of resources or different pieces of code (e.g., different classes). 
This means that one needs to use several grade algebras at the same time, that is, a family $(\Diag_\kind)_{k\in\KSet}$ of grade algebras\footnote{$\Diag$ stands for ``heterogeneous''.} indexed over a set $\KSet$ of \emph{grade kinds}. We assume grade kinds to always include $\NKind$ and $\TKind$, with $\Diag_\NKind$ and $\Diag_\TKind$ the grade algebras of natural numbers and trivial, respectively, as in \cref{ex:gr-alg}, since they play a special role, as will be shown.

\begin{example}\label{ex:kinds}
Assume to use, in a program, grade kinds $\NKind$, $\AKind$, $\PKind$, $\PPKind$, $\APKind$, and $\TKind$, where:
\begin{itemize}
\item $\Diag_\AKind$ is the affinity grade algebra, as in \refItem{ex:gr-alg}{2}.
\item $\Diag_\PKind$ and $\Diag_\PPKind$ are two different instantiations of the grade algebra of privacy levels, as in \cref{ex:privacy}; namely, in $\Diag_\PKind$ there are only two privacy levels $\public$ and $\private$, whereas in $\Diag_\PPKind$ we have privacy levels $\privA$, $\privB$, $\privC$, $\privD$, with $\privA\rord\privB\rord\privD$ and $\privA\rord\privC\rord\privD$. 
\item Finally,  $\Diag_\APKind$ is $\Diag_\AKind\times\Diag_\PKind$,  as in \refItem{ex:gr-alg}{7},  tracking  simultaneously affinity and privacy.
\end{itemize}
\end{example}
%
%

 We want to make grades of all such kinds simultaneously available to the programmer.  In order to achieve this, we should specify how to \emph{combine} grades of different kinds through their distinctive operators; for instance, an object with grade of kind $\kind$ could have a field with grade of kind $\idxb$, hence a field access should be graded by their multiplication. 

In other words, we need to construct, starting from the family $(\Diag_\kind)_{k\in K}$, a single grade algebra of \emph{heterogeneous grades}. In this way, the meta-theory developed in previous sections for an arbitrary grade algebra applies also  to the case when several grade algebras are used at the same time. Note that this construction is necessary since we do not want available grades to be fixed, as in \cite{OrchardLE19}; rather, the programmer should be allowed to define grades for a specific application, using some linguistic support which could be the language itself, as will be described in \cref{sect:java}.

\subsection{Direct refinement}
The obvious approach is to define heterogeneous grades as pairs $\Pair{\kind}{\rGrad}$ where $\kind\in\KSet$, and $\rGrad\in\Diag_\kind$. Concerning the definition of the operators, in previous work, handling coeffects rather than grades, \cite{BianchiniDGZ22} we took the simplest choice, that is,  combining (by either sum or product) grades of different kinds always returns $\Pair{\infty}{\TKind}$, meaning, in a sense, that we ``do not know'' how the combination should be done. The only exception are grades of kind $\NKind$; indeed, since the corresponding grade algebra is initial, we know that, for any kind $\kind$, there is a unique grade homomorphism $\iota_{\kind}$ from  $\NN$ to $\Diag_\kind$, hence, to combine $\Pair{n}{\NKind}$ with $\Pair{\rGrad}{\kind}$, we can map $n$ into a grade of kind $\kind$ through such homomorphism, and then use the operator of kind $\kind$. In this paper, we generalize this idea, by allowing the programmer to specify, for each pair of kinds $\kind$ and $\idxb$, a uniquely determined kind $\kind\isum\idxb$ and two uniquely determined grade homomorphisms  $\fun{\injl[\idxa,\idxb]}{\Diag_\idxa}{\Diag_{\idxa\isum\idxb}}$,  and $\fun{\injr[\idxa,\idxb]}{\Diag_\idxb}{\Diag_{\idxa\isum\idxb}}$. In this way, to combine $\Pair{\idxa}{\rGrad}$ and $\Pair{\idxb}{\sGrad}$, we can map both in grades of kind $\kind\isum\idxb$, and then use the operator of kind $\kind\isum\idxb$.

The operator $\isum$ and the family of unique homomorphisms, one for each pair of kinds, can be specified by the programmer, in a minimal and easy to check way, by defininig a \emph{(direct) refinement relation} $\refines$, as defined below, and a family of grade homomorphisms $\fun{\Diag_{\idxa,\idxb}}{\Diag_\idxa}{\Diag_\idxb}$, indexed over pairs $\idxa\refines\idxb$.

Given a relation $\Rightarrow$ on kinds, a \emph{path} from $\kind_0$ to $\kind_n$ is a sequence $\kind_0\ldots\kind_n$ such as 
$\kind_i\Rightarrow\kind_{i+1}$, for all $i\in 1..n-1$. We say that $\idxb$ is an \emph{ancestor} of $\idxa$ if there is a path from $\idxa$ to $\idxb$.

Then, a \emph{(direct) refinement relation} is a relation $\refines$ on $\KSet\setminus\{\NKind,\TKind\}$ such as the following conditions hold:
\begin{enumerate}
\item for each $\idxa,\idxb$, there exists at most one path from $\idxa$ to $\idxb$
\item for each $\idxa,\idxb$ with a common ancestor, there is a \emph{least} common ancestor, denoted $\idxa\isum\idxb$; that is, such that, for any common ancestor $\idxc$, $\idxc$ is an ancestor of $\kind\isum\idxb$.
\end{enumerate}
Note that, thanks to requirement (1), requirement (2) means that the unique path, e.g.,  from $\idxa$ to $\idxc$, consists of a unique path from $\idxa$ to $\kind\isum\idxb$, and then a unique path from $\kind\isum\idxb$ to $\idxc$.

Given a direct refinement relation $\refines$, we can derive the following structure on $\KSet$:
\begin{itemize}
\item $\refines$ can be extended to a partial order $\pord$ on $\KSet$, by taking the reflexive and transitive closure of $\refines$ and adding $\NKind\pord\idxa\pord\TKind$ for all $\idxa\in\KSet$.
\item $\isum$ can be extended to all pairs, by defining $\idxa\isum\idxb=\TKind$ if $\idxa$ and $\idxb$ have no common ancestor.
\end{itemize}
Altogether, we obtain an instance of a structure called \emph{grade signature}, as will be  detailed  in \cref{def:gradesig}. Moreover, given a $\refines$-family of homomorphisms:
\begin{itemize}
\item they can be extended, by composition\footnote{ Note that in this way we obtain, in particular, all the identities.}, to all pairs of grades $\Pair{\idxa}{\idxb}\in\KSet\setminus\{\NKind,\TKind\}$ such that there is a path from $\idxa$ to $\idxb$; since this path is unique, the resulting homomorphism is uniquely defined
\item for each kind $\idxa$, we add the unique homomorphisms from $\NN$ and to $\Triv$. 
\end{itemize}
Altogether, besides a grade algebra for each kind, we get a grade homomophism for each pair $\Pair{\idxa}{\idxb}$ such that $\idxa\pord\idxb$. That is,
we obtain an instance of a structure called \emph{heterogeneous grade algebra}, as will be   detailed  in \cref{def:hgradealg}.

 Thus, as desired,  combining grades of kinds $\Pair{\idxa}{\rGrad}$ and $\Pair{\idxb}{\sGrad}$ can be defined by mapping both $\rGrad$ and $\sGrad$ into grades of kind $\idxa\isum\idxb$, and then the operator of kind $\idxa\isum\idxb$ is applied.

The fact that in this way we actually obtain a grade algebra, that is, all required axioms are satisfied, is proved in the next subsection on the more general case of an arbitrary grade signature and heterogeneous grade algebra.

 Note the special role played by the grade kinds $\NKind$ and $\TKind$, with their corresponding grade algebras. The former turns out to be the minimal kind required in a grade signature (\cref{def:gradesig}); this is important since the zero and one of the resulting grade algebra (hence the zero and one used in the type system) will be those of this kind. The latter, as shown above, is used as default common ancestor for pairs of kinds which do not have one. 

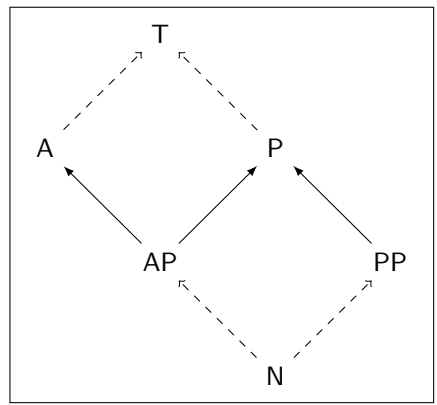
\begin{wrapfigure}[15]{r}{0.5\textwidth}
\centering
\framebox{
\begin{tikzpicture}
\node (triv) at (10ex,30ex) {$\TKind$};
\node (A) at (0ex,20ex) {$\AKind$};
\node (P) at (20ex,20ex) {$\PKind$};
\node (AP) at (10ex,10ex) {$\APKind$};
\node (PP) at (30ex,10ex) {$\PPKind$};
\node (nat) at (20ex,0ex) {$\NKind$};
\draw [arrows={-latex}] (AP) -- (A);
\draw [arrows={-latex}] (AP) -- (P);
\draw [arrows={-latex}] (PP) -- (P);
\draw [dashed,->] (nat) -- (AP);
\draw [dashed,->] (nat) -- (PP);
\draw [dashed,->] (A) -- (triv);
\draw [dashed,->] (P) -- (triv);
\end{tikzpicture}
}
\caption{Direct refinement diagram}
\label{fig:refinement}
\end{wrapfigure}

\begin{example}\label{ex:homo}
Coming back to our example, a programmer could define the direct refinement relation and the corresponding homomorphisms as follows:
\begin{itemize}
\item $\PPKind\refines\PKind$, and the corresponding homomorphism maps, e.g., $\privA$, $\privB$, and $\privC$ into $\private$ and $\privD$ into $\public$
\item $\APKind\refines\AKind$, and $\APKind\refines\PKind$, and the corresponding  homomorphisms are the projections.
\end{itemize}
Thus, for instance, the  grade  $\Pair{\APKind}{\Pair{\omega}{\private}}$, meaning that we can use the resource an arbitrary number of times in $\private$ mode,  and $\Pair{\PPKind}{\privD}$, meaning that we can use the resource in $\privD$ mode, gives $\private$. Indeed, both grades are mapped into the grade algebra of privacy levels $0\rord\private\rord\public$; for the former, the information about the affinity is lost, whereas for the second the privacy level $\privD$ is mapped into $\public$; finally, we get $\private = \private\rmul\public$.

 The direct refinement relation is pictorially shown in \cref{fig:refinement}. Dotted arrows denote (some of) the order relations added for $\NKind$ and $\TKind$.
\end{example}
Note that specifying the grade signature and the heterogeneous grade algebra indirectly, by means of the direct refinement relation and the corresponding homomorphisms, has a fundamental advantage: the semantic check that, for each $\idxa$, $\idxb$, we can map grades of grade $\idxa$ into grades of kind $\idxb$ in a unique way (that is, there is at most one homomorphism), which would require checking the equivalence of function definitions, is replaced by the checks (1) and (2) in the definition of direct refinement, which are purely syntactic and can be easily implemented in a type system (a simple stronger condition is to impose that each kind has a unique parent in the  direct refinement relation, as it is for single inheritance).

In \cref{sect:java}, we will see how to express both grade algebras and homomorphisms in Java; roughly, both will be represented by classes implementing a suitable generic interface.

\subsection{A general construction}
We provide a construction that, starting from a family of grade algebras with  a suitable structure, yields a unique grade algebra summarising the whole family. 
As a consequence, the meta-theory developed in previous sections for a single grade algebra applies also  to the case when several grade algebras are used at the same time. 

To develop this construction, we use simple and standard categorical tools, referring to \cite{MacLane13,Riehl07} for more details. 
Given a category \ct{C}, we denote by $\ct{C}_0$ the collection of objects in \ct{C} and we say that \ct{C} is \emph{small} when $\ct{C}_0$ is a set. 
Recall that any partially ordered set $\PP = \ple{\PP_0,\pord}$ can be seen as a small category where objects are the elements of $\PP_0$ and, for all $x,y \in \PP_0$, there is an arrow $x\to y$ iff $x\pord y$; 
hence, for every pair of objects in $\PP_0$, there is at most one arrow between them, and the only isomorphisms are the identities. 
\begin{definition}\label{def:gradesig}
A \emph{grade signature} \Sh is a partially ordered set with finite suprema, that is, 
it consists of the following data: 
\begin{itemize}
\item a partially ordered set \ple{\Sh_0,\pord}; 
\item a function \fun{\isum}{\Sh_0\times\Sh_0}{\Sh_0} monotone in both arguments and such that 
for all $\idxa,\idxb,\idxc\in\Sh_0$, 
$\idxa\isum\idxb\pord\idxc$ iff $\idxa\pord\idxc$ and $\idxb\pord\idxc$; 
\item a distinguished object $\izero \in \Sh_0$ such that $\izero\pord\idxa$, for all $\idxa\in\Sh_0$. 
\end{itemize}
\end{definition}
Intuitively, objects in \Sh represent the \emph{kinds} of grades one wants to work with, while the arrows, namely, the order relation, model a refinement between such kinds: $\idxa\pord\idxb$ means that the kind $\idxa$ is \emph{more specific} than the kind $\idxb$. 
The operation $\isum$ combines two kinds to  produce the most specific kind generalising both. 
Finally, the kind $\izero$ is the most specific one. 
Reading a grade signature \Sh as a category, being a grade signature means having finite coproducts. 

It is easy to check that the following properties hold for all $\idxa,\idxb,\idxc\in\Sh_0$: 
\begin{align*}
(\idxa\isum\idxb)\isum\idxc &= \idxa \isum(\idxb\isum\idxc) 
  &
\idxa\isum\idxa &=\idxa 
  \\
\idxa\isum\idxb &=\idxb\isum\idxa 
  & 
\idxa\isum\izero &= \idxa 
\end{align*}
namely, \ple{\Sh_0,\isum,\izero} is a commutative idempotent monoid. 

\begin{definition}\label{def:hgradealg}
A \emph{heterogeneous grade algebra} over the grade signature $\Sh$ is just  a functor \fun{\Diag}{\Sh}{\GrAlg}. 
This means that 
it consists of a grade algebra
$\Diag(\idxa)$, written also $\Diag_\idxa$, for every kind $\idxa \in \Sh_0$, and 
a grade algebra homomorphism \fun{\Diag_{\idxa,\idxb}}{\Diag_\idxa}{\Diag_\idxb} for every arrow $\idxa \pord \idxb$, respecting composition and identities\footnote{ The notation $\Diag_{\idxa,\idxb}$ makes sense, since between $\idxa$ and $\idxb$ there is at most one arrow.} 
, that is, 
$\idxa\pord\idxb\pord\idxc$ implies $\Diag_{\idxa,\idxc} = \Diag_{\idxb,\idxc}\circ \Diag_{\idxa,\idxb}$ and 
$\Diag_{\idxa,\idxa} = \id_{\Diag_\idxa}$. 
\end{definition}
Essentially, the homomorphisms $\Diag_{\idxa,\idxb}$ realise the refinement $\idxa\pord\idxb$, transforming grades of kind $\idxa$ into grades of kind $\idxb$, preserving the grade algebra structure. 

Observe that the arrows $\izero\pord\idxa$ and $\idxa\pord\idxa\isum\idxb$ and $\idxb\pord\idxa\isum\idxb$ in $\Sh$ give rise to the following grade algebra homomorphisms: 
\[
\fun{\injz[\idxa] = \Diag_{\izero,\idxa}}{\Diag_\izero}{\Diag_\idxa}
\quad 
\fun{\injl[\idxa,\idxb] = \Diag_{\idxa,\idxa\isum\idxb}}{\Diag_\idxa}{\Diag_{\idxa\isum\idxb}} 
\quad 
\fun{\injr[\idxa,\idxb] = \Diag_{\idxb,\idxa\isum\idxb}}{\Diag_\idxb}{\Diag_{\idxa\isum\idxb}} 
\]
which provide us with a way to map 
grades of kind $\izero$ into grades of any other kind, and 
grades of kind $\idxa$ and $\idxb$ into grades of their composition $\idxa\isum\idxb$. 
By functoriality of $\Diag$ and using the commutative idempotent monoid structure of $\Sh$, we get the following equalities hold in the category \GrAlg, ensuring consistency of such transformations: 
\begin{align}
\injl[\idxa\isum\idxb,\idxc] \circ \injl[\idxa,\idxb] &= \injl[\idxa,\idxb\isum\idxc] 
\label{eq:inj-assoc1} \\ 
\injr[\idxa\isum\idxb,\idxc] \circ \injl[\idxa,\idxb] &= \injr[\idxa,\idxb\isum\idxc]\circ \injl[\idxb,\idxc] 
\label{eq:inj-assoc2}  \\ 
\injl[\idxa,\idxb] &= \injr[\idxb,\idxa] 
\label{eq:inj-comm}  \\ 
\injl[\idxa,\idxa] &= \id_{\Diag_\idxa} 
\label{eq:inj-idm} \\ 
\injl[\idxa,\izero] &= \id_{\Diag_\idxa} 
\label{eq:inj-id1} \\
\injr[\idxa,\izero] &= \injz[\idxa]
\label{eq:inj-id2} 
\end{align}

In the following, we will show how to turn a heterogeneous grade algebra into a single grade algebra. 
The procedure we will describe is based on a general construction due to Grothendieck \cite{Grothendieck71} defined on indexed categories.  
 
Let us assume a grade signature \Sh  and a heterogeneous grade algebra \fun{\Diag}{\Sh}{\GrAlg}. 
We consider the following set: 
\[ 
\RC{\GrD\Diag} = \{ \ple{\idxa,\rgr} \mid \idxa\in\Sh_0,\ \rgr \in \RC{\Diag_\idxa} \} 
\]
That is, elements of $\GrD\Diag$ will be \emph{kinded grades}, namely, pairs of a kind $\idxa$ and a grade of that kind. 
Note that this is indeed a set because $\Sh$ is small, that is, $\Sh_0$ is a set. 
Then, we define a binary relation $\rord_\Diag$ on $\RC{\GrD\Diag}$ as follows: 
\[
\ple{\idxa,\rgr} \rord_\Diag \ple{\idxb,\sgr} \quad \text{iff}\quad \idxa\pord\idxb \text{ and } \Diag_{\idxa,\idxb}(\rgr)\rord_\idxb \sgr  
\]
that is, the kind $\idxa$ must be more specific than the kind $\idxb$ and, transforming $\rgr$ by $\Diag_{\idxa,\idxb}$, we obtain a grade of kind $\idxb$ which is smaller than $\sgr$. 
These data define a partially ordered set as the following proposition shows. 

\begin{proposition}\label{prop:grd-preord}
\ple{\RC{\GrD\Diag}, \rord_\Diag} is a partially ordered set. 
\end{proposition} 

The additive structure is given by a binary operation 
\fun{\rsum_\Diag}{\RC{\GrD\Diag}\times\RC{\GrD\Diag}}{\RC{\GrD\Diag}} and an element $\rzero_\Diag$ in $\RC{\GrD\Diag}$ defined as follows: 
\begin{align*}
\ple{\idxa,\rgr}\rsum_\Diag\ple{\idxb,\sgr} &= \ple{\idxa\isum\idxb,\injl[\idxa,\idxb](\rgr) \rsum_{\idxa\isum\idxb} \injr[\idxa,\idxb](\sgr)}
  &
\rzero_\Diag &= \ple{\izero,\rzero_\izero} 
\end{align*} 
This means that, the addition of two elements \ple{\idxa,\rgr} and \ple{\idxb,\sgr} is performed by first mapping $\rgr$ and $\sgr$ in the most specific kind generalising both $\idxa$ and $\idxb$, namely $\idxa\isum\idxb$, and then by summing them in the grade algebra over that kind. 
The zero element is just the zero of the most specific kind.
 
\begin{proposition}\label{prop:sum-mon}
\ple{\RC{\GrD\Diag},\rord_\Diag,\rsum_\Diag,\rzero_\Diag} is an ordered commutative monoid. 
\end{proposition}

\begin{proposition}\label{prop:zero-min}
$\rzero_\Diag\rord_\Diag\ple{\idxa,\rgr}$ for every $\ple{\idxa,\rgr} \in \RC{\GrD\Diag}$. 
\end{proposition}

Similarly, the multiplicative structure is given by a binary operation 
\fun{\rmul_\Diag}{\RC{\GrD\Diag}\times\RC{\GrD\Diag}}{\RC{\GrD\Diag}} and an element $\rone_\Diag$ in $\RC{\GrD\Diag}$ defined as follows: 
\begin{small}
\begin{align*} 
\ple{\idxa,\rgr}\rmul_\Diag\ple{\idxb,\sgr} &= \begin{cases}
 \ple{\idxa\isum\idxb,\injl[\idxa,\idxb](\rgr) \rmul_{\idxa\isum\idxb} \injr[\idxa,\idxb](\sgr)}  & \ple{\idxa,\rgr}\ne\rzero_\Diag\text{ and }\ple{\idxb,\sgr}\ne\rzero_\Diag \\ 
\rzero_\Diag  & \text{otherwise} 
\end{cases} 
  & 
\rone_\Diag &= \ple{\izero, \rone_{\izero}} 
\end{align*}
\end{small}
Notice that the definitions above follow almost the same pattern as additive operations, 
but we force that multiplying by $\rzero_\Diag$ we get again $\rzero_\Diag$, which is a key property of grade algebras. 

\begin{proposition}\label{prop:mul-mon}
\ple{\RC{\GrD\Diag},\rord_\Diag,\rmul_\Diag,\rone_\Diag} is an ordered monoid. 
\end{proposition}

Altogether, we finally get the following result. 

\begin{theorem}\label{thm:multigrade}
$\GrD\Diag = \ple{\RC{\GrD\Diag},\rord_\Diag,\rsum_\Diag,\rmul_\Diag,\rzero_\Diag,\rone_\Diag}$ is a grade algebra. 
\end{theorem}


\section{Grades as Java expressions}\label{sect:java}
In \cref{sect:GrFJ} we  described how a Java-like language could be equipped with \emph{grades} decorating types, taken in an arbitrary grade algebra. Moreover, in \cref{sect:combining}, we have shown that such grade algebra could have been obtained by composing, in a specific way determined by providing a minimal collection of grade homomorphisms, different grade algebras corresponding to different ways to track usage of resources.
In this way, \emph{heterogeneous} grades can coexist in the same program, and do not need to  be  fixed once and for all, but can be enriched by user-defined grade algebras and homomorphisms.
In this section, we consider the issue of providing a linguistic support to this aim. This could be done by using an ad-hoc configuration  language, however,   we believe an interesting solution is that the grade annotations in types could be written themselves in Java. 

The key idea is that such grade annotations are Java expressions of (classes implementing) a predefined interface \lstinline{Grade}, analogously to Java exceptions which are expressions of
(subclasses of) \lstinline{Exception}. Moreover, grade homomorphisms are user-defined as well.
 Namely,  a user program can include:
\begin{itemize}
\item pairs of \emph{grade classes} and \emph{grade factory classes}, each one modeling a grade algebra desired for the specific application, with the factory class providing  its  constants
\item \emph{grade homomorphism classes}, each one modeling a homomorphism from a grade algebra (class) to another.
\end{itemize}

When typechecking code with grade annotations, the grades internally used by the type system are those obtained by combining all the declared grade algebras (classes) by means of the declared grade homomorphism classes,  with  the construction described in \cref{sect:combining}. Operations on grades in the same grade algebra (class) are derived from user-defined methods, as discussed more in detail below, whereas operations on heterogeneous grades are derived as in the construction in \cref{sect:combining}.

In the following we describe each point. 

\subsection{Grade and grade factory classes}
 They  are classes implementing the following generic interfaces, respectively:\footnote{\bez We omit access modifiers to make the code lighter.\eez}
\begin{lstlisting}
interface Grade<T extends Grade<T>> {
  boolean leq(T x);
  T sum(T x); 
  T mult(T x); 
}
interface GradeFactory<T extends Grade<T>>{
  T zero();
  T one();
}
\end{lstlisting}
For instance, considering \cref{ex:kinds}, we show below the complete implementation of affinity grades, and only the skeletons of the other classes, for brevity:

\begin{lstlisting}
abstract class Affinity implements Grade<Affinity>{}

class AffinityFactory implements GradeFactory<Affinity>{
  Affinity zero(){return new ZeroA();}
  Affinity one(){return new One();}
}

class ZeroA extends Affinity{
  boolean leq(Affinity x){return true;}
  Affinity sum(Affinity x){return x;}
  Affinity mult(Affinity x){return this;}
}

class One extends Affinity{
  boolean leq(Affinity x){return !(x instanceof ZeroA);}
  Affinity sum(Affinity x){
    if (x instanceof ZeroA) return this; 
    return new Omega();
  }
  Affinity mult(Affinity x){return x;}
}
class Omega extends Affinity{
  boolean leq(Affinity x){return (x instanceof Omega);}
  Affinity sum(Affinity x){return this;}
  Affinity mult(Affinity x){
    if (x instanceof ZeroA) return x; 
    return this;
  }
}

abstract class Privacy implements Grade<Privacy> {...}

class ZeroP extends Privacy{...}
class Private extends Privacy{...}
class Public extends Privacy{...}
class PrivacyFactory implements GradeFactory<Privacy>{...}

abstract class PPrivacy implements Grade<PPrivacy> {...}

class ZeroPP extends PPrivacy{...}
class LevelA extends PPrivacy{...}
...
class LevelD extends PPrivacy{...}

class PPrivacyFactory implements GradeFactory<PPrivacy>{...}

class AffinityPrivacy implements Grade<AffinityPrivacy>{
  Affinity first; Privacy second;
  ...
}
  
class AffinityPrivacyFactory implements GradeFactory<AffinityPrivacy>{
...
}
\end{lstlisting}
 Note that for \lstinline{Affinity}, \lstinline{Privacy}, and \lstinline{PPrivacy} we could have used an \lstinline{enum} class with constants as well; here we preferred to use a hierarchy of classes,  which is a core object-oriented feature provided in the $\FJ$ calculus. More in general, note that the implementation, as expected, depends very much on the features of the target language; for instance, in Haskell we could express \lstinline{Grade} as a typeclass, and grade algebras as its instances, and the same for \lstinline{GradeHom} below.   

\subsection{Grade homomorphism classes}
 They  are classes implementing the following generic interface:
\begin{lstlisting}
interface GradeHom<T extends Grade<T>, R extends Grade<R>> {
  R apply(T x); 
}
\end{lstlisting}
For instance, the grade homomorphisms of \cref{ex:homo} can be implemented as follows:
\begin{lstlisting}
class PPtoPGradeHom implements GradeHom<PPrivacy,Privacy>{
   Privacy apply (PPrivacy p){
     if (p instanceof LevelD) return new Public();
     return new Private();
   } 
}

class APtoAGradeHom implements GradeHom<AffinityPrivacy,Affinity>{
   Affinity apply (AffinityPrivacy ap){return ap.first;}
}

class APtoPGradeHom implements GradeHom<AffinityPrivacy,Privacy>{
   Privacy apply (AffinityPrivacy ap){return ap.second;}
}
\end{lstlisting}

\EZComm{notare che bisogna usare nomi diversi per i diversi zeri}

\subsection{Predefined grade classes} The grade algebra $\NN$ can be implemented by the following classes, assumed to be predefined.
\begin{lstlisting}
abstract class Nat implements Grade<Nat> { 
  abstract <R extends Grade<R>>  R transform(GradeFactory<R> factory); 
  } 

class Zero extends Nat {
  <R extends Grade<R>> R transform(GradeFactory<R> factory){
    return factory.zero();
  }
  boolean leq(Nat x){return true;}
  Nat sum(Nat x){return x;}
  Nat mult(Nat x){return this;}
}

class Succ extends Nat {
  Nat pred; 
  Succ(Nat pred){this.pred=pred;}

  <R extends Grade<R>> R transform(GradeFactory<R> factory) { 
    return factory.one().sum(pred.transform(factory));
  }
  boolean leq(Nat x){
  if (x instanceof Zero){ return false;}
  return pred.leq(((Succ) x).pred);
  }
  Nat sum(Nat x){return new Succ(pred.sum(x));}
  Nat mult(Nat x){return pred.mult(x).sum(x);}
}

class NatFactory implements GradeFactory<Nat>{
  Nat zero(){return new Zero();}
  Nat one(){return new Succ(zero());}
}

class NtoGradeHom<R extends Grade<R>> implements GradeHom<Nat,R> {
  GradeFactory<R> factory; 
  NtoGradeHom(GradeFactory<R> factory){this.factory=factory;}
  R apply(Nat n) {return n.transform(this.factory);}
}
\end{lstlisting}
Order, sum, and multiplication on natural numbers are defined in the expected inductive way. 
An object of class \lstinline{NtoGradeHom<R>} models the unique homomorphism $\iota_\RR$ from $\NN$ to the grade algebra (modeled by) \lstinline{R}, which maps a natural number $n$ to the sum in \RR of $n$ copies of $\rone_\RR$, as defined at page \pageref{nat-hom}. To achieve this, the object is created specifying a constructor argument of the corresponding factory class.\footnote{ To this end, we need a convention to construct names of factory classes.} When the homomorphism is applied to a natural number \lstinline{n} (invoking \lstinline{apply} with argument \lstinline{n}), the auxiliary method \lstinline{transform} constructs the object of class \lstinline{R} corresponding to \lstinline{n}, using the \lstinline{zero} and \lstinline{one} method provided by the factory. 

Analogously, for the grade algebra $\Triv$ we have:
\begin{lstlisting}
class Triv implements Grade<Triv>{
   boolean leq(Triv t){return true;}
   Triv sum(Triv t){return this;}
   Triv mult(Triv t){return this;}
}
class TrivFactory implements GradeFactory<Triv>{
  Triv zero(){return new Triv();}
  Triv one(){return new Triv();}
}

class toTrivGradeHom<R extends Grade<R>>implements GradeHom<R,Triv>{
  Triv apply(R r){return new Triv();}
}
\end{lstlisting}

\subsection{Grade annotations as Java expressions}
Assuming to have the classes decribed above, the programmer could write code with grade annotations being Java expressions. For instance,  \cref{ex:ex1} could be written as follows, where annotations are in square brackets: 

\begin{lstlisting}
class Pair { A$\annotation{new One()}$ first; A$\annotation{new One()}$ second; 
  A$\annotation{new ZeroA()}$ getFirstZero() $\annotation{new One()}${this.first}
  A$\annotation{new One()}$ getFirstAffine() $\annotation{new One()}${this.first}
  A$\annotation{new Omega()}$ getFirst() $\annotation{new One()}${this.first}
}
\end{lstlisting}
Typechecking could then be performed in two steps:
\begin{enumerate}
\item Code defining grades, which is assumed to be standard (that is, non-graded) Java code, is typechecked by the standard compiler.
\item Graded code (containing grade annotations written in Java) is typechecked accordingly to the graded type system in \cref{fig:typing}, where the underlying grade algebra is obtained by composing, by the construction described in \cref{sect:combining}, the user-defined grade algebras through the user-defined grade homomorphisms.
Each user-defined algebra has as carrier (set of grades) the Java values which are instances of the corresponding class, and the operations are computed by executing user-defined methods in such class. For instance, to compute the sum $\val_1\rsum\val_2$ of two grades which are values of a grade class, we evaluate $\SumMethod{\val_1}{\val_2}$. Analogously to compute the result of a grade homomorphism.
\end{enumerate}
 For  the whole process to work correctly, the following are responsabilities of \mbox{the programmer:}
\begin{itemize}
\item Grade classes, grade factory classes, and grade homomorphism classes should satisfy the axioms required for the structures they model, e.g., that the sum derived from \lstinline{sum} methods is commutative and associative. The same happens, for instance, in Haskell, when one defines instances of \lstinline{Functor} or \lstinline{Monad}.
\item Code defining grades should be \emph{terminating}, since, as described above, the second typechecking step requires to \emph{execute} code typechecked in the first step. 
\item Finally, the relation among grade classes implicitly defined by declaring grade homomorphism classes should actually be a direct refinement relation, that is, should satisfy the two requirements: (1) there exists at most one path between two grade classes, and (2) each two grade classes with a common ancestor have a \emph{least} common ancestor. These are requirements easy to check, similarly to the check that inheritance is acyclic, or that there are no diamonds in multiple inheritance. 
\end{itemize}
An interesting
point is that implementations could use in a parametric way auxiliary tools,
notably a termination checker to prevent divergence in methods implementing grade operations, and/or a verifer to ensure that they provide the
required properties.

\section{Related work}
\label{sect:related} 

The two contributions which have been more inspiring for the work in this paper are the instrumented semantics proposed in \cite{ChoudhuryEEW21} and the Granule language \cite{OrchardLE19}.
In \cite{ChoudhuryEEW21}, the authors develop \textsc{GraD}, a graded dependent type system that includes functions,
tensor products, additive sums, and a unit type. \bez Moreover, they define an instrumented operational semantics which tracks usage of resources, and prove that the graded type system is sound with respect to such instrumented semantics. \eez In this paper, we take the same approach to define a resource-aware semantics, parametric on an arbitrary grade algebra. However, differently from \cite{ChoudhuryEEW21}, where such semantics is defined on typed terms, with the only aim to show the role of the type system, the definition of our semantics is given \emph{independently} from the  type system, as is the standard approach in calculi. That is, the aim is also to provide a simple purely semantic model which takes into account usage of resources.

Granule \cite{OrchardLE19} is a functional language equipped with graded modal types, where different kinds of grades can be used at the same time, including naturals for exact usage, security levels, intervals, infinity, and products of coeffects. We owe to Granule the idea of allowing different kinds of grades to coexist, and the overall objective to exploit graded modal types in a programming language.
Concerning heterogeneous grades, in this paper we push forward the Granule approach, since we do not  want  this grade algebra  to be fixed, but extendable by the programmer with user-defined grades. To this aim we define the construction in \cref{sect:combining}. Concerning the design of a  graded programming language, here we investigate the object-oriented rather than functional paradigm, taking some solutions which seem more adequate in that context, e.g., to have once-graded types and no boxing/unboxing. The design and implementation of a real Java-like language are not objectives of the current paper; however, we outline in \cref{sect:java} a possible interesting solution, where grade annotations are written in the language itself.  

Coming more in general to resource-aware type systems, coeffects were  first  introduced by \cite{PetricekOM13} and further analyzed by \cite{PetricekOM14}. 
In particular, \cite{PetricekOM14} develops a  generic  coeffect system which augments the simply-typed $\lambda$-calculus with context annotations indexed by \emph{coeffect shapes}. 
The proposed framework is very abstract, and the authors focus only on two opposite instances: 
structural (per-variable) and flat (whole context) coeffects, identified by specific choices of context shapes.

Most of the  subsequent  literature on coeffects focuses on structural ones, for which there is a clear algebraic description in terms of semirings. 
This was first noticed by \cite{BrunelGMZ14}, who developed a framework for structural coeffects 
for a functional language. 
This approach is inspired by a generalization of the exponential modality of linear logic, see, e.g., \cite{BreuvartP15}. 
That is, the distinction between linear and  unrestricted  variables of linear systems is generalized to have variables  
decorated by coeffects, or {\em grades}, that determine how much they can be used. 
In this setting, many advances have been made to combine coeffects with other programming features, such as 
computational effects \cite{GaboardiKOBU16,OrchardLE19,DalLagoG22}, 
dependent types \cite{Atkey18,ChoudhuryEEW21,McBride16}, and 
polymorphism \cite{AbelB20}. 
Other graded type systems are explored in \cite{Atkey18,GhicaS14,AbelB20},  also combining effects and coeffects \cite{GaboardiKOBU16,OrchardLE19}.  
In all these papers, the process of
tracking usage through grades is a powerful method of instrumenting
type systems with analyses of irrelevance and linearity that have practical benefits like erasure
of irrelevant terms (resulting in speed-up) and compiler optimizations (such as in-place update). 

As already mentioned, \cite{McBride16} and \cite{WoodA22}  observed that contexts in a structural coeffect system form a module over the semiring of grades, event though they do not use this structure in its full generality, restricting themselves to free modules, that is, to structural coeffect systems. Recently, \cite{BianchiniDGZS22} shows a significant non-structural instance, namely, a coeffect system to track sharing in the imperative paradigm.


\section{Conclusion}\label{sect:conclu}
The contributions of the paper can be summarized as follows:
\begin{itemize}
\item Resource-aware extension of $\FJ$ reduction, parametric on an arbitrary grade algebra.
\item Resource-aware extension of the type system, proved to ensure soundness-may of the resource-aware semantics.
\item Formal construction which, given grades of different kinds and grade transformations corresponding to a refinement relation among kinds (formally, a functor over a grade signature), provides a grade algebra of \emph{heterogeneous grades}.
\item Notion of \emph{direct refinement} allowing a minimal and easy to check way to specify the above functor.  
\item Outline of a Java extension where grades are user-defined, and grade annotations are written in the language itself.
\end{itemize}
As already noted, the key novel ideas in the contributions above are mostly independent from the language. So, a first natural direction for future work is to explore their incarnation in another paradigm, e.g., the functional one. That would include the definition of a parametric resource-aware reduction independent from types, the design of a type system with once-graded types, and possibly the design of user-defined grades in a functional language, e.g., in Haskell by relying on the typeclass feature. Though the overall approach should still apply, we  expect the investigation to be significant due to the specific features of the paradigm. 

The resource-aware operational semantics defined in this paper requires \emph{annotations} in subterms, with the only aim to fix their reduction grade in the reduction of the enclosing term.
As mentioned in \cref{sect:resource-aware}, adopting a big-step style would clearly remove the  need of such technical artifice; only annotations in constructor subterms should be kept, since they express a true constraint on the semantics. Thus, a very interesting alternative to be studied is a big-step version of resource-aware semantics, allowing a more abstract and clean presentation. With this choice, we should employ, to prove soundness-may, the techniques recently introduced in \cite{DagninoBZD20,Dagnino22}. 

Coming back to Java-like languages, the $\FJ$ language considered in the paper does not include imperative features. Adding mutable memory leads to many significant research directions.
First, besides the model presented in this paper, and in general in literature, where ``using a resource'' means ``replacing a variable with its value'', another view is possible where the resource is the memory and ``using'' means ``interacting with the memory''.  
Moreover, we would like to investigate more in detail how to express by grade algebras forms of usages which are typical of the imperative paradigm, such as the $\readonly$ modifier, and, more in general, \emph{capabilities} \cite{HallerOdersky10,Gordon20}.


\bibliography{main}

\begin{thebibliography}{10}

\bibitem{AbelB20}
Andreas Abel and Jean{-}Philippe Bernardy.
\newblock A unified view of modalities in type systems.
\newblock {\em Proceedings of {ACM} on Programming Languages},
  4({ICFP}):90:1--90:28, 2020.
\newblock \href {https://doi.org/10.1145/3408972} {\path{doi:10.1145/3408972}}.

\bibitem{Atkey18}
Robert Atkey.
\newblock Syntax and semantics of quantitative type theory.
\newblock In Anuj Dawar and Erich Gr{\"{a}}del, editors, {\em {IEEE} Symposium
  on Logic in Computer Science, {LICS} 2018}, pages 56--65. {ACM} Press, 2018.
\newblock \href {https://doi.org/10.1145/3209108.3209189}
  {\path{doi:10.1145/3209108.3209189}}.

\bibitem{BianchiniDGZ22}
Riccardo Bianchini, Francesco Dagnino, Paola Giannini, and Elena Zucca.
\newblock A {J}ava-like calculus with user-defined coeffects.
\newblock In Ugo~Dal Lago and Daniele Gorla, editors, {\em ICTCS'22 - Italian
  Conference on Theoretical Computer Science}, volume 3284 of {\em {CEUR}
  Workshop Proceedings}, pages 66--78. CEUR-WS.org, 2022.

\bibitem{BianchiniDGZS22}
Riccardo Bianchini, Francesco Dagnino, Paola Giannini, Elena Zucca, and Marco
  Servetto.
\newblock Coeffects for sharing and mutation.
\newblock {\em Proceedings of {ACM} on Programming Languages},
  6({OOPSLA}):870--898, 2022.
\newblock \href {https://doi.org/10.1145/3563319} {\path{doi:10.1145/3563319}}.

\bibitem{BreuvartP15}
Flavien Breuvart and Michele Pagani.
\newblock Modelling coeffects in the relational semantics of linear logic.
\newblock In Stephan Kreutzer, editor, {\em 24th {EACSL} Annual Conference on
  Computer Science Logic, {CSL} 2015}, volume~41 of {\em LIPIcs}, pages
  567--581. Schloss Dagstuhl - Leibniz-Zentrum fuer Informatik, 2015.
\newblock \href {https://doi.org/10.4230/LIPIcs.CSL.2015.567}
  {\path{doi:10.4230/LIPIcs.CSL.2015.567}}.

\bibitem{BrunelGMZ14}
Alo{\"{\i}}s Brunel, Marco Gaboardi, Damiano Mazza, and Steve Zdancewic.
\newblock A core quantitative coeffect calculus.
\newblock In Zhong Shao, editor, {\em European Symposium on Programming, {ESOP}
  2013}, volume 8410 of {\em Lecture Notes in Computer Science}, pages
  351--370. Springer, 2014.
\newblock \href {https://doi.org/10.1007/978-3-642-54833-8\_19}
  {\path{doi:10.1007/978-3-642-54833-8\_19}}.

\bibitem{ChoudhuryEEW21}
Pritam Choudhury, Harley~Eades III, Richard~A. Eisenberg, and Stephanie
  Weirich.
\newblock A graded dependent type system with a usage-aware semantics.
\newblock {\em Proceedings of {ACM} on Programming Languages}, 5({POPL}):1--32,
  2021.
\newblock \href {https://doi.org/10.1145/3434331} {\path{doi:10.1145/3434331}}.

\bibitem{Dagnino22}
Francesco Dagnino.
\newblock A meta-theory for big-step semantics.
\newblock {\em {ACM} Transactions on Computational Logic}, 23(3):20:1--20:50,
  2022.
\newblock \href {https://doi.org/10.1145/3522729} {\path{doi:10.1145/3522729}}.

\bibitem{DagninoBZD20}
Francesco Dagnino, Viviana Bono, Elena Zucca, and Mariangiola
  Dezani{-}Ciancaglini.
\newblock Soundness conditions for big-step semantics.
\newblock In Peter M{\"{u}}ller, editor, {\em European Symposium on
  Programming, {ESOP} 2020}, volume 12075 of {\em Lecture Notes in Computer
  Science}, pages 169--196. Springer, 2020.
\newblock \href {https://doi.org/10.1007/978-3-030-44914-8\_7}
  {\path{doi:10.1007/978-3-030-44914-8\_7}}.

\bibitem{DalLagoG22}
Ugo {Dal Lago} and Francesco Gavazzo.
\newblock A relational theory of effects and coeffects.
\newblock {\em Proceedings of {ACM} on Programming Languages}, 6({POPL}):1--28,
  2022.
\newblock \href {https://doi.org/10.1145/3498692} {\path{doi:10.1145/3498692}}.

\bibitem{DeNicolaH84}
Rocco {De Nicola} and Matthew Hennessy.
\newblock Testing equivalences for processes.
\newblock {\em Theoretical Computer Science}, 34(1):83 -- 133, 1984.
\newblock \href {https://doi.org/https://doi.org/10.1016/0304-3975(84)90113-0}
  {\path{doi:https://doi.org/10.1016/0304-3975(84)90113-0}}.

\bibitem{DietlEtAl07}
Werner Dietl, Sophia Drossopoulou, and Peter M{\"u}ller.
\newblock Generic universe types.
\newblock In Erik Ernst, editor, {\em European Conference on Object-Oriented
  Programming, {ECOOP} 2007}, volume 4609 of {\em Lecture Notes in Computer
  Science}, pages 28--53. Springer, 2007.

\bibitem{GaboardiKOBU16}
Marco Gaboardi, Shin{-}ya Katsumata, Dominic~A. Orchard, Flavien Breuvart, and
  Tarmo Uustalu.
\newblock Combining effects and coeffects via grading.
\newblock In Jacques Garrigue, Gabriele Keller, and Eijiro Sumii, editors, {\em
  {ACM} International Conference on Functional Programming, {ICFP} 2016}, pages
  476--489. {ACM} Press, 2016.
\newblock \href {https://doi.org/10.1145/2951913.2951939}
  {\path{doi:10.1145/2951913.2951939}}.

\bibitem{GhicaS14}
Dan~R. Ghica and Alex~I. Smith.
\newblock Bounded linear types in a resource semiring.
\newblock In Zhong Shao, editor, {\em European Symposium on Programming, {ESOP}
  2013}, volume 8410 of {\em Lecture Notes in Computer Science}, pages
  331--350. Springer, 2014.
\newblock \href {https://doi.org/10.1007/978-3-642-54833-8\_18}
  {\path{doi:10.1007/978-3-642-54833-8\_18}}.

\bibitem{Gordon20}
Colin~S. Gordon.
\newblock Designing with static capabilities and effects: Use, mention, and
  invariants (pearl).
\newblock In Robert Hirschfeld and Tobias Pape, editors, {\em European
  Conference on Object-Oriented Programming, {ECOOP} 2020}, volume 166 of {\em
  LIPIcs}, pages 10:1--10:25. Schloss Dagstuhl - Leibniz-Zentrum f{\"{u}}r
  Informatik, 2020.
\newblock \href {https://doi.org/10.4230/LIPIcs.ECOOP.2020.10}
  {\path{doi:10.4230/LIPIcs.ECOOP.2020.10}}.

\bibitem{Grothendieck71}
Alexander Grothendieck.
\newblock Cat{\'e}gories fibr{\'e}es et descente.
\newblock In {\em Rev{\^e}tements {\'e}tales et groupe fondamental}, pages
  145--194. Springer, 1971.

\bibitem{HallerOdersky10}
Philipp Haller and Martin Odersky.
\newblock Capabilities for uniqueness and borrowing.
\newblock In Theo D'Hondt, editor, {\em European Conference on Object-Oriented
  Programming, {ECOOP} 2010}, volume 6183 of {\em Lecture Notes in Computer
  Science}, pages 354--378. Springer, 2010.

\bibitem{IgarashiPW99}
Atsushi Igarashi, Benjamin~C. Pierce, and Philip Wadler.
\newblock Featherweight {J}ava: A minimal core calculus for {J}ava and {GJ}.
\newblock In {\em ACM Symp. on Object-Oriented Programming: Systems, Languages
  and Applications 1999}, pages 132--146. {ACM} Press, 1999.
\newblock \href {https://doi.org/10.1145/320384.320395}
  {\path{doi:10.1145/320384.320395}}.

\bibitem{MacLane13}
Saunders Mac~Lane.
\newblock {\em Categories for the working mathematician}, volume~5.
\newblock Springer Science \& Business Media, 2013.

\bibitem{MarshalV22}
Daniel Marshall, Michael Vollmer, and Dominic Orchard.
\newblock Linearity and uniqueness: An entente cordiale.
\newblock In Ilya Sergey, editor, {\em European Symposium on Programming,
  {ESOP} 2022}, volume 13240 of {\em Lecture Notes in Computer Science}, pages
  346--375. Springer, 2022.
\newblock \href {https://doi.org/10.1007/978-3-030-99336-8\_13}
  {\path{doi:10.1007/978-3-030-99336-8\_13}}.

\bibitem{McBride16}
Conor McBride.
\newblock I got plenty o' nuttin'.
\newblock In Sam Lindley, Conor McBride, Philip~W. Trinder, and Donald
  Sannella, editors, {\em A List of Successes That Can Change the World -
  Essays Dedicated to Philip Wadler on the Occasion of His 60th Birthday},
  volume 9600 of {\em Lecture Notes in Computer Science}, pages 207--233.
  Springer, 2016.
\newblock \href {https://doi.org/10.1007/978-3-319-30936-1\_12}
  {\path{doi:10.1007/978-3-319-30936-1\_12}}.

\bibitem{OrchardLE19}
Dominic Orchard, Vilem{-}Benjamin Liepelt, and Harley~Eades III.
\newblock Quantitative program reasoning with graded modal types.
\newblock {\em Proceedings of {ACM} on Programming Languages},
  3({ICFP}):110:1--110:30, 2019.
\newblock \href {https://doi.org/10.1145/3341714} {\path{doi:10.1145/3341714}}.

\bibitem{PetricekOM13}
Tomas Petricek, Dominic~A. Orchard, and Alan Mycroft.
\newblock Coeffects: Unified static analysis of context-dependence.
\newblock In Fedor~V. Fomin, Rusins Freivalds, Marta~Z. Kwiatkowska, and David
  Peleg, editors, {\em Automata, Languages and Programming, {ICALP} 2013},
  volume 7966 of {\em Lecture Notes in Computer Science}, pages 385--397.
  Springer, 2013.
\newblock \href {https://doi.org/10.1007/978-3-642-39212-2\_35}
  {\path{doi:10.1007/978-3-642-39212-2\_35}}.

\bibitem{PetricekOM14}
Tomas Petricek, Dominic~A. Orchard, and Alan Mycroft.
\newblock Coeffects: a calculus of context-dependent computation.
\newblock In Johan Jeuring and Manuel M.~T. Chakravarty, editors, {\em {ACM}
  International Conference on Functional Programming, {ICFP} 2014}, pages
  123--135. {ACM} Press, 2014.
\newblock \href {https://doi.org/10.1145/2628136.2628160}
  {\path{doi:10.1145/2628136.2628160}}.

\bibitem{Riehl07}
Emily Riehl.
\newblock {\em Category theory in context}.
\newblock Courier Dover Publications, 2017.

\bibitem{WoodA22}
James Wood and Robert Atkey.
\newblock A framework for substructural type systems.
\newblock In Ilya Sergey, editor, {\em European Symposium on Programming,
  {ESOP} 2022}, volume 13240 of {\em Lecture Notes in Computer Science}, pages
  376--402. Springer, 2022.
\newblock \href {https://doi.org/10.1007/978-3-030-99336-8\_14}
  {\path{doi:10.1007/978-3-030-99336-8\_14}}.

\end{thebibliography}

\appendix 

\section{Proofs of \cref{sect:combining}} 

\begin{proof}[Proof of \cref{prop:ini-fin-gralg}] 
\cref{prop:ini-fin-gralg:2} is straightforward as the singleton set is a terminal object in  the category of sets and functions. 
Towards a proof of \cref{prop:ini-fin-gralg:1}, 
let \mbox{\fun{f}{\NN}{\RR}} be a grade algebra homomorphism and note that, 
since $n = 1 + \cdots + 1$ ($n$ times), for all $n \in \N$, and $f$ preserves sums and the unit, 
we get $f(n) = f(1) \rsum_\RR \cdots \rsum_\RR f(1) = \rone_\RR \rsum_\RR\cdots\rsum_\RR\rone_\RR$ ($n$ times). 
That is, we have $f(n) = \iota_\RR(n)$, for all $n \in \N$. 
Therefore, to conclude, we just have to show that the map 
$\iota_\RR$ is a grade algebra homomorphism. 
The fact that $\iota_\RR(0) = \rzero_\RR$ and $\iota_\RR(1) = \rone_\RR$ is immediate. 
The fact that $\iota_\RR(n + m) = \iota_\RR(n) \rsum_\RR \iota_\RR(m)$ and ${\iota_\RR(n\cdot m)} = \iota_\RR(n)\rmul_\RR\iota_\RR(m)$ follows from a straightforward induction on $n$, using distributivity and nullity properties of the grade algebra \RR. 
Finally, to prove monotonicity, consider $n \le m$ and proceed by induction on $m - n$. 
If $m - n = 0$, then $n = m$ and so the thesis is trivial. 
If $m - n = k + 1$, we have $m - (n + 1) = k$, then by induction hypothesis we get $\iota_\RR(n + 1) \rord_\RR(\iota_\RR(m)$. 
Since $\iota_\RR(n+1) = \iota_\RR(n) \rsum_\RR \iota_\RR(1)$ and $\rzero_\RR \rord_\RR \iota_\RR(1)$, we get 
$\iota_\RR(n) = \iota_\RR(n) \rsum_\RR \rzero_\RR \rord_\RR \iota_\RR(n) \rsum_\RR \iota_\RR(1) \rord_\RR \iota_\RR(m)$, as needed. 
\end{proof}

\begin{proof}[Proof of \cref{prop:grd-preord}] 
We have to prove that $\rord_\Diag$ is reflexive, transitive and antisymmetric. 
 
Given an element $\ple{\idxa,\rgr}\in \RC{\GrD\Diag}$, since $\idxa\pord\idxa$ and $\Diag_{\idxa,\idxa} = \id_{\Diag_\idxa}$, by functoriality of $\Diag$, we have 
\[ \Diag_{\idxa,\idxa}(\rgr) = \id_{\Diag_\idxa}(\rgr) = \rgr \]
hence $\ple{\idxa,\rgr}\rord_\Diag\ple{\idxa,\rgr}$, which proves reflexivity. 

Given $\ple{\idxa,\rgr}\rord_\Diag\ple{\idxb,\sgr}\rord_\Diag\ple{\idxc,\tgr}$, we know that $\idxa\pord\idxb\pord\idxc$ and 
$\Diag_{\idxa,\idxb}(\rgr)\rord_\idxb \sgr$ and $\Diag_{\idxb,\idxc}(\sgr)\rord_\idxc\tgr$ and, by functoriality of $\Diag$, 
$\Diag_{\idxa,\idxc} = \Diag_{\idxb,\idxc}\circ\Diag_{\idxa,\idxb}$. 
Therefore, we get 
\begin{align*}
\Diag_{\idxa,\idxc}(\rgr) 
  &= \Diag_{\idxb,\idxc}(\Diag_{\idxa,\idxb}(\rgr)) 
   \rord_\idxc \Diag_{\idxb,\idxc}(\sgr) 
   \rord_\idxc \tgr 
\end{align*}
hence $\ple{\idxa,\rgr} \rord_\Diag \ple{\idxc,\tgr}$, which proves transitivity. 

Given $\ple{\idxa,\rgr}\rord_\Diag\ple{\idxb,\sgr}\rord_\Diag\ple{\idxa,\rgr}$, we know that 
$\idxa\pord\idxb\pord\idxa$ and 
$\Diag_{\idxa,\idxb}(\rgr)\rord_\idxb\sgr$ and $\Diag_{\idxb,\idxa}(\sgr)\rord_\idxa\rgr$. 
Since $\pord$ is antisymmetric, we get $\idxa = \idxb$, 
hence $\Diag_{\idxa,\idxb} = \Diag_{\idxb,\idxa} = \Diag_{\idxa,\idxa}$, which, 
and, by functoriality of $\Diag$, is equal to the identity $\id_{\Diag_\idxa}$. 
This implies that $\rgr\rord_\idxa \sgr$ and $\sgr\rord_\idxa\rgr$, which implies $\rgr = \sgr$ by antisymmetry of $\rord_\idxa$. 
\end{proof}

\begin{proof}[Proof of \cref{prop:sum-mon}] 
We check the four properties. 
\begin{description}
\item [Monotonicity of $\rsum_\Diag$]
Consider $\ple{\idxa,\rgr}\rord_\Diag\ple{\idxb,\sgr}$ and $\ple{\idxa',\rgr'}\rord_\Diag\ple{\idxb',\sgr'}$, then we know that 
$\idxa\pord\idxb$ and $\idxa'\pord\idxb'$ and 
$\Diag_{\idxa,\idxb}(\rgr) \rord_\idxb \sgr$ and $\Diag_{\idxa',\idxb'}(\rgr')\rord_{\idxb'}\sgr'$. 
By monotonicity of $\injl[\idxb,\idxb']$ and $\injr[\idxb,\idxb']$ and $\rsum_{\idxb\isum\idxb'}$, we get 
\[ 
\injl[\idxb,\idxb'](\Diag_{\idxa,\idxb}(\rgr)) \rsum_{\idxb\isum\idxb'} \injr[\idxb,\idxb'](\Diag_{\idxa',\idxb'}(\rgr')) 
\rord_{\idxb\isum\idxb'}
\injl[\idxb,\idxb'](\sgr) \rsum_{\idxb\isum\idxb'} \injr[\idxb,\idxb'](\sgr') 
\]
By monotonicity of $\isum$, we get $\idxa\isum\idxa'\pord\idxb\isum\idxb'$, then, by functoriality of $\Diag$ and by definition of $\injl$ and $injr$, we have 
\[
\injl[\idxb,\idxb'] \circ \Diag_{\idxa,\idxb} = \Diag_{\idxa\isum\idxa',\idxb\isum\idxb'} \circ \injl[\idxa,\idxa']
\qquad 
\injr[\idxb,\idxb'] \circ \Diag_{\idxa',\idxb'} = \Diag_{\idxa\isum\idxa',\idxb\isum\idxb'} \circ \injr[\idxa,\idxa']
\]
therefore, we get 
\[
\begin{array}{l}
\Diag_{\idxa\isum\idxa',\idxb\isum\idxb'}(\injl[\idxa,\idxa'](\rgr)\rsum_{\idxa\isum\idxa'} \injr[\idxa,\idxa'](\rgr'))
\rord_{\idxb\isum\idxb'} \\
\injl[\idxb,\idxb'](\Diag_{\idxa,\idxb}(\rgr)) \rsum_{\idxb\isum\idxb'} \injr[\idxb,\idxb'](\Diag_{\idxa',\idxb'}(\rgr'))
\rord_{\idxb\isum\idxb'}\\
\injl[\idxb,\idxb'](\sgr) \rsum_{\idxb\isum\idxb'} \injr[\idxb,\idxb'](\sgr') 
\end{array}
\]
This proves 
$\ple{\idxa,\rgr}\rsum_\Diag\ple{\idxa',\rgr'} \rord_\Diag \ple{\idxb,\sgr}\rsum_\Diag\ple{\idxb',\sgr'}$, as needed. 
\item [Associativity of $\rsum_\Diag$]
Consider elements \ple{\idxa,\rgr}, \ple{\idxb,\sgr} and \ple{\idxc,\tgr} in $\RC{\GrD\Diag}$. Using \cref{eq:inj-assoc1,eq:inj-assoc2}, we have the following
\begin{align*}
(\ple{\idxa,\rgr}\rsum_\Diag\ple{\idxb,\sgr})\rsum_\Diag\ple{\idxc,\tgr} 
  &= \ple{\idxa\isum\idxb, \injl[\idxa,\idxb](\rgr)\rsum\injr[\idxa,\idxb](\sgr)}\rsum_\Diag\ple{\idxc,\tgr} \\
  &= \ple{(\idxa\isum\idxb)\isum\idxc, \injl[\idxa\isum\idxb,\idxc](\injl[\idxa,\idxb](\rgr)\rsum\injr[\idxa,\idxb](\sgr))\rsum\injr[\idxa\isum\idxb,\idxc](\tgr)} \\ 
  &= \ple{\idxa\isum(\idxb\isum\idxc), \injl[\idxa,\idxb\isum\idxc](\rgr)\rsum(\injr[\idxa,\idxb\isum\idxc](\injl[\idxb,\idxc](\sgr)\rsum\injr[\idxb,\idxc](\tgr))} \\
  &= \ple{\idxa,\rgr}\rsum_\Diag(\ple{\idxb,\sgr}\rsum_\Diag\ple{\idxc,\tgr}) 
\end{align*}
\item [Commutativity of $\rsum_\Diag$]
Consider elements $\ple{\idxa,\rgr}$ and $\ple{\idxb,\sgr}$ in $\RC{\GrD\Diag}$. 
Using \cref{eq:inj-comm}, we get the following 
\begin{align*}
\ple{\idxa,\rgr}\rsum_\Diag\ple{\idxb,\sgr} 
  &= \ple{\idxa\isum\idxb, \injl[\idxa,\idxb](\rgr) \rsum_{\idxa\isum\idxb} \injr[\idxa,\idxb](\sgr)}
   = \ple{\idxa\isum\idxb, \injr[\idxa,\idxb](\sgr)\rsum_{\idxa\isum\idxb} \injl[\idxa,\idxb](\rgr)} \\
  &= \ple{\idxb\isum\idxa, \injl[\idxb,\idxa](\sgr) \rsum_{\idxb\isum\idxa} \injr[\idxb,\idxa](\rgr)} 
   = \ple{\idxb,\sgr} \rsum_\Diag \ple{\idxa,\rgr} 
\end{align*}
\item [Neutrality of $\rzero_\Diag$]
Consider $\ple{\idxa,\rgr}$ in $\RC{\GrD\Diag}$. 
Using \cref{eq:inj-id1,eq:inj-id2,eq:inj-comm}, we the the following
\begin{align*}
\ple{\idxa,\rgr} \rsum_\Diag \rzero_\Diag 
  &= \ple{\idxa\isum\izero, \injl[\idxa,\izero](\rgr)\rsum_{\idxa\isum\izero} \injr[\idxa,\izero](\rzero_\izero)}
   = \ple{\idxa,\rgr\rsum_\idxa \injz[\idxa](\rzero_\izero)} \\ 
  &= \ple{\idxa,\rgr\rsum_\idxa\rzero_\idxa} 
   = \ple{\idxa,\rgr} 
\end{align*}
\end{description}
\end{proof}

\begin{proof}[Proof of \cref{prop:zero-min}] 
Since $\injz[\idxa]$ is a grade algebra homomorphism, we have $\injz[\idxa](\rzero_\izero) = \rzero_\idxa$ and by definition of grade algebra we have $\rzero_\idxa \rord_\idxa \rgr$. 
Therefore, $\injz[\idxa](\rzero_\izero \rord_\idxa\rgr$, which proves the thesis by definition of $\rord_\Diag$. 
\end{proof}

\begin{proof}[Proof of \cref{prop:mul-mon}] 
To prove the equational axioms of monoid (associativity, and neutrality) the proof is the same as \cref{prop:sum-mon} when all the involved elements are different from $\rzero_\izero$, and it is trivial otherwise. 
Indeed, if one of such elements is $\rzero_\izero$, then the whole multiplication gives $\rzero_\izero$ by definition. 

To prove monotonicity of $\rmul_\Diag$, consider $\ple{\idxa,\rgr}\rord_\Diag\ple{\idxb,\sgr}$ and $\ple{\idxa',\rgr'}\rord_\Diag\ple{\idxb',\sgr'}$ in $\RC{\GrD\Diag}$. 
If they are all different from $\rzero_\izero$, the proof goes as in \cref{prop:sum-mon}. 
If either $\ple{\idxa,\rgr} = \rzero_\izero$ or $\ple{\idxa',\rgr'} = \rzero_\izero$, then 
$\ple{\idxa,\rgr}\rmul_\Diag\ple{\idxa',\rgr'} = \rzero_\izero$ and so the thesis follows by \cref{prop:zero-min}. 
If $\ple{\idxb,\sgr} = \rzero_\izero$, then $\ple{\idxa,\rgr}\rord_\Diag\ple{\idxb,\sgr} = \ple{\izero,\rzero_\izero}$ implies 
$\idxa\pord\izero$ and $\Diag_{\idxa,\izero}(\rgr) \rord_\izero \rzero_\izero$  and, since $\izero\pord\idxa$ by definition of grade signature, we get $\idxa = \izero$. 
Therefore, by functoriality of $\Diag$, we have $\Diag_{\idxa,\izero} = \id_{\Diag_\izero}$, hence $\rgr\rord_\izero \rzero_\izero$ which implies $\rgr = \rzero_\izero$, since $\Diag_\izero$ is a grade algebra. 
This proves that $\ple{\idxa,\rgr} = \ple{\izero,\rzero_\izero} = \rzero_\Diag$ and so we get 
$\ple{\idxa,\rgr}\rmul_\Diag\ple{\idxa',\rgr'} = \rzero_\Diag = \ple{\idxb,\sgr}\rmul_\Diag\ple{\idxb',\sgr'}$, as needed. 
Finally, the case $\ple{\idxb',\sgr'} = \ple{\izero,\rzero_\izero}$ is analogous, hence we get the thesis. 
\end{proof}

\begin{proof}[Proof of \cref{thm:multigrade}] 
By \cref{prop:sum-mon,prop:mul-mon} we have both the additive and multiplicative monoid structures. 
\cref{prop:zero-min} proves that $\rzero_\Diag$ is the least element of the order $\rord_\Diag$. 
The fact that multiplying by $\rzero_\izero$ we get again $\rzero_\izero$ holds by definition. 
Hence, it remains to prove that $\rmul_\Diag$ distributes over $\rsum_\Diag$. 
To this end, 
consider $\ple{\idxa,\rgr}$, \ple{\idxb,\sgr} and \ple{\idxc,\tgr} in $\RC{\GrD\Diag}$ and assume they are all different from $\rzero_\izero$. 
Using \cref{eq:inj-assoc1,eq:inj-assoc2,eq:inj-comm,eq:inj-idm}, we get the following  equations: 
\begin{align*}
\injl[\idxa\isum\idxb,\idxa\isum\idxc] \circ \injl[\idxa,\idxb] 
  &= \injl[(\idxa\isum\idxb)\isum\idxa,\idxc] \circ \injl[\idxa\isum\idxb,\idxa] \circ \injl[\idxa,\idxb] \\
  &= \injl[\idxa\isum(\idxa\isum\idxb),\idxc] \circ \injr[\idxa,\idxa\isum\idxb] \circ \injl[\idxa,\idxb] \\
  &= \injl[(\idxa\isum\idxa)\isum\idxb,\idxc] \circ \injl[\idxa\isum\idxa,\idxb] \circ \injl[\idxa,\idxa] \\ 
  &= \injl[\idxa\isum\idxb,\idxc] \circ \injl[\idxa,\idxb] 
   = \injl[\idxa,\idxb\isum\idxc] \\
\injr[\idxa\isum\idxb,\idxa\isum\idxc] \circ \injl[\idxa,\idxc] 
  &= \injl[(\idxa\isum\idxb)\isum\idxa,\idxc] \circ \injr[\idxa\isum\idxb,\idxa] 
   = \injl[\idxa\isum(\idxa\isum\idxb),\idxc] \circ \injl[\idxa,\idxa\isum\idxb] \\
  &= \injl[(\idxa\isum\idxa)\isum\idxb,\idxc] \circ \injl[\idxa\isum\idxa,\idxb] \circ \injl[\idxa,\idxa] \\
  &= \injl[\idxa\isum\idxb,\idxc] \circ \injl[\idxa,\idxb] \\
   &= \injl[\idxa,\idxb\isum\idxc] \\
\injl[\idxa\isum\idxb,\idxa\isum\idxc] \circ \injr[\idxa,\idxb] 
  &= \injl[(\idxa\isum\idxb)\isum\idxa,\idxc] \circ \injl[\idxa\isum\idxb,\idxa] \circ \injr[\idxa,\idxb] \\ 
  &= \injl[\idxa\isum(\idxa\isum\idxb),\idxc] \circ \injr[\idxa,\idxa\isum\idxb] \circ \injr[\idxa,\idxb] 
   = \injl[(\idxa\isum\idxa)\isum\idxb,\idxc] \circ \injr[\idxa\isum\idxa,\idxb] \\
  &= \injr[\idxa,\idxb\isum\idxc]\circ\injl[\idxb,\idxc] \\
\injr[\idxa\isum\idxb,\idxa\isum\idxc]\circ \injr[\idxa,\idxc] 
  &= \injr[(\idxa\isum\idxb)\isum\idxa,\idxc] 
   = \injr[\idxa\isum\idxb,\idxc] \\
  &= \injr[\idxa,\idxb\isum\idxc] \circ \injr[\idxb,\idxc] 
\end{align*}
which imply the following
\begin{align*}
(\ple{\idxa,\rgr} & \rmul_\Diag\ple{\idxb,\sgr}) \rsum_\Diag (\ple{\idxa,\rgr} \rmul_\Diag\ple{\idxc,\tgr}) =  \\ 
  &= \ple{(\idxa\isum\idxb)\isum(\idxa\isum\idxc), \injl[\idxa\isum\idxb,\idxa\isum\idxc](\injl[\idxa,\idxb](\rgr)\rsum\injr[\idxa,\idxb](\sgr)) \rsum \injr[\idxa\isum\idxb,\idxa\isum\idxc](\injl[\idxa,\idxc](\rgr)\rsum\injr[\idxa,\idxc](\tgr))} \\
  &= \ple{\idxa\isum(\idxb\isum\idxc), (\injl[\idxa,\idxb\isum\idxc](\rgr)\rmul\injr[\idxa,\idxb\isum\idxc](\injl[\idxb,\idxc](\sgr))) \rsum (\injl[\idxa,\idxb\isum\idxc](\rgr)\rmul\injr[\idxa,\idxb\isum\idxc](\injr[\idxb,\idxc](\tgr)))} \\
  &= \ple{\idxa\isum(\idxb\isum\idxc), \injl[\idxa,\idxb\isum\idxc](\rgr)\rmul \injr[\idxa,\idxb\isum\idxc](\injl[\idxb,\idxc](\sgr)\rsum\injr[\idxb,\idxc](\tgr)) }  \\ 
  &= \ple{\idxa,\rgr}\rmul_\Diag(\ple{\idxb,\sgr}\rsum_\Diag\ple{\idxc,\tgr}) 
\end{align*}
which proves distributivity when all the elements are different from $\rzero_\Diag$. 

Now, suppose that $\ple{\idxa,\rgr} = \rzero_\Diag$,  then we have 
$\ple{\idxa,\rgr}\rmul_\Diag(\ple{\idxb,\sgr}\rsum_\Diag\ple{\idxc,\tgr}) = \ple{\idxa,\rgr}\rmul_\Diag\ple{\idxb,\sgr} = \ple{\idxa,\rgr}\rmul_\Diag\ple{\idxc,\tgr} = \rzero_\Diag$, 
hence distributivity trivialy holds. 
Finally, suppose $\ple{\idxb,\sgr} = \rzero_\Diag$ (the case $\ple{\idxc,\tgr} = \rzero_\Diag$ is similar), then we have 
$\ple{\idxb,\sgr}\rsum_\Diag\ple{\idxc,\tgr} = \ple{\idxc,\tgr}$ and $\ple{\idxa,\rgr}\rmul_\Diag\ple{\idxb,\sgr} = \rzero_\Diag$. 
Therefore, we get 
\begin{align*}
\ple{\idxa,\rgr} \rmul_\Diag (\ple{\idxb,\sgr}\rsum_\Diag\ple{\idxc,\tgr}) 
  &= \ple{\idxa,\rgr}\rmul_\Diag\ple{\idxc,\tgr} 
   = \rzero_\Diag\rsum_\Diag \ple{\idxa,\rgr}\rmul_\Diag\ple{\idxc,\tgr}  \\ 
  &= \ple{\idxa,\rgr}\rmul_\Diag\ple{\idxb,\sgr}  \rsum_\Diag \ple{\idxa,\rgr}\rmul_\Diag\ple{\idxc,\tgr} 
\end{align*}
\end{proof}

\end{document}